\documentclass[11pt]{article}
\RequirePackage{etex}
\usepackage[margin=1in]{geometry}
\usepackage[hyphens]{url}
\usepackage[colorlinks=true, allcolors=blue, breaklinks=true]{hyperref}
\usepackage{amsfonts,mathrsfs,mathpazo,eucal,xspace,graphicx}
\usepackage[dvipsnames]{xcolor}
\usepackage{endnotes}
\usepackage{amssymb,latexsym}
\usepackage{enumitem}
\setlist{itemsep=0mm}
\usepackage{tikz-network}
\usepackage{tkz-graph}
\usepackage{float}
\usepackage{amsmath, mathtools, physics}
\usepackage{complexity}
\usepackage{qcircuit}
\usepackage[normalem]{ulem}
\usepackage{enumitem}
\usepackage{amssymb}
\usepackage{amsthm}
\usepackage[noend]{algpseudocode}
\usepackage[font=small]{caption}
\usepackage{subcaption}

\usetikzlibrary{arrows,%
                petri,%
                topaths}%
\usepackage{tkz-berge}

\usepackage{varioref}
\usepackage{hyperref}
\usepackage[nameinlink]{cleveref}
\crefname{case}{case}{cases}

\usepackage[linesnumbered,ruled,vlined,noend]{algorithm2e}
\usepackage{verbatim}

\usepackage{amsthm}
\newtheorem{theorem}{Theorem}
\newtheorem{theorem*}{Theorem}

\newtheorem{lemma}[theorem]{Lemma}

\newtheorem{corollary}[theorem]{Corollary}

\newtheorem{definition}[theorem]{Definition}
\theoremstyle{remark}
\newtheorem{remark}[theorem]{Remark}

\mathchardef\period=\mathcode`.
\DeclareMathSymbol{\comma}{\mathord}{letters}{"3B}

\let\originalleft\left
\let\originalright\right
\renewcommand{\left}{\mathopen{}\mathclose\bgroup\originalleft}
\renewcommand{\right}{\aftergroup\egroup\originalright}

\newcommand{\spn}{\operatorname{span}}
\newcommand{\floor}[1]{{\left\lfloor #1 \right\rfloor}}
\newcommand{\ceil}[1]{{\left\lceil #1 \right\rceil}}

\newcommand{\zeros}{\textsc{zeroaddress}}
\newcommand{\invalidaddress}{\textsc{invalidaddress}}
\newcommand{\invalidstring}{\textsc{invalidstring}}
\newcommand{\emptyaddress}{\textsc{emptyaddress}}
\newcommand{\emptystring}{\textsc{emptystring}}
\newcommand{\noedgeaddress}{\textsc{noedgeaddress}}
\newcommand{\noedgestring}{\textsc{noedgestring}}
\newcommand{\specialvertices}{\mathsf{SpecialVertices}}
\newcommand{\specialaddresses}{\mathsf{SpecialAddresses}}
\newcommand{\emptyadd}{\textsc{empty}}
\newcommand{\zero}{\textsc{zero}}

\newcommand\extralabel[4][0mm]{\node[label={[label distance=#1]#2:#3}] at (#4){};}

\newcommand{\Al}{\mathcal{A}}

\newcommand{\pgood}{\mathcal{P}_{\mathrm{good}_{\Al}}}

\newcommand{\pgreat}{\mathcal{P}_{\mathrm{great}_{\Al}}}
\newcommand{\cbad}{c_{*}}

\newcommand{\allbad}{\mathrm{ugly}}
\newcommand{\good}{\mathrm{good}}
\newcommand{\great}{\mathrm{great}}
\newcommand{\bad}{\mathrm{bad}}
\newcommand{\pigoodphi}{\Pi^\phi_{\mathrm{good}}}
\newcommand{\pibadphi}{\Pi^\phi_{\mathrm{bad}}}
\newcommand{\pigoodpsi}{\Pi^\psi_{\mathrm{good}}}
\newcommand{\pibadpsi}{\Pi^\psi_{\mathrm{bad}}}

\newcommand{\entrance}{\textsc{entrance}}
\newcommand{\exit}{\textsc{exit}}
\newcommand{\addresses}{\mathsf{Addresses}}

\newcommand{\pre}{\mathsf{pre}}

\newcommand{\sub}{\mathsf{sub}}
\newcommand{\suc}{\mathsf{suc}}

\let\Pr\relax
\DeclareMathOperator{\Pr}{\mathbb{P}}

\newcommand{\defeq}{\coloneqq}

\newcommand{\genuine}{genuine}
\newcommand{\rooted}{rooted}

\newcommand{\workspace}{\mathrm{workspace}}

\newcommand{\Oracle}{O}

\newcommand{\invalid}{\textsc{invalid}}
\newcommand{\noedge}{\textsc{noedge}}

\SetKwInput{KwInput}{Input}
\SetKwInput{KwOutput}{Output}

\crefname{enumi}{part}{parts}
\Crefname{enumi}{Part}{Parts}

\crefname{algocf}{algorithm}{algorithms}
\Crefname{algocf}{Algorithm}{Algorithms}
\crefname{algocfline}{algorithm}{algorithms} 
\Crefname{algocfline}{Algorithm}{Algorithms}

\newcommand{\wtp}{welded tree problem}
\newcommand{\wtg}{welded tree graph}
\newcommand{\wto}{welded tree oracle}
\newcommand{\valid}{\mathcal{V}_\mathcal{G}}
\newcommand{\BAD}{\mathsf{BAD}}
\newcommand{\GOOD}{\mathsf{GOOD}}

\newcommand{\weld}{\mathsf{WELD}}

\newcommand{\RGB}{\{\mathrm{red}, \mathrm{green}, \mathrm{blue}\}}

\newcommand{\controlled}[1]{\mathop\wedge(#1)}

\title{\bfseries\Large
Quantum algorithms and the power of forgetting
}
\author{Andrew M. Childs, Matthew Coudron, and Amin Shiraz Gilani \smallskip \\
\small Department of Computer Science, University of Maryland \\
\small Institute for Advanced Computer Studies, University of Maryland \\
\small Joint Center for Quantum Information and Computer Science, NIST/University of Maryland}
\date{}

\begin{document}
\sloppy

\maketitle
\begin{abstract}

The so-called \wtp\ provides an example of a black-box problem that can be solved exponentially faster by a quantum walk than by any classical algorithm \cite{ChildsCDFGS03}. Given the name of a special $\entrance$ vertex, a quantum walk can find another distinguished $\exit$ vertex using polynomially many queries, though without finding any particular path from $\entrance$ to $\exit$. It has been an open problem for twenty years whether there is an efficient quantum algorithm for finding such a path, or if the path-finding problem is hard even for quantum computers.
We show that a natural class of efficient quantum algorithms provably cannot find a path from $\entrance$ to $\exit$. Specifically, we consider algorithms that, within each branch of their superposition, always store a set of vertex labels that form a connected subgraph including the $\entrance$, and that only provide these vertex labels as inputs to the oracle. While this does not rule out the possibility of a quantum algorithm that efficiently finds a path, it is unclear how an algorithm could benefit by deviating from this behavior. 
Our no-go result suggests that, for some problems, quantum algorithms must necessarily forget the path they take to reach a solution in order to outperform classical computation. 

\end{abstract}

\tableofcontents
\pagebreak

\section{Introduction}

Quantum algorithms use interference of many branches of a superposition to solve problems faster than is possible classically. Shor's factoring algorithm \cite{Shor97} achieves a superpolynomial speedup over the best known classical algorithms by efficiently finding the period of a modular exponentiation function, and several other quantum algorithms (e.g., \cite{Hallgren,Kedlaya,EHKS14}) provide a speedup by similarly detecting periodic structures. While a few other examples of dramatic quantum speedup are known---notably including the simulation of quantum dynamics \cite{Lloyd96}---our understanding of the capabilities of quantum algorithms remains limited. To gain more insight into the possible applications of quantum computers, we would like to better understand the kinds of problems they can solve efficiently and what features of problems they are able to exploit.

Another example of exponential quantum speedup is based on quantum analogs of random walks. Specifically, quantum walks provide an exponential speedup for the so-called \wtp\ \cite{ChildsCDFGS03}.  The symmetries of this problem, and the structure of the quantum algorithm for solving it, seem fundamentally different from all preceding exponential quantum speedups.  In particular, the \wtp\ provably requires polynomial ``quantum depth'' to solve efficiently \cite{CoudronM19}, whereas all previously known exponential quantum speedups only require logarithmic quantum depth, including Shor's factoring algorithm \cite{CleveWatrous00}.  (The only other known computational problem exhibiting an exponential quantum speedup, yet requiring polynomial quantum depth, was recently constructed in \cite{CCL20}.)

The \wtp\ is defined on a ``welded tree graph'' that is formed by joining the leaves of two binary trees with a cycle that alternates between them, as shown in \Cref{fig:weldedtree}. The root of one tree is designated as the $\entrance$, and the root of the other tree is designated as the $\exit$. The graph structure is provided through an oracle that gives adjacency-list access to the graph, where the vertices are labeled arbitrarily.  Given the label of the $\entrance$ vertex and access to the oracle, the goal of the \wtp\ is to return the label of the $\exit$ vertex. On a quantum computer, this black box allows one to perform a quantum walk, whereby the graph is explored locally in superposition. Interference obtained by following many paths coherently causes the quantum walk to reach the $\exit$ in polynomial time. In contrast, no polynomial-time classical algorithm can efficiently find the $\exit$---essentially because it cannot distinguish the welded tree graph from a large binary tree---so the quantum walk achieves exponential speedup. 

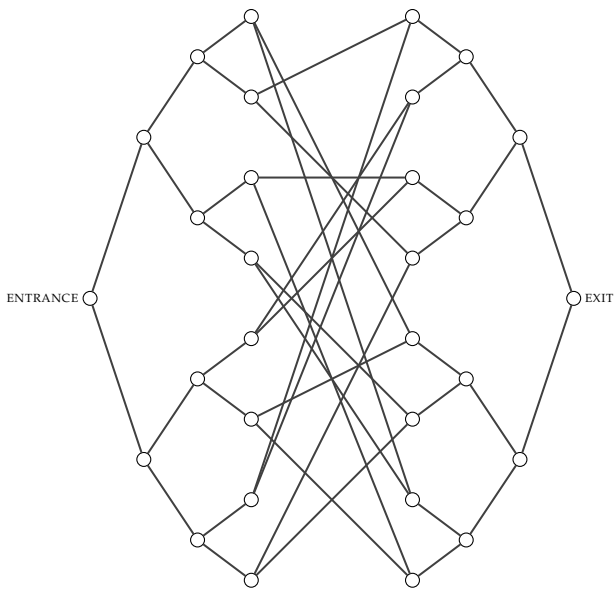
\begin{figure}
    \centering
    \resizebox{0.5\textwidth}{!}{%
    \begin{tikzpicture}[scale=0.35, auto, node distance=0.2cm, every loop/.style={}, thick, fill=black!20, every arrow/.append style={dash,thick}]
        \tikzset{VertexStyle/.style = {shape = ellipse, draw}}
        \Vertex[x=0,y=0,L=$ $]{ent} 
        \extralabel{180}{\entrance}{ent}
        \Vertex[x=4,y=-12,L=$ $]{Lr}
        \Vertex[x=4,y=12,L=$ $]{Lb}
        \Vertex[x=8,y=-18,L=$ $]{Lrg}
        \Vertex[x=8,y=-6,L=$ $]{Lrb}
        \Vertex[x=8,y=6,L=$ $]{Lbg}
        \Vertex[x=8,y=18,L=$ $]{Lbr}
        \Vertex[x=12,y=-21,L=$ $]{Lrgb}
        \Vertex[x=12,y=-15,L=$ $]{Lrgr}
        \Vertex[x=12,y=-9,L=$ $]{Lrbg}
        \Vertex[x=12,y=-3,L=$ $]{Lrbr}
        \Vertex[x=12,y=3,L=$ $]{Lbgr}
        \Vertex[x=12,y=9,L=$ $]{Lbgb}
        \Vertex[x=12,y=15,L=$ $]{Lbrg}
        \Vertex[x=12,y=21,L=$ $]{Lbrb}
    
        \Vertex[x=36,y=0,L=$ $]{exit} 
        \extralabel{0}{\exit}{exit}
        \Vertex[x=32,y=-12,L=$ $]{Rb}
        \Vertex[x=32,y=12,L=$ $]{Rr}
        \Vertex[x=28,y=-18,L=$ $]{Rbr}
        \Vertex[x=28,y=-6,L=$ $]{Rbg}
        \Vertex[x=28,y=6,L=$ $]{Rrb}
        \Vertex[x=28,y=18,L=$ $]{Rrg}
        \Vertex[x=24,y=-21,L=$ $]{Rbrb}
        \Vertex[x=24,y=-15,L=$ $]{Rbrg}
        \Vertex[x=24,y=-9,L=$ $]{Rbgb}
        \Vertex[x=24,y=-3,L=$ $]{Rbgr}
        \Vertex[x=24,y=3,L=$ $]{Rrbr}
        \Vertex[x=24,y=9,L=$ $]{Rrbg}
        \Vertex[x=24,y=15,L=$ $]{Rrgr}
        \Vertex[x=24,y=21,L=$ $]{Rrgb}
    
        \Edge[](ent)(Lr)
        \Edge[](ent)(Lb)
        \Edge[](Lr)(Lrg)
        \Edge[](Lr)(Lrb)
        \Edge[](Lb)(Lbg)
        \Edge[](Lb)(Lbr)
        \Edge[](Lrg)(Lrgb)
        \Edge[](Lrg)(Lrgr)
        \Edge[](Lrb)(Lrbg)
        \Edge[](Lrb)(Lrbr)
        \Edge[](Lbg)(Lbgr)
        \Edge[](Lbg)(Lbgb)
        \Edge[](Lbr)(Lbrg)
        \Edge[](Lbr)(Lbrb)
    
        \Edge[](exit)(Rb)
        \Edge[](exit)(Rr)
        \Edge[](Rb)(Rbg)
        \Edge[](Rb)(Rbr)
        \Edge[](Rr)(Rrg)
        \Edge[](Rr)(Rrb)
        \Edge[](Rbg)(Rbgr)
        \Edge[](Rbg)(Rbgb)
        \Edge[](Rbr)(Rbrg)
        \Edge[](Rbr)(Rbrb)
        \Edge[](Rrg)(Rrgb)
        \Edge[](Rrg)(Rrgr)
        \Edge[](Rrb)(Rrbg)
        \Edge[](Rrb)(Rrbr)
    
        \Edge[](Lbrb)(Rbgr)
        \Edge[](Rbgr)(Lrbg)
        \Edge[](Lrbg)(Rbrb)
        \Edge[](Rbrb)(Lbgb)
        \Edge[](Lbgb)(Rrbg)
        \Edge[](Rrbg)(Lrbr)
        \Edge[](Lrbr)(Rrgr)
        \Edge[](Rrgr)(Lrgr)
        \Edge[](Lrgr)(Rrgb)
        \Edge[](Rrgb)(Lbrg)
        \Edge[](Lbrg)(Rrbr)
        \Edge[](Rrbr)(Lrgb)
        \Edge[](Lrgb)(Rbgb)
        \Edge[](Rbgb)(Lbgr)
        \Edge[](Lbgr)(Rbrg)
        \Edge[](Rbrg)(Lbrb)
    \end{tikzpicture}
    }
    \caption{Example of a welded tree graph with $n=3$.}
    \label{fig:weldedtree}
\end{figure}

While the quantum walk algorithm efficiently finds the $\exit$ by following exponentially many paths in superposition, it does not actually output any of those paths. Classical intuition might suggest that an efficient algorithm for finding the $\exit$ could be used to efficiently find a path by simply recording every intermediate state of the $\exit$-finding algorithm.  However, in general, the intermediate state of a quantum algorithm cannot be recorded without destroying superposition and ruining the algorithm. In other words, the \wtp\ can be viewed as a kind of multi-slit experiment that takes the well-known double-slit experiment into the high-complexity regime. This raises a natural question: Is it possible for some quantum algorithm to efficiently find a path from the $\entrance$ to the $\exit$? This question already arose in the original paper on the \wtp\ \cite{ChildsCDFGS03} and has remained open since, recently being highlighted in a survey of Aaronson \cite{Aaronson21}.

In one attempt to solve this problem, Rosmanis studied a model of ``snake walks'', which allow extended objects to move in superposition through graphs \cite{Rosmanis11}. The state of a snake walk is a superposition of ``snakes'' of adjacent vertices, rather than a superposition of individual vertices as in a standard quantum walk.  While Rosmanis did not show conclusively that snake walks cannot find a path through the \wtg, his analysis suggests that a snake walk algorithm is unlikely to accomplish this using only polynomially many queries to the \wto . 
Although this is only one particular approach, its failure supports the conjecture that it might not be possible to find a path efficiently. If such an impossibility result could be shown for general quantum algorithms, it would establish that, in order to find the solution to some computational problems, a quantum algorithm must necessarily ``forget'' the path it takes to that solution.
While forgetting information is a common feature of quantum algorithms, which often uncompute intermediate results to facilitate interference, many algorithms are able to efficiently produce a classically verifiable certificate for the solution once they have solved the problem.\footnote{For example, in Simon's algorithm \cite{Simon97}, we learn the hidden string and can easily find collisions. In Shor's algorithm \cite{Shor97}, the factors reveal the structure of the input number and their correctness can be easily checked.}
In contrast, hardness of path finding in the \wtp\ would show not only that trying to remember a path would cause one particular algorithm to fail, but in fact \emph{no} algorithm can efficiently collect such information.

In this paper, we take a step toward showing hardness of the welded tree path-finding problem. Specifically, we show hardness under two natural assumptions that we formalize in \Cref{sec:rooted}, namely that the algorithm is \emph{\genuine} and \emph{\rooted}.

First, we assume that the algorithm accesses the oracle for the input graph in a way that we call \emph{genuine}. A genuine algorithm is essentially one that only provides meaningful vertex labels as inputs to the \wto . Both the ordinary quantum walk \cite{ChildsCDFGS03} and the snake walk \cite{Rosmanis11} can be implemented by genuine algorithms. It is hard to imagine how non-genuine algorithms could gain an advantage over genuine algorithms, but we leave further exploration of this topic for future work.

We also assume that the algorithm is \emph{rooted}. Informally, a rooted algorithm is one that always maintains (i.e., remembers) a path from the $\entrance$ to every vertex appearing in its state. Note that the $\exit$-finding algorithm of \cite{ChildsCDFGS03} is crucially \emph{not} rooted.  Nonetheless, it is natural to focus on rooted algorithms when considering the problem of path-finding, since a non-rooted path-finding algorithm would effectively have to detach from the $\entrance$ and later find it again.  While we cannot rule out this possibility, it seems implausible. Although remembering a path to the $\entrance$ limits how quantum interference can occur, it does not eliminate interference entirely---in fact, rooted algorithms can exhibit exponential constructive and destructive interference. Furthermore, if the snake walk \cite{Rosmanis11} were to find a path from $\entrance$ to $\exit$, the most natural way of doing so would be in a rooted fashion. 

Our main result is that a genuine, rooted quantum algorithm cannot find a path from $\entrance$ to $\exit$ in the \wtg\ using only polynomially many queries. To establish this, we show that for any genuine, rooted quantum query algorithm $\Al$, there is a classical query algorithm using at most polynomially more queries that can approximately sample from the output state of $\Al$ (measured in the computational basis) up to a certain error term $\ket{\psi_{\mathrm{ugly}}}$ (defined in \Cref{subsec:goodbadandugly}).  This error term can be intuitively described as the part of the superposition of rooted configurations that has ever encountered a cycle or the $\exit$ in the \wtg\ during the entire course of the algorithm. Because elements of a quantum superposition need not have a well-defined classical history, the precise definition of our error term is somewhat involved. Nonetheless, we are able to bound the sampling error of our efficient classical simulation of $\Al$ using an inductive argument.

We construct the classical simulation as follows. First, using exponential time and only a constant number of classical queries, the algorithm processes the circuit diagram of the genuine, rooted quantum algorithm and samples a ``transcript'' (defined in \Cref{sec:transcript}) that describes a computational path the quantum algorithm could have taken, neglecting the possibility of encountering a cycle or the $\exit$. The classical algorithm then makes polynomially many queries to the classical oracle, in a manner prescribed by the sampled transcript, and outputs the vertices of the \wtg\ that were reached by those queries. We prove in \Cref{sec:simulation} that this efficient classical query algorithm is almost as likely to find an $\entrance$--$\exit$ path as the original genuine, rooted quantum algorithm.  We do this by showing that the classical algorithm exactly simulates the part of the quantum state that does not encounter a cycle or the $\exit$, and that the remaining error term, $\ket{\psi_{\mathrm{ugly}}}$, is exponentially small.  A major technical challenge is that our classical hardness result (shown in \Cref{sec:3coloring}) does not immediately show that $\ket{\psi_{\mathrm{ugly}}}$ is small as it may be possible for a quantum algorithm to foil a classical simulation by computing and uncomputing a cycle, and ``pretending'' to have never computed it. We overcome this issue by considering the portions of the state that encounter a cycle at each step and inductively bounding their total mass.

A subtle---yet unexpectedly significant---detail in our analysis is that we consider a version of the \wtp\ in which the oracle provides a 3-coloring of the edges of the graph, instead of using a 9-coloring as in the lower bound of \cite{ChildsCDFGS03}.\footnote{Note that the quantum walk algorithm can solve the \wtp\ using a 3-coloring, or even if it is not provided with a coloring at all \cite{ChildsCDFGS03}.} 
This alternative coloring scheme substantially reduces the complexity of the analysis in  \Cref{sec:simulation,sec:transcript,sec:rooted}.
This is because it allows us to determine, with high probability, whether starting at the $\entrance$ and following the edges prescribed by a polynomial-length color sequence $t$ will lead to a valid vertex of the \wtg, using only a constant number of classical queries to the \wto. In particular, it suffices to check whether $t$ departs from the $\entrance$ along one of the two valid edges (which can be determined using only three queries to the oracle). This is a key property used in our argument that the transcript state (see \Cref{def:transcript}) can track much of the behavior of a \genuine, \rooted\ quantum algorithm while only making a small number of classical queries to the \wto. 

However, our choice of the 3-coloring model comes at the cost of having to redesign the proof of classical hardness of finding the $\exit$ vertex in the \wtg. The original classical hardness proof \cite{ChildsCDFGS03} crucially considers a special type of 9-coloring with the property that, starting from a valid coloring, the color of any edge can be altered arbitrarily, and only edges within distance 2 need to be re-colored to produce a valid coloring with that newly assigned edge color.  This ``local re-colorability'' property is used at the crux of the classical hardness result, first in reducing from Game 2 to Game 3, and again implicitly in part (i) of the proof of Lemma 8 \cite{ChildsCDFGS03}. In contrast, a valid 3-coloring of the \wtg\ does not have this ``local re-colorability'' property: changing a single edge color might require a global change of many other edge colors to re-establish validity of the coloring. Thus we are forced to develop a modification of the classical hardness proof, given in \Cref{sec:3coloring}, which may also be of independent interest.

Note that our hardness result for $\exit$-finding in the 3-color model implies the hardness result for $\exit$-finding in the colorless model of \cite{ChildsCDFGS03} (which is equivalent to a restricted class of locally-constructible 9-colorings), but not the other way around.  This is because, a priori, the given 3-coloring could leak global information about the graph that the 9-coloring does not.  On the other hand, our hardness result for $\exit$-finding in the 3-color model combined with the analysis of \Cref{sec:rooted} implies hardness of path finding for genuine, rooted algorithms in both the 3-color and 9-color models, as well as the colorless model.

While our result does not definitively rule out the possibility of an efficient quantum algorithm for finding a path from $\entrance$ to $\exit$ in the \wtg, it constrains the form that such an algorithm could take. In particular, it shows that the most natural application of a snake walk to the \wtp, in which the snake always remains connected to the $\entrance$, cannot solve the problem. While it is conceivable that a snake could detach from the $\entrance$ and later expand to connect the $\entrance$ and $\exit$, this seems unlikely. More generally, non-genuine and non-rooted behavior do not intuitively seem useful for solving the problem. We hope that future work will be able to make aspects of this intuition rigorous.

\paragraph{Open questions} This work leaves several natural open questions. The most immediate is to remove the assumption of a rooted, genuine algorithm to show unconditional hardness of finding a path (or, alternatively, to give an efficient path-finding algorithm by exploiting non-genuine or non-rooted behavior). We also think it should be possible to show classical hardness of the general \wtp\ when the oracle provides a 3-coloring. Finally, it would be instructive to find a way of instantiating the welded tree problem in an explicit (non-black box) fashion, giving a quantum speedup in a non-oracular setting.

\section{Genuine and rooted algorithms}\label{sec:rooted}

In this section, we precisely define the aforementioned notion of genuine, rooted quantum query algorithms. Intuitively, an algorithm is genuine if it only allows for ``meaningful'' processing of vertex labels, and it is rooted if it remains ``connected to the \entrance" throughout its course. We begin by describing our setup and recalling the definition of the \wto .

\begin{definition} [Welded tree] \label{def:weldedtree}
A graph $\mathcal{G}_n$ is a \emph{welded tree} of size $n$ if it is formed by joining the $2 \cdot 2^n$ leaves of two balanced binary trees of height $n$ with a cycle that alternates between the two sets of leaves (as shown in \Cref{fig:weldedtree}). Each vertex in $\mathcal{G}_n$ is labeled by a $2n$-bit string. 
\end{definition}

Henceforth, we refer to the input \wtg\ of size $n$ as $\mathcal{G}$. We use $\valid$ to denote the set of vertices of $\mathcal{G}$. Since $\mathcal{G}$ is bipartite and each vertex $v \in \valid$ has degree at most $3$, $\mathcal{G}$ can be edge-colored using only $3$ colors \cite{Kon16}. Therefore, we suppose that the edges of $\mathcal{G}$ are colored from the set $\mathcal{C} \defeq \RGB$. We define a classical oracle function $\eta_c\colon \{0,1\}^{2n} \to \{0,1\}^{2n}$ that encodes the edges of color $c \in \mathcal{C}$ in $\mathcal{G}$. \Cref{fig:weldedtreecoloredlabeled} shows a valid coloring and labeling of the \wtg\ from \Cref{fig:weldedtree}.

\begin{figure}
    \centering
    \resizebox{0.8\textwidth}{!}{%
    \begin{tikzpicture}[scale=0.55,auto, node distance=0.2cm, every loop/.style={}, thick, fill=black!20, every arrow/.append style={dash,thick}]
    \tikzset{VertexStyle/.style = {shape = ellipse, draw}}
    \Vertex[x=0,y=0,L=$\entrance$]{ent} 
    \Vertex[x=4,y=-12,L=$010110$]{Lr}
    \Vertex[x=4,y=12,L=$101000$]{Lb}
    \Vertex[x=8,y=-18,L=$011101$]{Lrg}
    \Vertex[x=8,y=-6,L=$101001$]{Lrb}
    \Vertex[x=8,y=6,L=$110100$]{Lbg}
    \Vertex[x=8,y=18,L=$101010$]{Lbr}
    \Vertex[x=12,y=-21,L=$001100$]{Lrgb}
    \Vertex[x=12,y=-15,L=$110011$]{Lrgr}
    \Vertex[x=12,y=-9,L=$101111$]{Lrbg}
    \Vertex[x=12,y=-3,L=$111001$]{Lrbr}
    \Vertex[x=12,y=3,L=$011000$]{Lbgr}
    \Vertex[x=12,y=9,L=$010100$]{Lbgb}
    \Vertex[x=12,y=15,L=$101110$]{Lbrg}
    \Vertex[x=12,y=21,L=$000001$]{Lbrb}
    
    \Vertex[x=36,y=0,L=$\exit$]{exit}
    \Vertex[x=32,y=-12,L=$101101$]{Rb}
    \Vertex[x=32,y=12,L=$001010$]{Rr}
    \Vertex[x=28,y=-18,L=$101100$]{Rbr}
    \Vertex[x=28,y=-6,L=$001011$]{Rbg}
    \Vertex[x=28,y=6,L=$011110$]{Rrb}
    \Vertex[x=28,y=18,L=$010101$]{Rrg}
    \Vertex[x=24,y=-21,L=$100100$]{Rbrb}
    \Vertex[x=24,y=-15,L=$001001$]{Rbrg}
    \Vertex[x=24,y=-9,L=$010001$]{Rbgb}
    \Vertex[x=24,y=-3,L=$000101$]{Rbgr}
    \Vertex[x=24,y=3,L=$000100$]{Rrbr}
    \Vertex[x=24,y=9,L=$110101$]{Rrbg}
    \Vertex[x=24,y=15,L=$001110$]{Rrgr}
    \Vertex[x=24,y=21,L=$110110$]{Rrgb}
    
    \Edge[color=red](ent)(Lr)
    \Edge[color=blue](ent)(Lb)
    \Edge[color=ForestGreen](Lr)(Lrg)
    \Edge[color=blue](Lr)(Lrb)
    \Edge[color=ForestGreen](Lb)(Lbg)
    \Edge[color=red](Lb)(Lbr)
    \Edge[color=blue](Lrg)(Lrgb)
    \Edge[color=red](Lrg)(Lrgr)
    \Edge[color=ForestGreen](Lrb)(Lrbg)
    \Edge[color=red](Lrb)(Lrbr)
    \Edge[color=red](Lbg)(Lbgr)
    \Edge[color=blue](Lbg)(Lbgb)
    \Edge[color=ForestGreen](Lbr)(Lbrg)
    \Edge[color=blue](Lbr)(Lbrb)
    
    \Edge[color=blue](exit)(Rb)
    \Edge[color=red](exit)(Rr)
    \Edge[color=ForestGreen](Rb)(Rbg)
    \Edge[color=red](Rb)(Rbr)
    \Edge[color=ForestGreen](Rr)(Rrg)
    \Edge[color=blue](Rr)(Rrb)
    \Edge[color=red](Rbg)(Rbgr)
    \Edge[color=blue](Rbg)(Rbgb)
    \Edge[color=ForestGreen](Rbr)(Rbrg)
    \Edge[color=blue](Rbr)(Rbrb)
    \Edge[color=blue](Rrg)(Rrgb)
    \Edge[color=red](Rrg)(Rrgr)
    \Edge[color=ForestGreen](Rrb)(Rrbg)
    \Edge[color=red](Rrb)(Rrbr)
    
    \Edge[color=ForestGreen](Lbrb)(Rbgr)
    \Edge[color=blue](Rbgr)(Lrbg)
    \Edge[color=red](Lrbg)(Rbrb)
    \Edge[color=ForestGreen](Rbrb)(Lbgb)
    \Edge[color=red](Lbgb)(Rrbg)
    \Edge[color=blue](Rrbg)(Lrbr)
    \Edge[color=ForestGreen](Lrbr)(Rrgr)
    \Edge[color=blue](Rrgr)(Lrgr)
    \Edge[color=ForestGreen](Lrgr)(Rrgb)
    \Edge[color=red](Rrgb)(Lbrg)
    \Edge[color=blue](Lbrg)(Rrbr)
    \Edge[color=ForestGreen](Rrbr)(Lrgb)
    \Edge[color=red](Lrgb)(Rbgb)
    \Edge[color=ForestGreen](Rbgb)(Lbgr)
    \Edge[color=blue](Lbgr)(Rbrg)
    \Edge[color=red](Rbrg)(Lbrb)
    \end{tikzpicture}
    }
    \caption{Example of a 3-colored labeled welded tree graph for $n=3$.}
    \label{fig:weldedtreecoloredlabeled}
\end{figure}
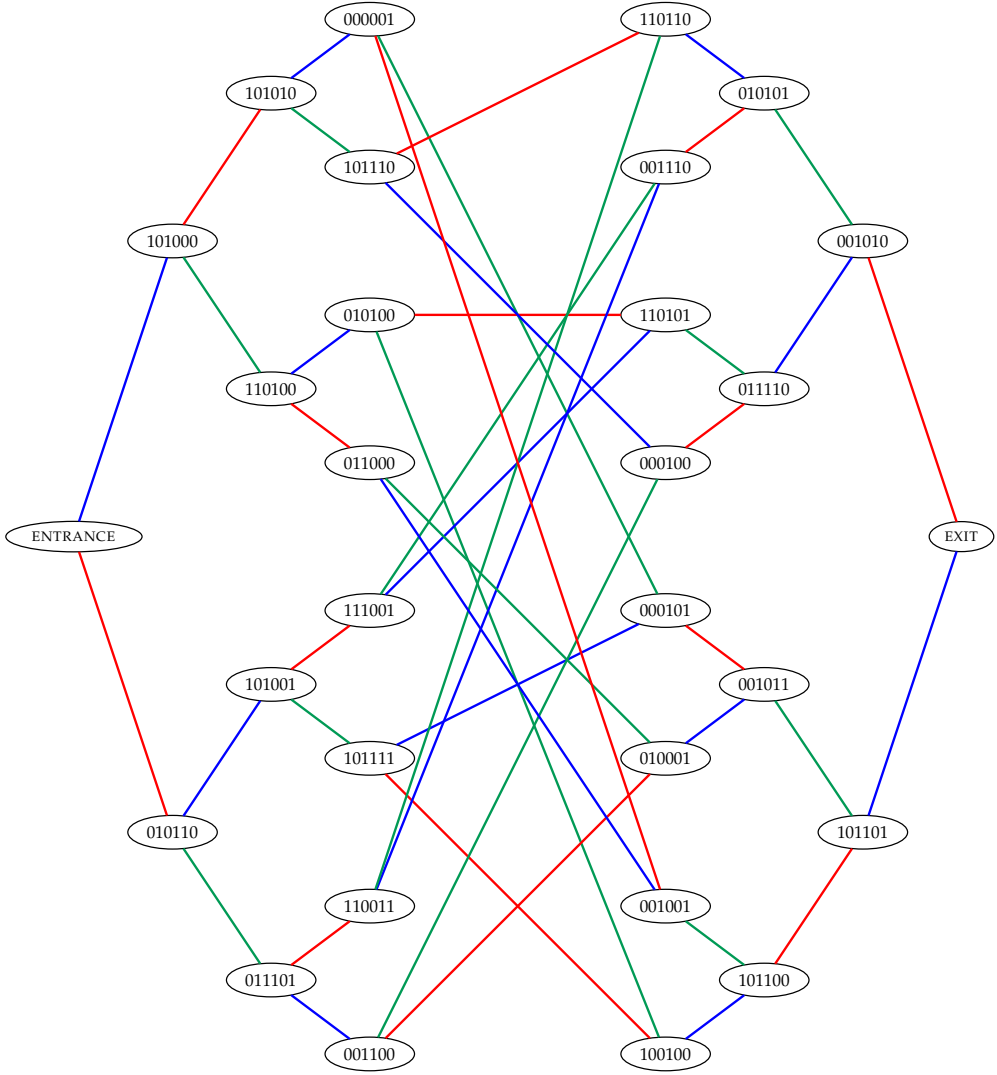

\begin{definition}[$\eta_c$] \label{def:eta_c}
For any $v \in \valid$ and $c \in \mathcal{C}$, let $I_c(v)$ be the indicator variable that is $1$ if the vertex labeled $v$ has an edge colored $c$ and $0$ otherwise. If $I_c(v) =  1$ for some $v \in \valid$ and $c \in \mathcal{C}$, let $N_c(v)$ be the label of the vertex joined to $v$ with an edge of color $c$. Then
\begin{align}
    \eta_c(v) \defeq
    \begin{cases}
        N_c(v) & v \in \valid \text{ and } I_c(v) = 1 \\
        \noedge & v \in \valid \text{ and } I_c(v) = 0 \\
        \invalid & v \notin \valid,
    \end{cases}
\end{align}
where $\noedge$ and $\invalid$ are special reserved strings in $\{0,1\}^{2n} \setminus \valid$. We also define  $\specialvertices \defeq \{0^{2n}, \allowbreak \entrance, \allowbreak \exit, \noedge, \allowbreak \invalid\}$.
\end{definition}

Since $\mathcal{G}$ is 3-colored, for any vertex label $v \in \valid$, $I_c(v) = 0$ only if $v \in \{\entrance, \exit\}$. We now describe the spaces that our algorithms act on.

\begin{definition}[Vertex register and vertex space] \label{def:vertexspace}
A \emph{vertex register} is a $2n$-qubit register that stores a vertex label. We consider quantum states that have exactly $p(n)$ vertex registers, and refer to the $2np(n)$-qubit space consisting of all the vertex registers as the \emph{vertex space}. 
\end{definition}

Any computational basis state in the vertex space stores $p(n)$ vertex labels, corresponding to a subgraph of $\mathcal{G}$. A quantum algorithm can also store additional information using its workspace.

\begin{definition}[Workspace and workspace register] \label{def:workspace}
A \emph{workspace register} is a single-qubit register that can store arbitrary ancillary states. We allow for arbitrarily many workspace registers, and refer to the space consisting of all workspace registers as the \emph{workspace}. 
\end{definition}

\subsection{Genuine algorithms} \label{subsec:genuine}

We now precisely describe the set of gates that we allow quantum query algorithms to employ for querying and manipulating the vertex labels in a meaningful way.

\begin{definition}[Genuine circuit] \label{def:genuinecircuit}
We say that a quantum circuit $C$ is \emph{\genuine} if it is built from the following unitary gates.
\begin{enumerate}
\item \label{itm:genuineoracle} Controlled-oracle query gates $O_c$ for $c \in \mathcal{C}$ where the control qubit is in the workspace, and $O_c$ acts on the $j$th and $k$th vertex registers for some distinct $j, k \in [p(n)] \defeq \{1,\ldots,p(n)\}$ as 
\begin{equation} \label{eq:genuineoracle}
    O_c \colon \ket{v_j} \ket{v_k} \mapsto \ket{v_j} \ket{v_k \oplus \eta_c(v_j)}
\end{equation}
where $\eta_c$ is specified in \Cref{def:eta_c}.

Furthermore, in a \genuine\ circuit, $O_c$ can only be applied if $v_k = 0^{2n}$ or $v_k = \eta_c(v_j)$ for every $v_j, v_k$ pair appearing in those respective registers in the superposition.

We let $\controlled{A}$ denote a controlled-$A$ gate, so that $\controlled{O_c}$ denotes the controlled-$O_c$ gate.

\item Controlled-$e^{i\theta T}$ rotations for any $\theta \in [0,2\pi)$ where the control qubit is in the workspace and the Hamiltonian $T$ is defined, similarly to \cite{ChildsCDFGS03}, to act on the $j$th and $k$th vertex registers for some distinct $j, k \in [p(n)]$ as
\begin{equation}
    T\colon\ket{v_j}\ket{v_k} \mapsto \ket{v_k}\ket{v_j}.
\end{equation}
As per \cref{itm:genuineoracle}, $\controlled{e^{i\theta T}}$ denotes the controlled-$e^{i\theta T}$ gate.

\item Equality check gates $\mathcal{E}$, which act on the $j$th and $k$th vertex registers for some distinct $j, k \in [p(n)]$, and on the $a$th workspace register for some workspace index $a$, as
\begin{equation}
    \mathcal{E}\colon \ket{v_j}\ket{v_k}\ket{w_a} \mapsto \ket{v_j}\ket{v_k}\ket{w_a \oplus \delta[v_j = v_k]}
\end{equation}
where $\delta[P]$ is $1$ if $P$ is true and $0$ if $P$ is false.
	
\item $\noedge$ check gates $\mathcal{N}$, which act on the $j$th vertex register for some $j \in [p(n)]$, and on the $a^{th}$ workspace register for some workspace index $a$, as
\begin{equation}
    \mathcal{N}\colon \ket{v_j}\ket{w_a} \mapsto \ket{v_j}\ket{w_a \oplus \delta[v_j=\noedge]}.
\end{equation}

\item $\zero$ check gates $\mathcal{Z}$, which act on the $j$th vertex register for some $j \in [p(n)]$, and on the $a^{th}$ workspace register for some workspace index $a$, as
\begin{equation}
    \mathcal{Z}\colon \ket{v_j}\ket{w_a} \mapsto \ket{v_j}\ket{w_a \oplus \delta[v_j=0^{2n}]}.
\end{equation}

\item Arbitrary two-qubit gates (or, equivalently, arbitrary unitary transformations) restricted to the workspace register.
\end{enumerate}
\end{definition}

We now define the notion of genuine algorithms using \Cref{def:genuinecircuit}. Let $O = \{O_c: c \in \mathcal{C}\}$ denote a particular randomly selected \wto , and let $\mathcal{A}(O)$ denote a quantum algorithm that makes quantum queries to $O$.

\begin{definition}[Genuine algorithm] \label{def:genuine}
We call a quantum query algorithm $\mathcal{A}$ \emph{genuine} if, for the given \wto\ $O$, $\mathcal{A}(O)$ acts on the vertex space and the workspace as follows.
\begin{enumerate}
    \item $\mathcal{A}(O)$ begins with an initial state 
    \begin{align} \label{eq:genuineinitial}
	    \ket{\psi_{\mathrm{initial}}} \defeq \ket{\entrance} \otimes \left(\ket{0^{2n}} \right)^{\otimes(p(n)-1)} \otimes \ket{0}_{\workspace}.
	\end{align}
	\item Then, it applies a $p(n)$-gate genuine circuit $C$ (as in \Cref{def:genuinecircuit}) on $\ket{\psi_{\mathrm{initial}}}$ to get the state $\ket{\psi_\Al}$.
	\item Finally, it measures all the vertex registers of $\ket{\psi_\Al}$ in the computational basis and outputs the corresponding vertex labels.
\end{enumerate}
\end{definition}

We focus on genuine algorithms because they are easier to analyze than fully general algorithms, but do not seem to eliminate features that would be useful in a path-finding algorithm. Genuine algorithms are only restricted in the sense that they cannot use vertex labels other than by storing them, acting on them with the input or output register of an oracle gate, performing phased swaps of the vertex label positions, and checking whether they are equal to zero or $\noedge$. Since the vertex labels are arbitrary and uncorrelated with the structure of the welded tree, it is hard to imagine how a general quantum algorithm could gain an advantage over genuine quantum algorithms by using the vertex labels in any other way.\footnote{The proof that classical algorithms cannot efficiently find the $\exit$ effectively shows that classical algorithms cannot benefit from non-genuine behavior \cite{ChildsCDFGS03}. While it seems harder to make this rigorous for quantum algorithms, similar intuition holds.} Thus, genuine algorithms describe a natural class of strategies that should offer insight into the more general case.

We also emphasize that the only proposed algorithms for the \wtp\ (and the associated path-finding problem) are genuine. The only such algorithms we are aware of are the exit-finding algorithm of \cite{ChildsCDFGS03} and the snake-walk algorithm analyzed by Rosmanis \cite{Rosmanis11}.

Intuitively, the $\exit$-finding algorithm of \cite{ChildsCDFGS03} is genuine since it performs a quantum walk on the \wtg, and such a process does not depend on the vertex labels.
More concretely, a close inspection of the $\exit$-finding algorithm of \cite{ChildsCDFGS03} reveals that every gate in the algorithm is an allowed gate in \Cref{def:genuinecircuit} (even with the above minor modification).  This means that the algorithm is genuine as per \Cref{def:genuine}.  (As a technical aside, note that the algorithm works with any valid coloring of the \wtg, so in particular, it works for our chosen 3-color model by simply limiting the set of colors in the algorithm.)

Similarly, in \cite{Rosmanis11}, Rosmanis defines a quantum snake walk algorithm on a particular \wtg\ $G$ to be a quantum walk on a corresponding ``snake graph'' $G_\ell$, which has one vertex for each distinct ``snake'' of length $\ell$ in $G$.  Here a ``snake'' of length $\ell$ refers to a length-$\ell$ vector of consecutive vertices of $G$.  Although it is more complicated to decompose this algorithm into the gates of \Cref{def:genuinecircuit}, this can be done, showing that the snake walk algorithm is genuine. Furthermore, it is intuitive that the snake walk algorithm should not depend on the specific vertex labels simply because it is defined to be a quantum walk on $G_n$, a graph whose connectivity does not depend on the vertex labels of $G$.

\subsection{Rooted algorithms} \label{subsec:rooted}

We now define the notion of a rooted algorithm. 
Intuitively, a state in the vertex space is rooted if it corresponds to a set of labels of vertices from $\valid$ (and the $\noedge$ and $0^{2n}$ labels) that form a connected subgraph containing the $\entrance$ (neglecting the $\noedge$ and $0^{2n}$ labels, if present).

\begin{definition}[Rooted state] \label{def:rootedstate}
We say that a computational basis state $\ket{\psi}$ in the vertex space is \emph{rooted} if $\entrance$ is stored in some vertex register of $\ket{\psi}$ and, for any vertex label $v$ stored in any of the vertex registers of $\ket{\psi}$,
\begin{enumerate}
    \item $v \in \valid \cup \{0^{2n}, \noedge\}$, and
    \item if $v \neq 0^{2n}$, then there exist $r$ vertex registers storing vertex labels $v_{j_1}, \ldots, v_{j_r}$ such that $v_{j_1} = \entrance$, $v_{j_r} = v$, and for each $k \in [r-1]$, $\eta_c(v_{j_k}) = v_{j_{k+1}}$ for some $c \in \mathcal{C}$. 
\end{enumerate}
\end{definition}

\begin{figure}
    \tiny
    \centering
    \begin{subfigure}[b]{0.44\linewidth}
    \begin{tikzpicture}[scale=0.27, auto, node distance=0.2cm, every loop/.style={}, thick, every arrow/.append style={dash,thick}]
    \tikzset{VertexStyle/.style = {shape = ellipse, minimum size = 20pt, draw}}
    \Vertex[x=0,y=30,L=$\entrance$]{emp}
    \Vertex[x=-8,y=26,L=$010110$]{r}
    \Vertex[x=8,y=26,L=$101000$]{b}
    \Vertex[x=-4,y=22,L=$101001$]{rb}
    \Vertex[x=4,y=22,L=$110100$]{bg}
    \Vertex[x=12,y=22,L=$101010$]{br}
    \Vertex[x=-8,y=18,L=$101111$]{rbg}
    \tikzset{every node/.style={opacity=0,text opacity=1,scale=0.75}}
    \tikzset{EdgeStyle/.style=auto, color=red}
    \Edge[color=red](emp)(r)
    \Edge[color=blue](emp)(b)
    \Edge[color=blue](r)(rb)
    \Edge[color=ForestGreen](b)(bg)
    \Edge[color=red](b)(br)
    \Edge[color=ForestGreen](rb)(rbg)
    \end{tikzpicture}
    \caption{Subgraph of \Cref{fig:weldedtreecoloredlabeled} corresponding to a state containing the vertex labels $\{\entrance,\allowbreak 010110, 101000, 101001, 101010, 101111, 110100\}$.} 
    \label{fig:rooted} 
    \end{subfigure}
    \hspace{3em}
    \begin{subfigure}[b]{0.44\linewidth}
    \begin{tikzpicture}[scale=0.27, auto, node distance=0.2cm, every loop/.style={}, thick, every arrow/.append style={dash,thick}]
    \tikzset{VertexStyle/.style = {shape = ellipse, minimum size = 20pt, draw}}
    \Vertex[x=0,y=30,L=$\entrance$]{emp}
    \Vertex[x=8,y=26,L=$101000$]{b}
    \Vertex[x=-4,y=22,L=$101001$]{rb}
    \Vertex[x=4,y=22,L=$110100$]{bg}
    \Vertex[x=12,y=22,L=$101010$]{br}
    \Vertex[x=-8,y=18,L=$101111$]{rbg}
    \tikzset{every node/.style={opacity=0,text opacity=1,scale=0.75}}
    \tikzset{EdgeStyle/.style=auto, color=red}
    \Edge[color=blue](emp)(b)
    \Edge[color=ForestGreen](b)(bg)
    \Edge[color=red](b)(br)
    \Edge[color=ForestGreen](rb)(rbg)
    \end{tikzpicture}
    \caption{Subgraph of \Cref{fig:weldedtreecoloredlabeled} corresponding to a state containing the vertex labels $\{\entrance,\allowbreak101000, 101001, 101010, 101111, 110100\}$. 
    }
    \label{fig:nonrooted} 
    \end{subfigure}
    \caption{Examples of rooted and non-rooted states.}
    \label{fig:rootedexamples}
\end{figure}
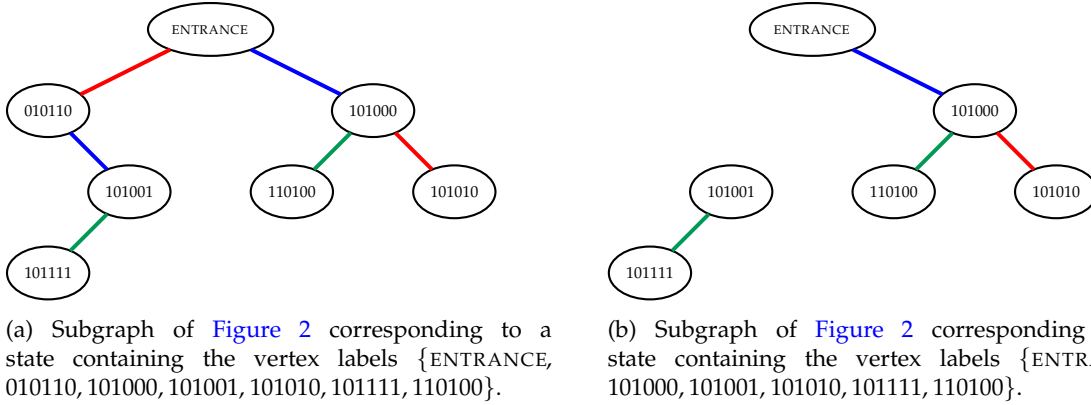

\Cref{fig:rootedexamples} shows examples of rooted and non-rooted states.
We say that an algorithm is rooted if all its intermediate states are superpositions of rooted states.

\begin{definition}[Rooted algorithm] \label{def:rooted}
A quantum query algorithm $\Al$ is \emph{\rooted} if, for the given \wto\ $O$, at each intermediate step of $\mathcal{A}(O)$, every computational basis state in the support of the vertex space of the quantum state maintained by $\Al$ is rooted.
\end{definition}

Non-rooted behavior can be useful for exploring the \wtg. In particular, the algorithm of \cite{ChildsCDFGS03} for finding the $\exit$ is not rooted, since it only maintains a single vertex (in superposition). However, a path-finding algorithm must store information about many vertices, and the value of detaching from the $\entrance$ is unclear since the algorithm must ultimately reattach. Note that the snake walk \cite{Rosmanis11} is initially rooted, the most natural way for it to find a path from $\entrance$ to $\exit$ is arguably to do so while remaining rooted, though the algorithm may become non-rooted if it is run for long enough.

\section{Transcript states}
\label{sec:transcript}

For any genuine quantum query algorithm $\mathcal{A}$ that makes $p(n)$ oracle queries to the oracle $\Oracle$ of the input welded tree $\mathcal{G}$, we associate a quantum state $\ket{\phi_{\mathcal{A}}}$, which we call the \textit{transcript state} of $\mathcal{A}(\Oracle)$. As we will see in \Cref{def:transcript} below, instead of storing the label of a vertex $v$, the transcript state $\ket{\phi_{\mathcal{A}}}$ stores a path from the $\entrance$ to $v$. We refer to this path as the \emph{address} of $v$, which we now define.

\begin{definition}[Vertex addresses] \label{def:addresses}
We say that a tuple $t$ of colors from $\mathcal{C}$ is an \emph{address} of a vertex $v$ of $\mathcal{G}$ if $v$ is reached by starting at the $\entrance$ and following the edge colors listed in $t$. For completeness, we assign special names $\zeros$, $\noedgeaddress$, and $\invalidaddress$ to denote the addresses of vertex labels $0^{2n}$, $\noedge$, and $\invalid$, respectively. We denote the empty tuple by the special name $\emptyaddress$. 
Let $\mathsf{SpecialAddresses} \defeq \{\zeros, \emptyaddress, \allowbreak \noedgeaddress, \allowbreak \invalidaddress\}$. 
We define 
\begin{equation}
\addresses \defeq \specialaddresses \cup \bigcup_{i \in [p(n)]} \mathcal{C}^{i}
\end{equation}
where $\mathcal{C}^{i}$ denotes the set of all $i$-tuples of colors from $\mathcal{C}$.
\end{definition}

Note that a given vertex can have many different associated addresses. Indeed, two addresses that differ by an even-length palindrome of colors are associated to the same vertex. Even greater multiplicity of addresses can occur because of the cycles in $\mathcal{G}$. We define the notion of the \textit{address tree} to deal with the former issue, and we delay consideration of the latter issue. To define the address tree, we need to know the color $\cbad$ that does not appear at the $\entrance$.

\begin{definition}[The missing color at the entrance]
Let $\cbad \in \mathcal{C}$ be the unique color such that there is no edge of color $\cbad$ incident to the $\entrance$ in $\mathcal{G}$. 
\end{definition}

\begin{definition}[Address tree] \label{def:addresstree}
The \emph{address tree} $\mathcal{T}$ (see \Cref{fig:addresstree}) is a binary tree of depth $p(n)$ with 3 additional vertices.\footnote{The address tree is not technically a tree, but we use this name since it contains no non-trivial cycles.} Its vertices and edges are labeled by addresses and colors, respectively, as follows. The 3 additional vertices are labeled by each address in $\mathsf{SpecialAddresses} \setminus \{\emptyaddress\}$. The root of $\mathcal{T}$ is labeled by $\emptyaddress$. It is joined to the vertex labeled $\noedgeaddress$ by a directed edge of color $\cbad$, and to 2 other vertices, each by an undirected edge of a distinct color from $\mathcal{C} \setminus \{\cbad\}$. For each color $c \in \mathcal{C}$, the vertices labeled $\zeros$ and $\noedgeaddress$ have a directed edge colored $c$ to the vertex labeled $\invalidaddress$. The vertex labeled $\invalidaddress$ also has 3 self-loop edges, each of a distinct color from $\mathcal{C}$. Every other vertex in $t$ is joined to 3 other vertices, each by an undirected edge of a distinct color from $\mathcal{C}$. Every vertex $t$ of $\mathcal{T}$ whose label is not in $\mathsf{SpecialAddresses}$ is labeled by the sequence of colors that specifies the (shortest) path from $\emptyaddress$ to $t$ in $\mathcal{T}$. For any vertex $t$ of $\mathcal{T}$, let $\lambda_c(t)$ be the vertex that is joined to $t$ by an edge of color $c$ in $\mathcal{T}$.
\end{definition}

\begin{figure}
    \centering
    \small
    \resizebox{0.92\textwidth}{!}{%
    \begin{tikzpicture}[scale=0.57, auto, node distance=0.2cm, every loop/.style={}, thick, every arrow/.append style={dash,thick}]
    \tikzset{VertexStyle/.style = {shape = ellipse,  minimum size = 30pt, draw}}
    \Vertex[x=0,y=30,L=$\emptyadd$]{emp}
    \Vertex[x=-12,y=26,L=$r$]{r}
    \Vertex[x=0,y=26,L=$b$]{b}
    \Vertex[x=12,y=26,L=$\noedge$]{noe}
    \Vertex[x=18,y=26,L=$\zero$]{zer}
    \Vertex[x=-15,y=22,L=$r\comma g$]{rg}
    \Vertex[x=-9,y=22,L=$r\comma b$]{rb}
    \Vertex[x=-3,y=22,L=$b\comma g$]{bg}
    \Vertex[x=3,y=22,L=$b\comma r$]{br}
    \Vertex[x=15,y=22,L=$\invalid$]{inv}
    \Vertex[x=-16.5,y=18,L=$r\comma g\comma b$]{rgb}
    \Vertex[x=-13.5,y=18,L=$r\comma g\comma r$]{rgr}
    \Vertex[x=-10.5,y=18,L=$r\comma b\comma g$]{rbg}
    \Vertex[x=-7.5,y=18,L=$r\comma b\comma r$]{rbr}
    \Vertex[x=-4.5,y=18,L=$b\comma g\comma r$]{bgr}
    \Vertex[x=-1.5,y=18,L=$b\comma g\comma b$]{bgb}
    \Vertex[x=1.5,y=18,L=$b\comma r\comma g$]{brg}
    \Vertex[x=4.5,y=18,L=$b\comma r\comma b$]{brb}
    \tikzset{every node/.style={opacity=0,text opacity=1,scale=0.75}}
    \tikzset{EdgeStyle/.style=auto, color=red}
    \Edge[color=red](emp)(r)
    \Edge[color=blue](emp)(b)
    \Edge[color=ForestGreen](r)(rg)
    \Edge[color=blue](r)(rb)
    \Edge[color=ForestGreen](b)(bg)
    \Edge[color=red](b)(br)
    \Edge[color=blue](rg)(rgb)
    \Edge[color=red](rg)(rgr)
    \Edge[color=ForestGreen](rb)(rbg)
    \Edge[color=red](rb)(rbr)
    \Edge[color=red](bg)(bgr)
    \Edge[color=blue](bg)(bgb)
    \Edge[color=ForestGreen](br)(brg)
    \Edge[color=blue](br)(brb)
    \tikzset{EdgeStyle/.style=auto,post}
    \Edge[color=ForestGreen](emp)(noe)
    \Edge[color=blue](noe)(inv)
    \Edge[color=blue](zer)(inv)
    \tikzset{every loop/.style=auto,post}
    \path[->,line width=1.4pt] (inv) edge[in=195,out=220,loop,color=ForestGreen] ();
    \path[->,line width=1.4pt] (inv) edge[in=247.5,out=292.5,loop,color=blue] ();
    \path[->,line width=1.4pt] (inv) edge[in=320,out=345,loop,color=red] ();
    \tikzset{EdgeStyle/.style=auto,post,bend right=20}
    \Edge[color=red](noe)(inv)
    \Edge[color=red](zer)(inv)
    \tikzset{EdgeStyle/.style=auto,post,bend left=20}
    \Edge[color=ForestGreen](noe)(inv)
    \Edge[color=ForestGreen](zer)(inv)
    \end{tikzpicture}
    }%
    \caption{Address tree $\mathcal{T}$ of depth $3$ corresponding to the graph in \Cref{fig:weldedtreecoloredlabeled}. For the sake of brevity, we have removed the suffix \textsc{address} for all the addresses in $\specialaddresses$ and the tuple brackets for all the addresses not in $\specialaddresses$. Notice that, for each vertex, there is an edge (either directed or undirected) of each color outgoing from each vertex in $\mathcal{T}$.}
    \label{fig:addresstree}
\end{figure}
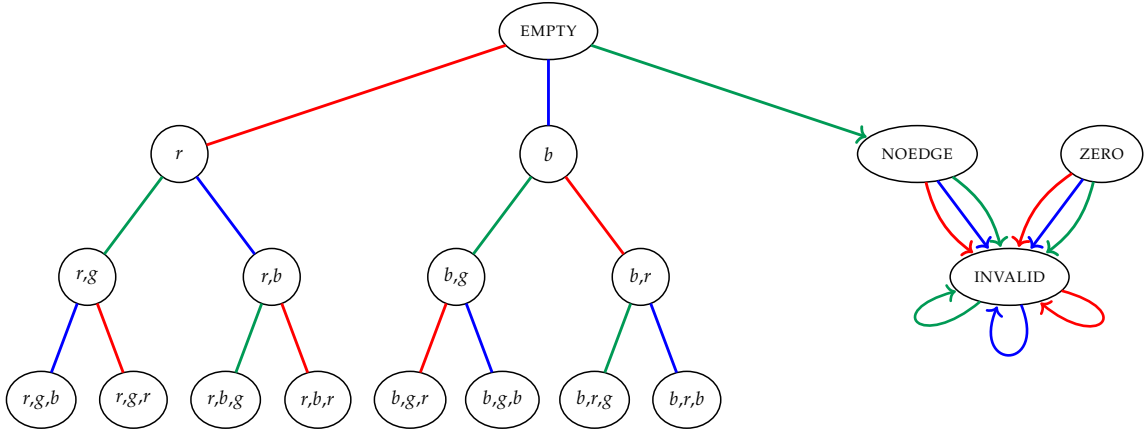

The following simple observations about the address tree $\mathcal{T}$ may be instructive.

\begin{itemize}
    \item Since the 3-coloring of $\mathcal{T}$ is a valid coloring, no vertex label of $\mathcal{T}$ contains an even-length palindrome. 
    \item Beginning at the vertex labeled $\emptyaddress$ and traversing any sequence of colors in $\mathcal{T}$ leads to some vertex of $\mathcal{T}$. Therefore, in the definition of the transcript state (\Cref{def:transcript}), and hence in the algorithm analyzed in \Cref{sec:simulation},
    the addresses that we consider are valid labels of vertices in $\mathcal{T}$, by construction.
    \item The color $\cbad$ can be computed with 2 queries to the oracle $\Oracle$. Therefore, the entire address tree can be computed with only 2 queries to $\Oracle$.
\end{itemize}

Intuitively, the transcript state $\ket{\phi_\Al}$ is the state that results from running the algorithm $\Al$ on the address tree $\mathcal{T}$ instead of the actual \wtg\ $\mathcal{G}$. If $\Al$ does not explore cycles in $\mathcal{G}$ to a significant extent, then $\ket{\phi_\Al}$ should be a good approximation of the state $\ket{\psi_\Al}$ produced by running $\Al$ on $\mathcal{G}$, as in \Cref{def:genuine}. In \Cref{sec:simulation}, we show that this is indeed the case for any genuine, rooted quantum algorithm $\Al$.

Now we define a mapping $B$ that turns addresses into strings, and another mapping $B^{\mathsf{inv}}$ that turns strings into addresses, such that $B^{\mathsf{inv}}$ is the inverse of $B$ on the range of $B$. In our analysis, the registers we consider can never contain any string that is not in the range of the $B$ mapping. Therefore, it is sufficient to define $B^{\mathsf{inv}}$ over the range of $B$. Nevertheless, we define $B^{\mathsf{inv}}$ over $\{0,1\}^{2p(n)}$ for the sake of completeness.

\begin{definition}[$B$ mapping] \label{def:Bmapping}
Let $\mathcal{V}_{\mathcal{T}}$ denote the set of labels of vertices of the address tree $\mathcal{T}$. Let $S$ be a subset of $\{0,1\}^{2p(n)}$ of size $|\mathcal{V}_{\mathcal{T}}|$ containing $0^{2p(n)}$. Let $\emptystring$, $\noedgestring$, and $\invalidstring$ be any distinct fixed strings in $S \setminus \{0^{2p(n)}\}$. Then $B\colon \mathcal{V}_{\mathcal{T}} \to S$ is a bijection mapping $\zeros$ to $0^{2p(n)}$, $\emptyaddress$ to $\emptystring$, $\noedgeaddress$ to $\noedgestring$, and $\invalidaddress$ to $\invalidstring$. We define the function $B^{\mathsf{inv}}\colon \{0,1\}^{2p(n)} \to \mathcal{V}_{\mathcal{T}}$ as
\begin{equation}
    B^{\mathsf{inv}}(s) \defeq
    \begin{cases}
        B^{-1}(s) & s \in S \\
        \invalidaddress & \text{otherwise}.
    \end{cases}
\end{equation}
\end{definition}

We now define analogs of the spaces introduced in \Cref{def:vertexspace,def:workspace} that our transcript state (\Cref{def:transcript}) lies in and that our classical simulation algorithm (\Cref{alg:C(T)1}) acts on.  

\begin{definition}[Address register and address space]
An \emph{address register} is a $2p(n)$-qubit register storing bit strings that are the image, under the map $B$, of the address of some vertex label in the address tree $\mathcal{T}$. We consider quantum states that have exactly $p(n)$ address registers, and refer to the $2p(n)^2$-qubit space of all the address registers as the \emph{address space}. 
\end{definition}

\begin{definition}[Address workspace and address workspace register]
An \emph{address workspace register} is a single-qubit register that stores arbitrary ancillary states. We allow arbitrarily many address workspace registers, and refer to the space consisting of all address workspace registers as the \emph{address workspace}.
\end{definition}

Notice the similarity between the definitions of workspace and address workspace. Indeed, we will later observe that the projection of $\ket{\psi_\Al}$ on the workspace is the same as the projection of the transcript state on the address workspace in the subspace not containing the $\exit$ or a cycle. We are now ready to state the definition of the transcript state $\ket{\phi_\Al}$  associated with the quantum state $\ket{\psi_\Al}$.

\begin{definition}[Transcript state] \label{def:transcript}
Consider a $p(n)$-query \genuine, \rooted\ quantum algorithm $\Al$. Given a circuit $C$ that implements $\Al$, acting on the vertex space and the workspace, let $\tilde{C}$ be the quantum circuit that acts on the address space and the address workspace, obtained by the following procedure.\footnote{Notice that the time complexity of this procedure is linear in the size of the circuit $C$.}
\begin{enumerate}

\item Determine $\cbad$ using two queries to the oracle $\Oracle$.

\item Replace each vertex register with an address register and each workspace register with an address workspace register. Replace the initial state used in the genuine algorithm (recall \Cref{def:genuine}) with the new initial state
\begin{align}\label{eq:initialtranscriptstate}
    \ket{\phi_{\mathrm{initial}}} \defeq \ket{\emptystring} \otimes \left(\ket{0^{2p(n)}} \right)^{\otimes(p(n)-1)} \otimes \ket{0}_{\mathrm{addressworkspace}}.
\end{align}

In \crefrange{item:oraclegates}{item:workspacegates} below, we describe gates that act on the address space analogously to how the gates in \Cref{def:genuinecircuit} act on the vertex space. For any vertex $v \in \valid$, we write $s_v \in \{0,1\}^{2p(n)}$ to denote the contents of the address register corresponding to the vertex register storing $v$. The transcript state is produced by the unitary operation that results by replacing each vertex-space gate in the quantum algorithm $\mathcal{A}$ with the corresponding address-space gate defined below.

\item \label{item:oraclegates} Replace any controlled-oracle gate in $C$ (controlled on workspace register $a$ and acting on vertex registers $j$ and $k$) with controlled-$\tilde{\Oracle}_c$ (controlled on address workspace register $a$ and acting on address registers $j$ and $k$), where
\begin{align}
    \tilde{\Oracle}_c\colon
	\ket{s_j} \ket{s_k} \mapsto \ket{s_j} \ket{s_k \oplus B(\lambda_c(B^{\mathsf{inv}}(s_j)))}.
\end{align} 

\item Replace any controlled-$e^{i \theta T}$ gate in $C$ (controlled on workspace register $a$ and acting on vertex registers $j$ and $k$) with a controlled-$e^{i \theta \tilde{T}}$ gate (controlled on address workspace register $a$ and acting on address registers $j$ and $k$), where
\begin{align}
    \tilde{T}\colon\ket{s_j}\ket{s_k} \mapsto \ket{s_k}\ket{s_j}.
\end{align}

\item Replace any equality check gate $\mathcal{E}$ in $C$ (controlled on vertex registers $j$ and $k$ and acting on workspace register $a$) with $\tilde{\mathcal{E}}$ (controlled on address registers $j$ and $k$ and acting on address workspace register $a$), where
\begin{equation}
    \tilde{\mathcal{E}}\colon \ket{s_j}\ket{s_k}\ket{w_a} \mapsto \ket{s_j}\ket{s_k}\ket{w_a \oplus \delta[s_j = s_k]}.
\end{equation}

\item Replace any $\noedge$-check gate $\mathcal{N}$ in $C$ (controlled on vertex register $j$ and acting on workspace register $a$) with $\tilde{\mathcal{N}}$ (controlled on address register $j$ and acting on address workspace register $a$), where
\begin{equation}
    \tilde{\mathcal{N}}\colon \ket{s_j}\ket{w_a} \mapsto \ket{s_j}\ket{w_a \oplus \delta[s_j = \noedgestring]}.
\end{equation}

\item Replace any $\zero$-check gate $\mathcal{Z}$ in $C$ (controlled on vertex register $j$ and acting on workspace register $a$) with $\tilde{\mathcal{Z}}$ (controlled on address register $j$ and acting on address workspace register $a$), where
\begin{equation}
    \tilde{\mathcal{Z}}\colon \ket{s_j}\ket{w_a} \mapsto \ket{s_j}\ket{w_a \oplus \delta[s_j = 0^{2p(n)}]}.
\end{equation}
		
\item \label{item:workspacegates} Leave gates acting on the workspace unchanged.
\end{enumerate}

The \emph{transcript state} $\ket{\phi_{\mathcal{A}}}$ is obtained by applying the circuit $\tilde{C}$ to the string  $\emptystring = B(\emptyaddress)$, together with $p(n)-1$ ancilla address registers storing $0^{2p(n)} = B(\zeros)$. In other words,
\begin{align} \label{eq:phiA}
	\ket{\phi_{\mathcal{A}}} \defeq \tilde{C} \ket{\phi_{\mathrm{initial}}}.
\end{align}
\end{definition}
	
Notice that whereas $C$ updates the vertex registers by making many oracle queries to $O$, the circuit $\tilde{C}$ only makes two queries to $O$. 

\section{Classical simulation of genuine, rooted algorithms}
\label{sec:simulation}

We now describe a classical algorithm for simulating genuine, rooted quantum algorithms.
We begin in \Cref{subsec:checking} by describing procedures for checking that the behavior of a quantum algorithm is genuine and rooted. While these procedures have no effect on a quantum algorithm with those properties, they enforce properties of the transcript state that are useful in our analysis.
Then, in \Cref{subsec:L'}, we describe a mapping that sends states in the address space to states in the vertex space, which is used to describe our simulation algorithm in  \Cref{subsec:classicalalgorithm}. In \Cref{subsec:goodbadandugly}, we decompose the state into components that assist in our analysis. We show in \Cref{subsec:analysisongoodpart} that the `good' part of the state of a genuine, rooted algorithm is related, via the mapping $L$ defined in \Cref{def:DefL}, to the `good' part of the state of our simulation at each intermediate step. Finally, we establish in \Cref{subsec:analysisonall} (using the result of \Cref{sec:3coloring}) that no genuine, rooted quantum algorithm can find an $\entrance$--$\exit$ path (or a cycle) with more than exponentially small probability.

\subsection{Checking procedures} \label{subsec:checking}

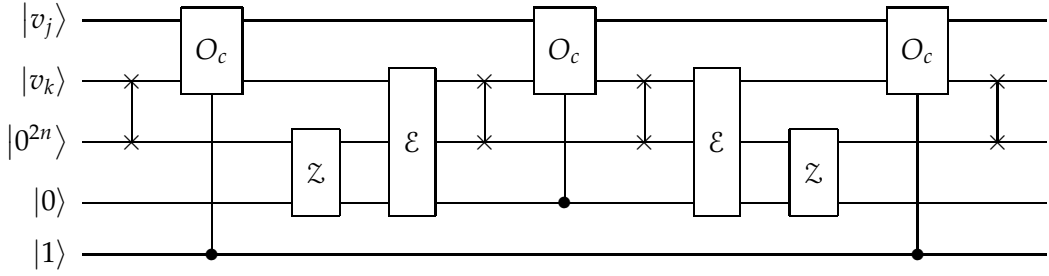
\begin{figure}
    \centering
    \[
    \Qcircuit @C=1.7em @R=1.2em {
    & \lstick{\ket{v_j}} & \qw & \multigate{1}{O_c} & \qw & \qw & \qw & \multigate{1}{O_c} & \qw & \qw & \qw & \multigate{1}{O_c} & \qw & \qw \\
    & \lstick{\ket{v_k}} & \qswap \qwx[1] & \ghost{O_c} & \qw & \multigate{2}{\mathcal{E}} & \qswap \qwx[1] & \ghost{O_c} & \qswap \qwx[1] & \multigate{2}{\mathcal{E}} & \qw & \ghost{O_c} & \qswap \qwx[1] & \qw \\
    & \lstick{\ket{0^{2n}}} & \qswap \qw & \qw & \multigate{1}{\mathcal{Z}} & \ghost{\mathcal{E}} & \qswap \qw & \qw & \qswap \qw & \ghost{\mathcal{E}} & \multigate{1}{\mathcal{Z}} & \qw & \qswap \qw & \qw  \\
    & \lstick{\ket{0}} & \qw & \qw & \ghost{\mathcal{Z}} & \ghost{\mathcal{E}} & \qw & \ctrl{-2} & \qw & \ghost{\mathcal{E}} & \ghost{\mathcal{Z}} & \qw & \qw & \qw \\
    & \lstick{\ket{1}} & \qw & \ctrl{-3} & \qw & \qw & \qw & \qw & \qw & \qw & \qw & \ctrl{-3} & \qw & \qw
    }
    \]
    \caption{Circuit diagram for checking whether $v_k \in \{0^{2n},\eta_c(v_j)\}$. The top three registers (i.e., those initialized with $\ket{v_j}$, $\ket{v_k}$, and $\ket{0^{2n}}$, respectively) are vertex registers and the bottom two (i.e., those initialized with $\ket{0}$ and $\ket{1}$, respectively) are workspace registers.}
    \label{fig:genuinitycheck}
\end{figure}

First note that we can efficiently check whether, for an oracle query $\controlled{O_c}$ with the input vertex register storing $\ket{v_j}$, the output vertex register contains $0^{2n}$ or $\eta_c(v_j)$. Indeed, we can replace each oracle query gate $O_c$ with the circuit shown in \Cref{fig:genuinitycheck} (where the entire circuit is controlled on the control register of $O_c$ in the workspace), which uses a constant number of gates from \Cref{def:genuinecircuit}. The swap gates need not be explicitly performed, and are included only so that all the wires that a certain gate acts on are adjacent. The last workspace register is used to apply uncontrolled-$O_c$ gates. In this circuit, the center oracle gate $O_c$ is only applied on the registers storing $\ket{v_j}$ and $\ket{v_k}$ if the first workspace register stores a $1$, which happens only when $v_k = 0^{2n}$ or $v_k = \eta_c(v_j)$ by the definitions of the $\zero$ check gate $\mathcal{Z}$ and the equality check gate $\mathcal{E}$. Since we are promised that the output register of any oracle gate satisfies \cref{itm:genuineoracle} of \Cref{def:genuinecircuit}, replacing each controlled-oracle gate in any given genuine circuit $C$ with the gadget of \Cref{fig:genuinitycheck} does not impact the output state of $C$, while only increasing the circuit size by a constant factor. 

\begin{remark}[Checking genuineness] \label{rem:genuinenesscheck}
Given any genuine circuit $C$ with $|C|$ gates, one can efficiently construct a genuine circuit $C'$ consisting of $O(|C|)$ gates such that $C'$ has the same functionality as $C$ and verifies the condition stated after \cref{eq:genuineoracle} in \cref{itm:genuineoracle} of \Cref{def:genuinecircuit} before applying each oracle call $O_c$. Therefore, we assume without loss of generality that the given genuine circuit $C$ has built-in gadgets that verify this condition. 
\end{remark}

The consequence of \Cref{def:transcript} and \Cref{rem:genuinenesscheck} is a crucial observation about transcript states, which will turn out be useful in our analysis in \Cref{sec:simulation}.

\begin{lemma} \label{lem:transcriptstateisintherangeofB}
For any given rooted genuine circuit $C$, any string stored in any computational basis state in the support of the address space of the state $\tilde{C}\ket{\phi_{\mathrm{initial}}}$ is in the range of $B$. 
\end{lemma}

\begin{proof}
Recall from \Cref{rem:genuinenesscheck} that the given genuine circuit $C$ has built-in gadgets described by \Cref{fig:genuinitycheck} that verify the condition asserted in \cref{itm:genuineoracle} of \Cref{def:genuinecircuit}. Since we construct $\tilde{C}$ from $C$ by a gate-by-gate process, we apply an oracle gate $\tilde{O}_c$ if the target register in the workspace of the gates $\tilde{\mathcal{Z}}$ and $\tilde{\mathcal{E}}$ just before the gate $\controlled{\tilde{O}_c}$ (as in \Cref{fig:genuinitycheck}) is in the state $\ket{1}$. Therefore, by the definition of the $\tilde{\mathcal{Z}}$ and $\tilde{\mathcal{E}}$ gates, the oracle gate $\controlled{\tilde{O}_c}$ acting on address registers storing $\ket{s_j}$ and $\ket{s_k}$ is only applied when $s_k = 0^{2p(n)}$ or $s_k = B(\lambda_c(B^{\mathsf{inv}}(s_j))$ (and the control qubit of $\controlled{\tilde{O}_c}$ in the address workspace is in the state $\ket{1}$). Notice that all the other gates in \Cref{def:transcript} either do not alter the address registers at all or shuffle their positions without changing the address strings stored. It follows that no gate from \Cref{def:transcript} will introduce address labels that are not in the range of the $B$ mapping defined in \Cref{def:Bmapping}. Since $\ket{\phi_{\mathrm{initial}}}$ does not store any address labels not in the range of $B$, applying the circuit $\tilde{C}$ to $\ket{\phi_{\mathrm{initial}}}$ will not generate address labels outside the range of $B$.
\end{proof}

We can use a similar approach to efficiently modify a given quantum query algorithm to ensure that its state is always rooted. This modification will be useful for the analysis in \Cref{subsec:goodbadandugly,subsec:analysisongoodpart}. 

In particular, we claim that given a genuine, rooted algorithm associated with a circuit $C$, one can efficiently construct a modified genuine, rooted algorithm associated with a circuit $C'$, with each gate $G$ in $C$ replaced by a sequence of gates in $C'$, ensuring that $G$ is only applied if the resulting state after applying $G$ would have been rooted, with a polynomial overhead in circuit size and no impact on the resulting state. Recall that (controlled) oracle gates are the only genuine gates that can alter the contents of vertex registers. This means that we only need to replace $G$ with this sequence of gates if $G$ is an oracle gate. 
Moreover, by \Cref{def:rootedstate}, to verify that a given state is rooted, we only need to check that the vertex registers not storing $0^{2n}$ and $\noedge$ labels form a connected subgraph of $\mathcal{G}$ containing the $\entrance$.
We now describe a rooted algorithm that accomplishes this task and argue how it can be implemented by a genuine circuit.

First, note that at the beginning of the circuit $C$, one can copy the label of the $\entrance$ (which is stored in the first register of $\ket{\psi_{\mathrm{initial}}}$) to an ancilla vertex register that is not meant to be used by any of the gates that follow. This step can be implemented by a genuine algorithm by querying a valid neighbor of the $\entrance$ and then computing the $\entrance$ in this special ancilla vertex register before uncomputing this neighbor of the $\entrance$. Similarly, we can uncompute the contents of the special ancilla vertex register at the end of our algorithm. Thus, we can make sure that the $\entrance$ label is always stored in a vertex register of any computational basis state in the support of our state at any step.

Given a computational basis state consisting of $p(n)$ vertex registers, by a standard breadth-first search procedure starting at the $\entrance$, one can check whether the subgraph $G$ of $\mathcal{G}$ induced by the vertex labels stored in these registers is connected. At each step of this breadth-first search, we determine which vertex registers store the neighbors of a particular vertex $v$. This can be done by looping over each vertex register, checking whether it stores a neighbor of $v$ and storing the outcome in a workspace register. Therefore, the task of checking the connectivity of $G$ is reduced to the task of checking whether two input vertex registers store labels of vertices that are neighbors in $\mathcal{G}$ and storing this in a workspace register. 

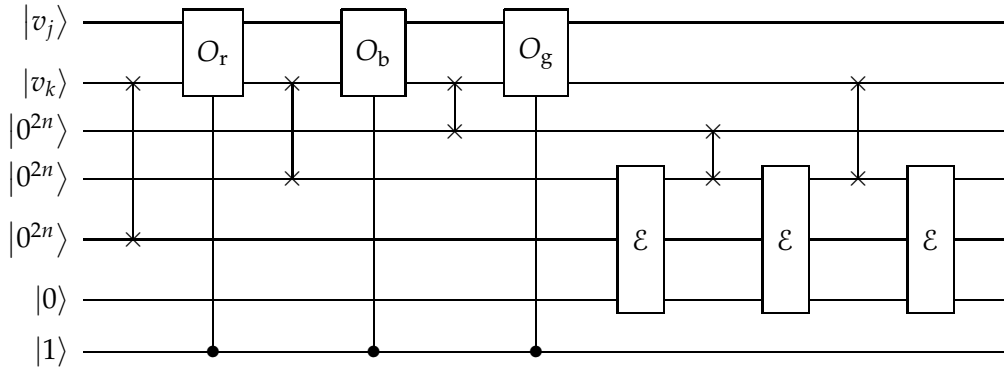
\begin{figure}
    \centering
    \[
    \Qcircuit @C=1.7em @R=1.2em {
    & \lstick{\ket{v_j}} & \qw & \multigate{1}{O_{\mathrm{r}}} & \qw & \multigate{1}{O_{\mathrm{b}}} & \qw & \multigate{1}{O_{\mathrm{g}}} & \qw & \qw & \qw & \qw & \qw & \qw \\
    & \lstick{\ket{v_k}} & \qswap \qwx[3] & \ghost{O_{\mathrm{r}}} & \qswap \qwx[2] & \ghost{O_{\mathrm{b}}} & \qswap \qwx[1] & \ghost{O_{\mathrm{g}}} & \qw & \qw & \qw & \qswap \qwx[2] & \qw & \qw \\
    & \lstick{\ket{0^{2n}}} & \qw & \qw & \qw & \qw & \qswap \qw & \qw & \qw & \qswap \qwx[1] & \qw & \qw & \qw & \qw \\
    & \lstick{\ket{0^{2n}}} & \qw & \qw & \qswap \qw & \qw & \qw & \qw & \multigate{2}{\mathcal{E}} & \qswap \qw & \multigate{2}{\mathcal{E}} & \qswap \qw & \multigate{2}{\mathcal{E}} & \qw \\
    & \lstick{\ket{0^{2n}}} & \qswap \qw & \qw & \qw & \qw & \qw & \qw & \ghost{\mathcal{E}} & \qw & \ghost{\mathcal{E}} & \qw & \ghost{\mathcal{E}} & \qw \\
    & \lstick{\ket{0}} & \qw & \qw & \qw & \qw & \qw & \qw & \ghost{\mathcal{E}} & \qw & \ghost{\mathcal{E}} & \qw & \ghost{\mathcal{E}} & \qw \\
    & \lstick{\ket{1}} & \qw & \ctrl{-5} & \qw & \ctrl{-5} & \qw & \ctrl{-5} & \qw & \qw & \qw & \qw & \qw & \qw
    }
    \]
    \caption{Circuit diagram for computing in an ancilla workspace register whether $v_j = \eta_c(v_k)$ for some $c \in \mathcal{C}$. The top five registers are vertex registers and the bottom two registers are workspace registers. For compactness, we write $O_{\mathrm{r}}$, $O_{\mathrm{b}}$, and $O_{\mathrm{g}}$ for $O_{\mathrm{red}}$, $O_{\mathrm{blue}}$, and $O_{\mathrm{green}}$, respectively.}
    \label{fig:neighborcheck}
\end{figure}

We show a genuine circuit for this procedure in \Cref{fig:neighborcheck}. The swap gates and the workspace register initialized to $\ket{1}$ have the same roles as in \Cref{fig:genuinitycheck}. Given vertex labels $v_j$ and $v_k$, we first compute each of the 3 neighbors of $v_j$ in 3 different ancilla vertex registers. Then, we check whether any of these vertices are equal to $v_k$ using equality check gates $\mathcal{E}$, and compute the output in an ancilla workspace regsiter. By the end of this circuit, this workspace register will store $1$ if and only if $v_j$ and $v_k$ are neighbors. Once we have used this workspace qubit to apply a controlled oracle gate, we uncompute the contents of this qubit by applying the circuit in \Cref{fig:neighborcheck} backwards.

Now that we have outlined a procedure for checking rootedness of a given state, we describe a procedure for implementing each controlled oracle query gate $O_c$ while maintaining rootedness.
We first compute the output of $O_c$ in an ancilla vertex register. Then we determine whether $O_c$ is meant to uncompute the contents of a particular vertex register.
We check this using an equality check gate applied on the ancilla vertex register and the target register of $O_c$. If $O_c$ is not meant to uncompute a vertex register, then applying it cannot result in a non-rooted state. If $O_c$ performs an uncomputation, we check rootedness of the collection of vertex registers not including the target register of $O_c$. We then apply $O_c$ controlled on the output of this check and uncompute all the information we computed in ancilla vertex and workspace registers. This sequence of gates never results in a non-rooted state and does not alter the output state of a genuine, rooted algorithm.
Thus we have the following.

\begin{remark}[Checking rootedness] \label{rem:rootednesscheck}
Given any genuine circuit $C$ with $|C|$ gates, one can efficiently construct a rooted genuine circuit $C'$ consisting of $\poly(|C|)$ gates such that $C'$ has the same functionality as $C$ and before applying each gate $G$, $C'$ checks whether the resulting state will remain rooted after the application of $G$. Therefore, we assume without loss of generality that the given rooted genuine circuit $C$ has built-in rootedness check gadgets.
\end{remark}

Analogous to the notion of a rooted state defined in \Cref{def:rootedstate}, we define the notion of an address-rooted state as follows. Informally, a state in the address space is address rooted if, when it contains a string that encodes an address, it also contains the string that encodes its parent in $\mathcal{T}$.

\begin{definition}[Address-rooted state] \label{def:addressrooted}
We say that a computational basis state $\ket{\phi}$ in the address space is \emph{address rooted} if for any string $s$ stored in any of the registers of $\ket{\phi}$, whenever the vertex $B^{\mathsf{inv}}(s) \in \mathcal{V}_{\mathcal{T}}$ has a parent $t \neq \zeros$, there exists a register of $\ket{\phi}$ that stores the string $B(t)$. 
\end{definition}

Now we show, using \Cref{rem:rootednesscheck}, that the notion of address rooted for states in the address space is analogous to the notion of rooted for states in the vertex space. This result substantially simplifies the analysis in \Cref{sec:simulation}.  

\begin{lemma} \label{lem:transcriptstateisalwaysrooted}
For any given rooted genuine circuit $C$, any computational basis state in the support of the address space of the state $\tilde{C}\ket{\phi_{\mathrm{initial}}}$ is address rooted.
\end{lemma}

\begin{proof}
Recall from \Cref{rem:rootednesscheck} that we assumed the given circuit has built-in rootedness check gadgets that only apply a gate $G$ if the resulting state is guaranteed to be rooted. 
Since $\ket{\phi_{\mathrm{initial}}}$ stores the $\emptystring$ label (instead of the $\entrance$ label stored in $\ket{\psi_{\mathrm{initial}}}$), the first step of the procedure outlined above \Cref{rem:rootednesscheck} ensures that $\emptyaddress$ is always stored in some register of every computational basis state at every step of our transcript state. 
Note that the analog of the circuit in \Cref{fig:neighborcheck} in the address space (constructed as per \Cref{def:transcript}) checks whether the addresses corresponding to two strings are neighbors in $\mathcal{T}$. Thus, the analog in the address space of the above procedure to check rootedness of a given state checks whether a given computational basis state in the address space corresponds to a subtree of $\mathcal{T}$ containing the $\emptystring$ in one of its registers. Since $\mathcal{T}$ is a tree, this check is equivalent to checking whether the given state is address rooted by \Cref{def:addressrooted}. 

As per the above procedure, before applying a $\controlled{\tilde{O_c}}$ gate on a computational basis state $\ket{\phi}$ controlled on a register storing $s_j$ with the target register storing $s_k$, we first check whether it is an uncomputation (i.e., $s_k = B\lambda_cB^{\mathrm{inv}}(s_j)$) via the $\tilde{\mathcal{E}}$ gate.
If this check passes, we apply the $\controlled{\tilde{O_c}}$ gate only after verifying that the resulting state (i.e., $\ket{\phi}$ with $s_k$ replaced with $0^{2p(n)}$) would be address rooted, analogous to the procedure preceding \Cref{rem:rootednesscheck}. Thus, the address rootedness of our state is ensured in this case. 
If this check fails, then we know, by the proof of \Cref{lem:transcriptstateisintherangeofB}, that $s_k = 0^{2p(n)}$. In this case, the resulting state (i.e. $\ket{\phi}$ with $s_k$ replaced with $B\lambda_cB^{\mathrm{inv}}(s_j)$) corresponds to a connected sub-tree of $\mathcal{T}$ containing the $\emptystring$ (as $\ket{\phi}$ is a connected sub-tree of $\mathcal{T}$ containing the $\emptystring$ and $s_j$ is stored in some register of $\ket{\phi}$) so will be address rooted.
\end{proof}

\subsection{Mapping addresses to vertices} \label{subsec:L'}

We now define an efficiently computable function $L'$ that maps an address $t$ to a corresponding vertex label $v$, and observe some relationships of addresses and the vertices they map to under $L'$. For $t \in \{\zeros,\invalidaddress\}$, this function simply outputs the corresponding vertex label. Otherwise, the image of $t$ under $L'$ is obtained by performing a sequence of oracle calls to determine the vertices reached by following edges of the colors specified by $t$, and outputting the vertex label reached at the end of that sequence. More precisely, $L'(t)$ is computed as follows.

\smallskip

\begin{algorithm}[H]\label{alg:lprime}
    \SetAlgoLined
    \KwInput{An address $t$}
    \KwOutput{A vertex label $v$}
    \If{$t=\zeros$}{
        \textbf{return} $0^{2n}$
    }
    \If{$t=\invalidaddress$}{
        \textbf{return} $\invalid$
    }
    $v\leftarrow \entrance$\;
    \For{$i = 1 \ldots |t|$}{
  	   $v\leftarrow \eta_{t[i]}(v)$\;
    }
    \textbf{return} $v$\;
  \caption{Classical query algorithm for computing $L'(t)$}
\end{algorithm}

\smallskip

Here $|t|$ denotes the length of the address $t$ (the number of colors in its color sequence) and $t[i]$ denotes the $i$th color. We now consider immediate implications of the definition of $L'$ in \Cref{alg:lprime}. We begin with the following lemma, which states that any address of $\entrance$ that is not the $\emptyaddress$ encodes a cycle in $\mathcal{G}$.

\begin{lemma} \label{lem:entrancecycle}
Let $t \neq \emptyaddress$ be such that $L'(t) = \entrance$. Then traversing the edge colors listed in $t$ beginning from the $\entrance$ yields a cycle in $\mathcal{G}$. 
\end{lemma}

\begin{proof}
Since $L'(t) = \entrance$ and $t \ne \emptyaddress$, we know that $t \notin \specialaddresses$.
Therefore, $t$ can be written as a sequence of edge colors. As noted earlier, this sequence of colors does not contain any even-length palindrome. This means beginning at the $\entrance$ in $\mathcal{G}$ and following the edge colors listed in $t$ does not involve any backtracking. Moreover, as $L'(t) = \entrance$, traversing this sequence of colors results in reaching the $\entrance$. Therefore, this sequence in $\mathcal{G}$ starting and ending at $\entrance$ forms a cycle.
\end{proof}

We can generalize \Cref{lem:entrancecycle} to show that any two distinct addresses of any vertex label in $\valid \cup \specialvertices$ together encode the address of the $\exit$ or a cycle in $\mathcal{G}$.

\begin{lemma} \label{lem:addresstreecycle}
Let $t$ and $t'$ be addresses with $t \neq t'$ and $L'(t) = L'(t')$. If $L'(t) = L'(t') \not \in \specialvertices$, then beginning from the $\entrance$ in $\mathcal{G}$ and following the edge colors listed in $t$ in order and then following the edge colors listed in $t'$ in reverse order forms a path that contains a non-trivial cycle in $\mathcal{G}$. Otherwise, there is a $\tau \in \{t, t'\}$ such that beginning from the $\entrance$ in $\mathcal{G}$ and following the edge colors listed in $\tau$ will result in either reaching the $\exit$ or forming a path that contains a cycle in $\mathcal{G}$.
\end{lemma}

\begin{proof}
Let $v = L'(t) = L'(t')$. We consider six cases:
\begin{enumerate}
\item $v = 0^{2n}$. By \Cref{alg:lprime}, $\tau = \zeros$ is the only address $\tau$ such that $L'(\tau) = 0^{2n}$. Therefore, $t = t' = \zeros$, so this case is not possible.

\item $v = \entrance$. Since $t \neq t'$, either $t \neq \emptyaddress$ or $t' \neq \emptyaddress$. In either case, the result follows from \Cref{lem:entrancecycle}.

\item $v = \exit$. Then for both $\tau \in \{t, t'\}$, beginning from the $\entrance$ in $\mathcal{G}$ and following the edge colors listed in $\tau$ will result in reaching the $\exit$.

\item \label[case]{case:noedgecase} $v = \noedge$. Since $t \neq t'$, either $t \neq \noedgeaddress$ or $t' \neq \noedgeaddress$. Without loss of generality, let $t \neq \noedgeaddress$. This means that $t \notin \specialaddresses$, so it can be written as a non-empty sequence $(c_1, \ldots, c_{|t|})$ of colors. Let $\tau$ be the address specified by the color sequence $(c_1, \ldots, c_{|t|-1})$. By \Cref{alg:lprime}, $\eta_{c_{|t|}}(L'(\tau)) = L'(t) = v$. Since $v = \noedge$, we have $L'(\tau) \in \{\entrance, \exit\}$ by \Cref{def:eta_c}. If $L'(\tau) = \exit$, then following the edge colors in $\tau$, and hence in $t$, results in reaching the $\exit$, so we are done. It remains to consider $L'(\tau) = \entrance$. It must be that $\tau \neq \emptystring$; otherwise, $t$ would be $\noedgeaddress$. Thus, by \Cref{lem:entrancecycle}, following the edge colors in $\tau$, and hence in $t$, beginning from the $\entrance$ forms a cycle in $\mathcal{G}$.

\item $v = \invalid$. Since $t \neq t'$, either $t \neq \invalidaddress$ or $t' \neq \invalidaddress$. Without loss of generality, let $t \neq \invalidaddress$. Then $t \notin \specialaddresses$, so it can be written as a non-empty sequence $(c_1, \ldots, c_{|t|})$ of colors. Let $\tau$ be the address specified by the color sequence $(c_1, \ldots, c_{|t|-1})$. By \Cref{alg:lprime}, $\eta_{c_{|t|}}(L'(\tau)) = L'(t) = v$. Since $v = \invalid$, we have $L'(\tau) \in \{0^{2n}, \noedge, \invalid\}$ by \Cref{def:eta_c}. By \Cref{alg:lprime}, $L'(\tau) = 0^{2n}$ only when $\tau = \zeros$. But we know that $\tau \neq \zeros$, so we cannot have $L'(\tau) = 0^{2n}$. If $L'(\tau) = \noedge$, then the desired result follows from \cref{case:noedgecase}. If $L'(\tau) = \invalid$, then we let $t = \tau$ and can apply the same argument recursively. We conclude that following the edge colors in $\tau$, and hence in $t$, beginning from the $\entrance$ forms a non-trivial cycle in $\mathcal{G}$.

\item $v \not \in \specialvertices$. This means that $t, t' \notin \specialaddresses$, so we can write $t = (c_1, \ldots, c_{|t|})$ and $t' = (c'_1, \ldots, c'_{|t'|})$ as non-empty sequences of colors. Let $\tau$ be the address specified by the color sequence $(c_1, \ldots, c_{|t|}, c'_{|t'|}, \ldots c'_1)$ formed by concatenating the sequence $t$ with the sequence $t'$ in reverse order. Since $v$ is a valid vertex of $\mathcal{G}$, following this sequence in $\mathcal{G}$ beginning with $\entrance$ will result in reaching the $\entrance$ (via $v$). That is, $L'(\tau) = \entrance$. Moreover, as $t$ and $t'$ are vertex labels in the address tree $\mathcal{T}$, following the sequence given by $\tau$ in $\mathcal{T}$ beginning with the vertex labeled $\emptyaddress$ will not result in $\emptyaddress$; otherwise, $t = t'$. Our desired result follows by \Cref{lem:entrancecycle}.
\end{enumerate}
Since these cases cover all possible $v$, the result follows.
\end{proof}

The following lemma is critical for the proof of many results that lead up to \Cref{lem:psigood=L(phigood)}. Informally, it states if an address $t$ does not encode an $\entrance$--$\exit$ path or a cycle in $\mathcal{G}$, then the $c$-neighbor of the vertex corresponding to $t$ in $\mathcal{G}$ is the same as the vertex corresponding to the $c$-neighbor of $t$ in $\mathcal{T}$.  

\begin{lemma} \label{lem:L'lambda=etaL'}
Let $v$ be any vertex label, let $t$ be an address of $v$, and let $c \in \mathcal{C}$. Furthermore, if $t \notin \specialaddresses$, suppose that following the edge colors given by $t$ starting at the $\entrance$ does not result in reaching the $\exit$ or finding a cycle in $\mathcal{G}$. Then $L'(\lambda_c(t)) = \eta_c(v)$.
\end{lemma}

\begin{proof}
As $t$ is the address of $v$, we know that $L'(t) = v$ by \Cref{def:addresses} and \Cref{alg:lprime}. It remains to show that $L'(\lambda_c(t)) = \eta_c(L'(t))$.

First, suppose $t \in \specialaddresses$. Then we have four cases:
\begin{enumerate}
\item \label[case]{enum:tv=zeros} $t = \zeros$. Then
\begin{align}
    L'(\lambda_c(\zeros)) 
    &= L'(\invalidaddress) = \invalid \\
    &= \eta_c(0^{2n}) = \eta_cL'(\zeros)
\end{align}
where the first step follows from \Cref{def:addresstree}, the second and fourth steps follow from \Cref{alg:lprime}, and the third step follows from \Cref{def:eta_c}.

\item $t = \emptystring$. Then
\begin{align}
    L'(\lambda_c(\emptystring)) = L'((c)) = \eta_c(\entrance) = \eta_cL'(\emptystring)
\end{align}
where the first step follows from \Cref{def:addresstree}, the second and fourth steps follow from \Cref{alg:lprime}, and the third step follows from \Cref{def:eta_c}.

\item $t = \noedgeaddress$. This case follows by an argument analogous to \cref{enum:tv=zeros}.

\item $t = \invalidaddress$. This case also follows by an argument analogous to \cref{enum:tv=zeros}.
\end{enumerate}

Therefore, assuming $t \in \specialaddresses$, $L'(\lambda_c(t)) = \invalid = \eta_c(L'(t))$.

Now, suppose $t \not \in \specialaddresses$. This means we can write $t$ as a sequence $(c_1, \ldots, c_{|t|})$ of colors. 

We claim that $v$ is a label of a degree-3 vertex of $\mathcal{G}$. For a contradiction, assume that $v \in \specialvertices$. In that case, \Cref{lem:addresstreecycle} implies that beginning from the $\entrance$ in $\mathcal{G}$ and following the edge colors listed in $t$ will result in either reaching the $\exit$ or forming a path that contains a cycle in $\mathcal{G}$, which directly contradicts the hypotheses of the lemma. 

Since $v$ is a degree-3 vertex of $\mathcal{G}$, we have $\eta_c(\eta_c(v))=v$ by \Cref{def:genuinecircuit}. Therefore,
\begin{align}
    L'(\lambda_c(t)) 
    &= L'(\lambda_c((c_1, \ldots, c_{|t|}))) \\
    &= 
    \begin{cases}
        L'((c_1, \ldots, c_{|t|-1})) & c = c_{|t|} \\
        L'((c_1, \ldots, c_{|t|}, c)) & \text{otherwise}
    \end{cases} \\
    &=
    \begin{cases}
        (\eta_{c_{|t|-1}} \circ \cdots \circ \eta_{c_1})(v) & c = c_{|t|} \\
        (\eta_c \circ \eta_{c_{|t|}} \circ \cdots \circ \eta_{c_1})(v) & \text{otherwise}
    \end{cases} \\
    &= (\eta_c \circ \eta_{c_{|t|}} \circ \cdots \circ \eta_{c_1})(v) \\
    &= \eta_c(L'((c_1, \ldots, c_{|t|})) \\
    &= \eta_c(L'(t))
\end{align}
where the second step follows from \Cref{def:addresstree}, the third and fifth from \Cref{alg:lprime}, and the fourth from an observation made above. 
\end{proof}

\subsection{The classical algorithm} \label{subsec:classicalalgorithm}

We now describe our classical algorithm (\Cref{alg:C(T)1}) for simulating genuine quantum algorithms. To state the algorithm, we introduce several definitions, beginning with a map based on the function $L'$ defined in \Cref{subsec:L'}.

\begin{definition} \label{def:DefL}
For any $m \in [p(n)]$, the mapping $L\colon \left(\{0,1\}^{2p(n)}\right)^{m} \to \left(\{0,1\}^{2n}\right)^{m}$ sends $m$ address strings to $m$ vertex labels by acting as $L'B^{\mathsf{inv}}$ on each of the $m$ registers.
\end{definition}

When considering the map $L$ applied to a quantum state $\ket{\chi}$ on both the workspace and the address space, we use the shorthand $L\ket{\chi}$ to denote the state $(I_{\mathrm{workspace}} \otimes L_{\mathrm{vertex}}) \ket{\chi}$, with the map acting as the identity on the workspace register and as $L$ on the vertex register. 

To describe and analyze \Cref{alg:C(T)1}, we consider individual gates and sequences of consecutive gates from the genuine circuit $C$ defined in \Cref{def:genuine}. For this purpose, we consider the following definition. 

\begin{definition} \label{def:DefinitionofC}
For any $i \in [p(n)]$, let $C_i$ denote the $i$th gate of the circuit $C$ in \Cref{def:genuine}. For any $i, j \in [p(n)] \cup \{0\}$ with $i < j$, Let $C_{i,j}$ be the subsequence of gates from the circuit $C$ starting with the $(i+1)$st gate and ending with the $j$th gate. That is, $C_{i,j} \defeq C_{j} \cdots C_{i+1}$. Similarly, using the circuit $\tilde{C}$ constructed in \Cref{def:transcript}, we define $\tilde{C}_i$ and $\tilde{C}_{i,j}$ for each $i, j \in [p(n)] \cup \{0\}$ with $i < j$. 
\end{definition}
Note that $C_{i,i} = I$ and $C_{i-1, i} = C_i$ (and similarly, $\tilde{C}_{i,i} = I$ and $\tilde{C}_{i-1, i} = \tilde{C}_i$) for all $i \in [p(n)]$. We use these gates to define transcript states and states of the quantum algorithm for partial executions.

\begin{definition} \label{def:phi(i)}
For each $i \in [p(n)] \cup \{0\}$, let
\begin{align}
	\ket{\phi_{\mathcal{A}}^{(i)}} \defeq \tilde{C}_{0,i} \ket{\phi_{\mathrm{initial}}}
\end{align}
be the transcript state for the quantum algorithm $\mathcal{A}$ restricted to the first $i$ gates of $\tilde{C}$.
Similarly, let
\begin{align}
	\ket{\psi_{\mathcal{A}}^{(i)}} \defeq C_{0,i} \ket{\psi_{\mathrm{initial}}}
\end{align}
denote the state of the quantum algorithm $\mathcal{A}$ restricted to the first $i$ gates of $C$.
\end{definition}

In particular, the state $\ket{\phi_\Al^{(p(n))}} = \ket{\phi_\Al}$ is the transcript state corresponding to the quantum state $\ket{\psi_\Al^{(p(n))}} = \ket{\psi_\Al}$ introduced in \Cref{def:genuine}.

Now consider the following classical query algorithm for finding a path from the $\entrance$ to the $\exit$.  

\smallskip

\begin{algorithm}[H]
    \caption{Classical query algorithm $\mathcal{C}(\mathcal{A}(\Oracle))$}
    \label{alg:C(T)1}
	\For{$i \in [p(n)]$}{
	    Given the circuit diagram $C_{0,i}$, compute the transcript state $\ket{\phi^{(i)}_{\mathcal{A}}}$ as per \Cref{def:transcript}.
	
	    Sample a computational basis state $\ket{\phi^{(i)}}$ in the address space at random with probability $\norm{\bra{\phi^{(i)}}\ket{\phi^{(i)}_{\mathcal{A}}}}^2$.

	    Compute the computational basis state $L\ket{\phi^{(i)}}$ in the vertex space.
	    
	    \textbf{Output} the labels of the vertices in $L\ket{\phi^{(i)}}$.
	    
	    }
\end{algorithm}

\smallskip

Note that when $\mathcal{A}$ is \genuine\ and \rooted, the output of \Cref{alg:C(T)1} must be a connected subgraph of $\mathcal{G}$ containing the $\entrance$. Therefore, if the output of \Cref{alg:C(T)1} contains the $\exit$, it must reveal an $\entrance$-to-$\exit$ path. In the remainder of \Cref{sec:simulation}, we show that the output of \Cref{alg:C(T)1} contains the $\exit$ (or a cycle) with exponentially small probability.

\subsection{The good, the bad, and the ugly} \label{subsec:goodbadandugly}

We now define states $\ket{\psi_\good^{(i)}}$, $\ket{\psi_\bad^{(i)}}$, and $\ket{\psi_\allbad^{(i)}}$, which are components of the state $\ket{\psi_\Al^{(i)}}$. Intuitively, $\ket{\psi_\good^{(i)}}$ represents the portion of the state of the algorithm after $i$ steps that has never encountered the $\exit$ or a near-cycle (i.e., a subgraph that differs from a cycle by a single edge) at any point in its history, $\ket{\psi_\bad^{(i)}}$ represents the portion of the state of the algorithm after $i$ steps that just encountered the $\exit$ or a near-cycle at the $i$th step, and $\ket{\psi_\allbad^{(i)}}$ combines the portions of the state of the algorithm after $i$ steps that encountered the $\exit$ or a near-cycle at some point in its history. To formally define these states, we introduce the notion of $\mathrm{good}$ and $\mathrm{bad}$ states, which we define as follows.

\begin{definition} \label{def:phi-badandphi-good}
We say that a computational basis state $\ket{\phi}$ in the address space is $\phi$-$\mathrm{bad}$ if the subgraph corresponding to $L\ket{\phi}$ contains the $\exit$ or is at most one edge away from containing a cycle. A computational basis state $\ket{\phi}$ in the address space is $\phi$-$\mathrm{good}$ if it is not $\phi$-$\mathrm{bad}$, i.e., if $L\ket{\phi}$ does not contain the $\exit$ and is more than one edge away from containing a cycle.

Similarly, a computational basis state $\ket{\psi}$ in the vertex space is $\psi$-$\mathrm{bad}$ if the subgraph corresponding to $\ket{\psi}$ contains the $\exit$ or is at most one edge away from containing a cycle and is $\psi$-$\mathrm{good}$ if it is not $\psi$-$\mathrm{bad}$.
\end{definition}

Note that the map $L$ is used to define the notion of good and bad states in the address space, but is not used for the corresponding notions in the vertex space.

The $\mathrm{good}$ and $\mathrm{bad}$ states span the $\mathrm{good}$ and $\mathrm{bad}$ subspaces, respectively.

\begin{definition} \label{def:piphi}
We define the $\phi$-$\BAD$ subspace as
\begin{equation}
    \phi\text{-}\BAD \defeq \spn\left\{\ket{\phi}: \ket{\phi} \text{ is a $\phi$-$\mathrm{bad}$ state}\right\}.
\end{equation}
The $\phi$-$\GOOD$ subspace is
\begin{equation}
    \phi\text{-}\GOOD \defeq \spn\left\{\ket{\phi}: \forall \ket{\phi'} \in \phi\text{-}\BAD, \bra{\phi'}\ket{\phi} = 0\right\} = \spn\left\{\ket{\phi}: \ket{\phi} \text{ is a $\phi$-$\mathrm{good}$ state}\right\}.
\end{equation}
Let $\pibadphi$ and $\pigoodphi$ denote the projectors onto $\phi$-$\BAD$ and $\phi$-$\GOOD$, respectively. 

The subspaces $\psi$-$\BAD$ and $\psi$-$\GOOD$, and the projectors $\pibadpsi$ and $\pigoodpsi$, are defined analogously.
\end{definition}

Notice that $\pibadphi\pigoodphi=\pigoodphi\pibadphi=0$ and $\pibadphi+\pigoodphi=I$. Similarly,  $\pibadpsi\pigoodpsi=\pigoodpsi\pibadpsi=0$ and $\pibadpsi+\pigoodpsi=I$. 
We now define the states $\ket{\psi_\good^{(i)}}$, $\ket{\psi_\bad^{(i)}}$, and $\ket{\psi_\allbad^{(i)}}$ that were described informally above.

\begin{definition} \label{def:phigood}   

We define 
\begin{equation}
    \ket{\phi_\good^{(i)}} \defeq \pigoodphi \left(\ket{\phi_\Al^{(i)}} - \tilde{C}_i\ket{\phi_{\allbad}^{(i-1)}}\right), \qquad \ket{\phi_\bad^{(i)}} \defeq \pibadphi \left(\ket{\phi_\Al^{(i)}} - \tilde{C}_i\ket{\phi_{\allbad}^{(i-1)}}\right)
\end{equation}
where 
\begin{equation}
    \ket{\phi_{\allbad}^{(i)}} \defeq \sum_{j=1}^{i} \tilde{C}_{j,i} \ket{\phi_\bad^{(j)}}.
\end{equation}

Moreover, let $\ket{\phi_{\good}} \defeq \ket{\phi_{\good}^{(p(n))}}$ and $\ket{\phi_{\allbad}} \defeq \ket{\phi_{\allbad}^{(p(n))}}$. For each $i \in [p(n)]$, we define $\ket{\psi_{\good}^{(i)}}$, $\ket{\psi_\bad^{(i)}}$, $\ket{\psi_{\allbad}^{(i)}}$, $\ket{\psi_{\good}}$, and $\ket{\psi_{\allbad}}$ analogously (using $C$ in lieu of $\tilde C$).
\end{definition}

We now observe some properties that can be deduced from \Cref{def:piphi,def:phigood}.
We use many of these properties throughout the rest of our analysis.

\begin{lemma} \label{lem:combined}
Let $i \in [p(n)] \cup \{0\}$. Then
\begin{enumerate}
    \item \label{prop:phibadisbadandphigoodisgood} $\pibadphi\ket{\phi^{(i)}_{\bad}} = \ket{\phi^{(i)}_{\bad}}$ and $\pigoodphi\ket{\phi^{(i)}_{\good}} = \ket{\phi^{(i)}_{\good}}$,
    \item \label{prop:psibadisbadandpsigoodisgood} $\pibadpsi\ket{\psi^{(i)}_{\bad}} = \ket{\psi^{(i)}_{\bad}}$ and $\pigoodpsi\ket{\psi^{(i)}_{\good}} = \ket{\psi^{(i)}_{\good}}$,
    \item \label{prop:phigoodandphibaddisjoint} $\ket{\phi_{\good}^{(i)}}$ has disjoint support from $\ket{\phi_\bad^{(i)}}$,
    \item \label{prop:psigoodandpsibaddisjoint} $\ket{\psi_{\good}^{(i)}}$ has disjoint support from $\ket{\psi_\bad^{(i)}}$,
    \item \label{prop:equivrepofphigood} $\ket{\phi^{(i)}_{\good}} = \ket{\phi^{(i)}_{\Al}} - \ket{\phi^{(i)}_{\allbad}}$,
    \item \label{prop:equivrepofpsigood} $\ket{\psi^{(i)}_{\good}} = \ket{\psi^{(i)}_{\Al}} - \ket{\psi^{(i)}_{\allbad}}$,
    \item \label{prop:decompositionofCphigood} $\tilde{C}_i\ket{\phi^{(i-1)}_{\good}} = \ket{\phi^{(i)}_{\good}} + \ket{\phi^{(i)}_\bad}$,
    \item \label{prop:decompositionofCpsigood} $C_i\ket{\psi^{(i-1)}_{\good}} = \ket{\psi^{(i)}_{\good}} + \ket{\psi^{(i)}_\bad}$,
    \item \label{prop:expansionofpibadphi} $\pibadphi \ket{\phi^{(i)}_\Al} = \ket{\phi^{(i)}_\bad} + \pibadphi \tilde{C}_i \ket{\phi^{(i-1)}_{\allbad}}$,
    \item \label{prop:expansionofpibadpsi} $\pibadpsi \ket{\psi^{(i)}_\Al} = \ket{\psi^{(i)}_\bad} + \pibadpsi C_i \ket{\psi^{(i-1)}_{\allbad}}$,
    \item \label{prop:expansionofpigoodphi} $\pigoodphi \ket{\phi^{(i)}_\Al} = \ket{\phi^{(i)}_{\good}} + \pigoodphi \tilde{C}_i \ket{\phi^{(i-1)}_{\allbad}}$, and
    \item \label{prop:expansionofpigoodpsi} $\pigoodpsi \ket{\psi^{(i)}_\Al} = \ket{\psi^{(i)}_{\good}} + \pigoodpsi C_i \ket{\psi^{(i-1)}_{\allbad}}$.
\end{enumerate}
\end{lemma}

\begin{proof}~\nopagebreak
\begin{enumerate}
    \item By \Cref{def:phigood}, any computational basis state in the support of $\ket{\phi^{(i)}_{\bad}}$ is $\phi$-bad and any computational basis state in the support of $\ket{\phi^{(i)}_{\good}}$ is $\phi$-good. Then the desired statement follows from \Cref{def:piphi}.
    \item Similar to the proof of \cref{prop:phibadisbadandphigoodisgood}.
    \item Since $\pibadphi$ and $\pigoodphi$ are orthogonal projectors, this follows from \cref{prop:phibadisbadandphigoodisgood}.
    \item Similar to the proof of \cref{prop:phigoodandphibaddisjoint}.
\end{enumerate}

We prove the remaining parts by induction on $i$. All the statements are trivially true for $i=0$ by \Cref{def:piphi,def:phigood}. We show each of them separately for all $i \in [p(n)]$ assuming that they are true for $i-1$.
\begin{enumerate}[resume]
    \item Note that 
    \begin{align}
        \ket{\phi^{(i)}_{\good}} &= \pigoodphi \left(\ket{\phi_\Al^{(i)}} - \tilde{C}_i\ket{\phi_{\allbad}^{(i-1)}}\right) \\
        &= (I - \pibadphi) \left(\ket{\phi_\Al^{(i)}} - \tilde{C}_i\ket{\phi_{\allbad}^{(i-1)}}\right) \\
        &= \ket{\phi_\Al^{(i)}} - \tilde{C}_i\ket{\phi_{\allbad}^{(i-1)}} - \ket{\phi^{(i)}_{\bad}} \\
        &= \ket{\phi_\Al^{(i)}} - \ket{\phi_{\allbad}^{(i)}}
    \end{align}
    where we used \Cref{def:phigood} in all steps except for the second one, where we used \Cref{def:piphi}. 
    \item Similar to the proof of \cref{prop:equivrepofphigood}.
    \item Note that 
    \begin{align}
        \tilde{C}_i\ket{\phi^{(i-1)}_{\good}} &= \tilde{C}_i\left( \ket{\phi^{(i-1)}_\Al} - \ket{\phi^{(i-1)}_\allbad} \right) \\
        &= \ket{\phi^{(i)}_\Al} - \tilde{C}_i \ket{\phi^{(i-1)}_\allbad} \\
        &= \ket{\phi^{(i)}_{\good}} + \ket{\phi^{(i)}_\allbad} - \tilde{C}_i \ket{\phi^{(i-1)}_\allbad}  \\
        &= \ket{\phi^{(i)}_{\good}} + \ket{\phi^{(i)}_\bad}
    \end{align}
    where we used \cref{prop:equivrepofphigood} in steps 1 and 3, and \Cref{def:phigood} in step 4. 
    \item Similar to the proof of \cref{prop:decompositionofCphigood}.
    \item Note that 
    \begin{align}
        \pibadphi \ket{\phi^{(i)}_\Al} &=  \pibadphi\left(\ket{\phi^{(i)}_\good} + \ket{\phi^{(i)}_\allbad}\right) \\
        &= \pibadphi \ket{\phi^{(i)}_{\allbad}} \\
        &= \pibadphi \left(\ket{\phi^{(i)}_{\bad}} + \tilde{C}_i \ket{\phi^{(i-1)}_{\allbad}}\right) \\
        &= \ket{\phi^{(i)}_\bad} + \pibadphi \tilde{C}_i \ket{\phi^{(i-1)}_{\allbad}}
    \end{align}
    where we used \cref{prop:equivrepofphigood} in step 1, \cref{prop:phibadisbadandphigoodisgood} in steps 2 and 4, and \Cref{def:phigood} in step 3.
    \item Similar to the proof of \cref{prop:expansionofpibadphi}.
    \item Note that 
    \begin{align}
        \pigoodpsi \ket{\phi^{(i)}_\Al} &= \pigoodphi \left(\ket{\phi^{(i)}_{\good}} + \ket{\phi^{(i)}_{\allbad}}\right) \\
        &= \ket{\phi^{(i)}_{\good}} + \pigoodphi \ket{\phi^{(i)}_{\allbad}} \\
        &= \ket{\phi^{(i)}_{\good}} + \pigoodphi \left(\ket{\phi^{(i)}_{\bad}} + \tilde{C}_i \ket{\phi^{(i-1)}_{\allbad}}\right) \\
        &= \ket{\phi^{(i)}_{\good}} + \pigoodphi \tilde{C}_i \ket{\phi^{(i-1)}_{\allbad}} 
    \end{align}
    where we used \cref{prop:equivrepofphigood} in step 1, \cref{prop:phibadisbadandphigoodisgood} in steps 2 and 4, and \Cref{def:phigood} in step 3.
    \item Similar to the proof of \cref{prop:expansionofpigoodphi}.
\qedhere
\end{enumerate}
\end{proof}

Based on the intuitive description of $\ket{\psi_\good^{(i)}}$ and $\ket{\psi_\bad^{(i)}}$ that we provided earlier, we anticipate that the size (as quantified by the total squared norm) of the portion of the state $\ket{\psi_\Al^{(i)}}$ that never encountered the $\exit$ or a near-cycle at any point in its history, and the size of the respective portions of the state $\ket{\psi_\Al^{(i)}}$ that encountered the $\exit$ or a cycle at the $i$th or earlier steps, to sum to the size of $\ket{\psi_\Al^{(i)}}$. The following lemma formalizes this intuition.
\begin{lemma} \label{lem:psigood^2+allpsibad^2=1}
Let $i \in [p(n)] \cup \{0\}$. Then $\norm{\ket{\psi_\good^{(i)}}}^2 + \sum_{j \in [i]}\norm{\ket{\psi_\bad^{(j)}}}^2 = 1$.
\end{lemma}

\begin{proof}
We prove this claim by induction on $i$. The base case is trivial as $\norm{\ket{\psi_\good^{(0)}}} = 1$ and $\norm{\ket{\psi_\bad^{(0)}}} = 0$. Now, suppose that the claim is true for some $i$ with $i+1 \in [p(n)]$. Note that 
\begin{align}
    \norm{\ket{\psi_\good^{(i)}}}^2 &= \norm{C_i\ket{\psi_\good^{(i)}}}^2 \\
    &= \norm{\ket{\psi_\good^{(i+1)}} + \ket{\psi_\bad^{(i+1)}}}^2 \\
    &= \norm{\ket{\psi_\good^{(i+1)}}}^2 + \norm{\ket{\psi_\bad^{(i+1)}}}^2
\end{align}
where we use the fact that the unitary $C_i$ preserves the norm in the first step, and \cref{prop:decompositionofCpsigood,prop:psigoodandpsibaddisjoint} of \Cref{lem:combined} in the second and third steps, respectively.

Therefore,
\begin{align}
    \norm{\ket{\psi_\good^{(i+1)}}}^2 + \sum_{j\in [i+1]} \norm{\ket{\psi_\bad^{(j)}}}^2 &= \norm{\ket{\psi_\good^{(i)}}}^2 - \norm{\ket{\psi_\bad^{(i+1)}}}^2 + \sum_{j \in [i+1]} \norm{\ket{\psi_\bad^{(j)}}}^2 \\
    &= \norm{\ket{\psi_\good^{(i)}}}^2 + \sum_{j \in [i]} \norm{\ket{\psi_\bad^{(j)}}}^2 \\
    &= 1
\end{align}
where the last step follows by the induction hypothesis.
\end{proof}

We conclude this section by strengthening the observations made in \Cref{rem:rootednesscheck} and \Cref{lem:transcriptstateisalwaysrooted} using \Cref{lem:combined}, which we apply in \Cref{subsec:analysisongoodpart}.

\begin{lemma} \label{lem:psigoodandpsibadisrooted}
Let $i \in [p(n)] \cup \{0\}$. Then
\begin{enumerate}
    \item \label{itm:psigoodandpsibadisrooted} any computational basis state in the support of $\ket{\psi_\good^{(i)}}$ or $\ket{\psi_\bad^{(i)}}$ is rooted, and
    \item \label{itm:phigoodandphibadisaddressrooted} any computational basis state in the support of $\ket{\phi_\good^{(i)}}$ or $\ket{\phi_\bad^{(i)}}$ is address rooted.
\end{enumerate}
\end{lemma}

\begin{proof}
We show \cref{itm:psigoodandpsibadisrooted} using \Cref{rem:rootednesscheck}, while \cref{itm:phigoodandphibadisaddressrooted} follows using \Cref{lem:transcriptstateisalwaysrooted} and analogous arguments.

From \Cref{rem:rootednesscheck}, we can infer that any consecutive sequence of gates in the circuit corresponding to the given genuine, rooted algorithm will map a rooted state to a rooted state. Notice that $\ket{\psi_\good^{(0)}} = \ket{\psi_{\mathrm{initial}}}$ and $\ket{\psi_{\bad}^{(0)}} = 0$. That is, $\ket{\psi_\good^{(0)}}$ and $\ket{\psi_\bad^{(0)}}$ are rooted states by \cref{eq:genuineinitial}. Thus, since $C_i\ket{\psi^{(i-1)}_{\good}} = \ket{\psi^{(i)}_{\good}} + \ket{\psi^{(i)}_\bad}$ by \cref{prop:decompositionofCpsigood} of \Cref{lem:combined}, any computational basis state in the support of $\ket{\psi^{(i)}_{\good}} + \ket{\psi^{(i)}_\bad}$ will be rooted for each $i \in [p(n)]$. But, by \cref{prop:psigoodandpsibaddisjoint} of \Cref{lem:combined}, $\ket{\psi_{\good}^{(i)}}$ has disjoint support from $\ket{\psi_\bad^{(i)}}$, so any computational basis state in the support of $\ket{\psi_{\good}^{(i)}}$ or in the support of $\ket{\psi_{\bad}^{(i)}}$ is rooted for each $i \in [p(n)]$. 
\end{proof}

\subsection{Faithful simulation of the good part} \label{subsec:analysisongoodpart}

As any subtree of $\mathcal{G}$ without the $\exit$ vertex can be embedded in $\mathcal{T}$, one might expect that the size of the portion of the state $\ket{\psi_\Al^{(i)}}$ that never encountered the $\exit$ or a near-cycle at any point in its history is the same as the size of the portion of the state $\ket{\phi_\Al^{(i)}}$ that never encountered the $\exit$ or a near-cycle at any point in its history. We formally show this via a sequence of lemmas that culminate in \Cref{lem:||psigood||=||phigood||}. We restrict our attention to the $\mathrm{good}$ parts of the states $\ket{\psi_{\Al}^{(i)}}$ and $\ket{\phi_{\Al}^{(i)}}$ in this subsection, beginning with a useful decomposition of $\ket{\phi_{\good}^{(i)}}$.

\begin{definition}\label{def:expandincompbasis}
We define an indexed expansion of $\ket{\phi_{\good}^{(i)}}$ in the computational basis, as follows. Write $\ket{\phi_{\good}^{(i)}} = \sum_{p,q} \alpha^{(i)}_{p,q} \ket{q^{(i)}} \ket{\phi_p^{(i)}}$, where each $\ket{\phi_p^{(i)}}$ denotes a computational basis state in the vertex register, each $\ket{q^{(i)}}$ specifies a computational basis state in the workspace register, and each $\alpha^{(i)}_{p,q}$ is an amplitude. Define $\pgood^{(i)}$ to be the set of all indices $p$ appearing in the expansion of $\ket{\phi_{\good}^{(i)}}$ with any corresponding non-zero amplitude $\alpha^{(i)}_{p,q}$.  
\end{definition}

Analogous to \Cref{def:phigood},
we define computational basis states $\ket{\psi_p^{(i)}}$ from $\ket{\phi_p^{(i)}}$, and hence from $\ket{\phi_\Al^{(i)}}$ rather than $\ket{\psi_\Al^{(i)}}$.

\begin{definition} \label{def:psi_p}
For $i \in [p(n)] \cup \{0\}$ and $p  \in \pgood^{(i)}$, let $\ket{\psi_p^{(i)}} \defeq L\ket{\phi_p^{(i)}}$. 
\end{definition}

Notice that it is not immediate from this definition that $\ket{\psi_p^{(i)}}$ is in the support of the part of the state $\ket{\psi_\Al^{(i)}}$ that is in the vertex space. However, by the end of this section, we will show that indeed this is the case.

For each $i \in [p(n)] \cup \{0\}$ and $p \in \pgood^{(i)}$, the state $\ket{\phi^{(i)}_p}$ is a computational basis state in the address space, so we can write it as 
\begin{equation} \label{eq:expansionofphip}
    \ket{\phi^{(i)}_p} = \bigotimes_{j \in [p(n)]} \ket{s_j}
\end{equation}
for some strings $s_j \in \{0,1\}^{2p(n)}$.

Similarly, for each $i \in [p(n)] \cup \{0\}$ and $p \in \pgood^{(i)}$, the state $\ket{\psi^{(i)}_p}$ is a computational basis state in the vertex space, so we can write it as 
\begin{equation} \label{eq:expansionofpsip}
    \ket{\psi^{(i)}_p} = \bigotimes_{j \in [p(n)]} \ket{v_j}
\end{equation}
for some vertex labels $v_j \in \{0,1\}^{2n}$.

By the above notation and \Cref{def:psi_p}, we have that for each $j \in [p(n)]$,
\begin{equation} \label{eq:L(s)=v}
    L(s_j) = v_j.
\end{equation}

Note that $\ket{s_j}$ and $\ket{v_j}$ also depend on $p$ and $i$. However, as $p$ and $i$ will be clear from context, we keep this dependence implicit to simplify notation.

In the same vein, for each $i \in [p(n)] \cup \{0\}$ and workspace index $q$, the state $\ket{q^{(i)}}$ is a computational basis state in the vertex space, so we can write it as
\begin{equation} \label{eq:expansionofq}
    \ket{q^{(i)}} = \bigotimes_{j \in [p(n)]} \ket{w_j}.
\end{equation}
Again we suppress the dependence on $q$ and $i$ for simplicity. 

We now show that the mapping $L$, defined in \Cref{def:DefL}, is a bijection from the set of address-rooted states in $\ket{\phi_\good^{(i)}}$ to the set of address-rooted states in $L\ket{\phi_\good^{(i)}}$.

\begin{lemma} \label{lem:pbadiscycle}
Let $i \in [p(n)]$ and $p, p' \in \pgood^{(i)}$. Suppose that $\ket{\psi_p^{(i)}} = \ket{\psi_{p'}^{(i)}}$. Then $\ket{\phi_p^{(i)}} = \ket{\phi_{p'}^{(i)}}$.
\end{lemma}

\begin{proof}
Suppose, towards contradiction, that $\ket{\psi_p^{(i)}} = \ket{\psi_{p'}^{(i)}}$ but $\ket{\phi_p^{(i)}} \neq \ket{\phi_{p'}^{(i)}}$. This means that there is an index $j \in [p(n)]$ such that the string $s_j$ stored in the $j$th register of $\ket{\phi_p^{(i)}}$ is not equal to the string $s'_j$ stored in the $j$th register of $\ket{\phi_{p'}^{(i)}}$, and yet the vertex label $v_j$ stored in the $j$th register of $\ket{\psi_p^{(i)}}$ is equal to the vertex label $v'_j$ stored in the $j$th register of $\ket{\psi_{p'}^{(i)}}$. Consider the addresses $t_j = B^{\mathsf{inv}}(s_j)$ and $t'_j = B^{\mathsf{inv}}(s'_j)$. We know that $L'(t_j) = L'(t'_j) = v_j$ from \cref{eq:L(s)=v}.

By \Cref{lem:transcriptstateisintherangeofB}, we have $s_j$ and $s_k$ in the range of $B$. Recall, from \Cref{def:Bmapping}, that $B$ is a bijection and $B^{\mathsf{inv}} = B^{-1}$ on the range of $B$. Therefore, $t_j \neq t'_j$. Moreover, this means that we can write $s_j = B(t_j)$ and $s'_j = B(t'_j)$.

We know, from \Cref{lem:psigoodandpsibadisrooted}, that $\ket{\phi^{(i)}_p}$ and $\ket{\phi^{(i)}_{p'}}$ are address rooted. This means that for any ancestor $\tau \neq \zeros$ of $t_j$ in $\mathcal{T}$, $B(\tau)$ is stored in one of the registers of $\ket{\phi^{(i)}_p}$. Therefore, there is a path from the vertex labeled $\emptyaddress$ and $t_j$ in $\mathcal{T}$ such that $\ket{\phi^{(i)}_p}$ contains $B(\tau)$ for all vertices $\tau$ in this path. Let $\tau_0 = \emptyaddress, \ldots, \tau_\gamma = t_j$ denote this path where $\gamma$ is the length of this path. By \Cref{def:psi_p}, $\ket{\psi^{(i)}_p}$ contains the vertex label $L'(\tau_i)$ for each $i \in [\gamma]$. Similarly, we can deduce that there is a path $\tau'_0 = \emptyaddress, \ldots, \tau'_{\gamma'} = t'_j$ in $\mathcal{T}$, with $\gamma'$ denoting the length of this path, such that $\ket{\psi^{(i)}_p} = \ket{\psi^{(i)}_{p'}}$ contains the vertex label $L'(\tau'_i)$ for each $i \in [\gamma']$. Since $t_j \neq t'_j$ and $v_j = v'_j$, it follows that $\ket{\psi^{(i)}_p}$ contains two distinct paths from $L'(\emptyaddress) = \entrance$ to $L'(t_j) = L'(t'_j) = v_j$ in $\mathcal{G}$. Hence, $\ket{\psi^{(i)}_p}$ contains a cycle. But this is not possible since $p \in \pgood^{(i)}$. Therefore, $\ket{\phi^{(i)}_p} \neq \ket{\phi^{(i)}_{p'}} \implies \ket{\psi^{(i)}_p} \neq \ket{\psi^{(i)}_{p'}}$.
\end{proof}

The following lemma states that there is a bijective correspondence between the contents of the vertex registers of $\ket{\psi_\good^{(i)}}$ and the address registers of $\ket{\phi_\good^{(i)}}$.

\begin{lemma} \label{lem:addressvertexmapping}
Let $i \in [p(n)] \cup \{0\}$, $j, k \in [p(n)]$, $p \in \pgood^{(i)}$ and $c \in \mathcal{C}$. Let $\ket{v_j}$ and $\ket{v_{k}}$ be the states stored in the $j$th and $k$th registers of $\ket{\psi_p^{(i)}}$, respectively, as in \cref{eq:expansionofpsip}. Similarly, let $\ket{s_j}$ and $\ket{s_{k}}$ be the states stored in the $j$th and $k$th registers of $\ket{\phi_p^{(i)}}$, respectively, as in \cref{eq:expansionofphip}. Then
\begin{enumerate}
    \item \label{prop:v=0iffs=0} $v_j = 0^{2n} \iff s_j = 0^{2p(n)}$,
    \item \label{prop:v=noedgeiffs=noedge} $v_j = \noedge \iff s_j = \noedgestring$,
    \item \label{prop:v=invalidiffs=invalid} $v_j = \invalid \iff s_j = \invalidstring$,
    \item \label{prop:v=v'iffs=s'} $v_j = v_{k} \iff s_j = s_{k}$, and
    \item \label{prop:v=eta_c(v')iffs=lambda_c(s')} $v_j = \eta_c(v_{k}) \iff s_j = B\lambda_cB^{\mathsf{inv}}(s_{k})$.
\end{enumerate}
\end{lemma}

\begin{proof}
We show each statement separately as follows.
\begin{enumerate}
\item 
First, suppose that $v_j = 0^{2n}$. By \cref{eq:L(s)=v}, we know that $L(s_j)=0^{2n}$. Let $t_j=B^{\mathsf{inv}}(s_j)$. Then, $L'(t_j)=0^{2n}$. It can be observed from \Cref{alg:lprime} that $t_j=\zeros$ as only in that case  $L'(t_j)=0^{2n}$. Thus, by \Cref{def:Bmapping}, it must be that $s_j = 0^{2p(n)}$ as it is the only value of $s_j$ for which $B^{\mathsf{inv}}(s_j)=\zeros$.

Now, suppose that $s_j = 0^{2p(n)}$. Then
\begin{equation}
    v_j = L(s_j) = L'B^{\mathsf{inv}}(0^{2p(n)}) = L'(\zeros) = 0^{2n}
\end{equation}
where the first step follows from \cref{eq:L(s)=v}, the third follows from \Cref{def:Bmapping}, and the fourth from \Cref{alg:lprime}. 

\item 
First, suppose that $v_j=\noedge$. By \cref{eq:L(s)=v}, we know that $L(s_j)=\noedge$. Let $t_j=B^{\mathsf{inv}}(s_j)$. Then $L'(t_j)=\noedge$. We claim that $t_j = \noedgeaddress$. Indeed, if $t_j \neq \noedgeaddress$, then, as $L'(\noedgeaddress) = L'(t_j) = \noedge$, we would have found the $\exit$ or a cycle in $\mathcal{G}$ by \Cref{lem:addresstreecycle}. But that is not possible since $p \in \pgood^{(i)}$, so $t_j = \noedgeaddress$. It follows that $s_j = \noedgestring$ since it is the only possible assignment of $s_j$ such that $B^{\mathsf{inv}}(s_j)=\noedgeaddress$ by \Cref{def:Bmapping}. 

Now, suppose that $s_j = \noedgestring$. Then
\begin{equation}
    v_j = L(s_j) = L'B^{\mathsf{inv}}(\noedgestring) = L'(\noedgeaddress) = \noedge
\end{equation}
where the first step follows from \cref{eq:L(s)=v}, the third follows from \Cref{def:Bmapping}, and the fourth from \Cref{alg:lprime}. 

\item 

First, suppose that $v_j=\invalid$. By \cref{eq:L(s)=v}, we know that $L(s_j)=\invalid$. Let $t_j=B^{\mathsf{inv}}(s_j)$. Then $L'(t_j)=\invalid$. We claim that $t_j = \invalidaddress$. Indeed, if $t_j \neq \invalidaddress$, then, as $L'(\invalidaddress) = L'(t_j) = \invalid$, we would have found the $\exit$ or a cycle in $\mathcal{G}$ by \Cref{lem:addresstreecycle}. But that is not possible since $p \in \pgood^{(i)}$, so $t_j = \invalidaddress$. Recall from \Cref{lem:transcriptstateisintherangeofB} that $s_j$ is in the range of the $B$ mapping. It follows that $s_j = \noedgestring$ since this is the only possible assignment of $s_j$ in the range of $B$ with $B^{\mathsf{inv}}(s_j)=\invalidaddress$.

Now, suppose that $s_j = \invalidstring$. Then
\begin{equation}
    v_j = L(s_j) = L'B^{\mathsf{inv}}(\invalidstring) = L'(\invalidaddress) = \invalid
\end{equation}
where the first step follows from     \cref{eq:L(s)=v}, the third follows from \Cref{def:Bmapping}, and the fourth from \Cref{alg:lprime}. 

\item
First, suppose that $v_j = v_{k}$. From \cref{eq:L(s)=v}, we can deduce that $L(s_j) = L(s_k)$. Let $t_j = B^{\mathsf{inv}}(s_j)$ and $t_{k} = B^{\mathsf{inv}}(s_{k})$. Then $L'(t_j) = L'(t_{k})$. If $t_j \neq t_{k}$, then, by \Cref{lem:addresstreecycle}, the concatenation of the paths specified by the sequence of colors $t_j$ and $t_{k}$ forms a cycle, which contradicts $p \in \pgood^{(i)}$. This means that $t_j = t_{k}$. Recall from \Cref{lem:transcriptstateisintherangeofB} that $s_j$ and $s_{k}$ are in the range of $B$. We know, from \Cref{def:Bmapping}, that $B$ is a bijection and $B^{\mathsf{inv}} = B^{-1}$ on the range of $B$. Therefore, $s_j = s_{k}$.

Now suppose that $s_j = s_{k}$. Thus, $L(s_j) = L(s_{k})$. By \cref{eq:L(s)=v}, we know that $L(s_j) = v_j$ and $L(s_{k}) = v_{k}$. Therefore, $v_j = v_{k}$.

\item 
First, suppose that $v_j = \eta_c(v_{k})$. From \cref{eq:L(s)=v}, we have $v_j = L(s_j)$ and $v_{k} = L(s_{k})$. Let $t_j = B^{\mathsf{inv}}(s_j)$ and $t_{k} = B^{\mathsf{inv}}(s_{k})$. Then $v_j = L'(t_j)$ and $v_{k} = L'(t_{k})$, which means that $t_j$ and $t_{k}$ are addresses of the vertices $v_j$ and $v_{k}$ respectively. By \Cref{lem:L'lambda=etaL'}, it follows that $\eta_c(v_{k}) = L'\lambda_c(t_{k})$. Altogether, we have $L'(t_j) = L'\lambda_c(t_{k})$. If $t_j \neq \lambda_c(t_{k})$, then by \Cref{lem:addresstreecycle}, the concatenation of the paths specified by the sequence of colors $t_j$ and $\lambda_c(t_{k})$ forms a cycle, which contradicts $p \in \pgood^{(i)}$. This means that $t_j = \lambda_c(t_{k})$. It follows that $BB^{\mathsf{inv}}(s_j) = B(t_j) = B\lambda_c(t_{k})$. Recall that $s_j$ is in the range of $B$ from \Cref{lem:transcriptstateisintherangeofB}. Since $B$ is a bijection and $B^{\mathsf{inv}} = B^{-1}$ on the range of $B$ from \Cref{def:Bmapping}, we can deduce that $BB^{\mathsf{inv}}(s_j) = s_j$. Hence, $s_j = B\lambda_cB^{\mathsf{inv}}(s_{k})$.

Now suppose that $s_j = B\lambda_cB^{\mathsf{inv}}(s_{k})$. Then $L(s_j) = LB\lambda_cB^{\mathsf{inv}}(s_{k}) = LB\lambda_c(t_{k})$. By \cref{eq:L(s)=v} we have $v_j = LB\lambda_c(t_{k})$. Note that $\lambda_cB^{\mathsf{inv}}(s_{k}) = \lambda_c(t_{k})$ is in the domain of $B$ since $B^{\mathsf{inv}}$ maps any string to the domain of $B$ and $\lambda_c$ preserves the domain of $B$ by \Cref{def:addresstree,def:Bmapping}. Therefore, as $B$ is a bijection and $B^{\mathsf{inv}} = B^{-1}$ on the range of $B$, $LB\lambda_c(t_{k}) = L'\lambda_c(t_{k})$. Since $t_{k}$ is an address of $v_{k}$, we have $L'\lambda_c(t_{k}) = \eta_c(v_{k})$ by \Cref{lem:L'lambda=etaL'}. Altogether, we have $LB\lambda_c(t_{k}) = \eta_c(v_{k})$. By the deductions we made above, we can conclude that $v_j = \eta_c(v_{k})$.
\qedhere
\end{enumerate}
\end{proof}

The next lemma forms a key ingredient of \Cref{lem:psigood=L(phigood)}, where we essentially show that the mapping $L$ and the gate $C_i$ commute: applying $L$ followed by $C_i$ is equivalent to applying $\tilde{C}_i$ followed by $L$. 

\begin{lemma} \label{lem:LC(phi)=C(psi)}
Let $i \in [p(n)]$, $p \in \pgood^{(i)}$ and $q$ be any workspace index. Then $L \tilde{C}_i \ket{q^{(i-1)}} \ket{\phi_p^{(i-1)}} = C_i \ket{q^{(i-1)}} \ket{\psi_p^{(i-1)}}$.
\end{lemma}

\begin{proof}
We prove the statement of the lemma for each of the possible gates in our gate set defined in \Cref{def:genuinecircuit}. For any quantum state $\ket{\chi} = \bigotimes_{j} \ket{\chi_j}$, and any indices $j_1, \ldots, j_r$, let $\ket{\chi}_{j_1, \ldots, j_r} \defeq \ket{\chi_{j_1}} \otimes \cdots \otimes \ket{\chi_{j_r}}$. 

\begin{enumerate}
\item \label{prop:oraclecase} $C_i = \controlled{O_c}$ for some $c \in \mathcal{C}$. Let $x, y \in [p(n)]$ denote the vertex register indices that $O_c$ acts on, and let $a$ denote the workspace register index of the control qubit. Recall, from \cref{eq:expansionofphip,eq:expansionofpsip}, that $v_x$ and $s_x$ denote the contents of the $x$th vertex register of $\ket{\psi_p^{(i-1)}}$ and the $x$th address register of $\ket{\phi_p^{(i-1)}}$, respectively. Note that $t_x = B^{\mathsf{inv}}(s_x)$ is an address of $v_x$ as $L'(t_x)=v_x$ by \cref{eq:L(s)=v}. Then
\begin{align} \label{eq:Llambda(tx)=eta(vx)}
    LB\lambda_c B^{\mathsf{inv}}(s_x) = L'B^{\mathsf{inv}}B\lambda_c(t_x) = L'\lambda_c(t_x) = \eta_c(v_x)
\end{align}
where the first equality follows from \Cref{def:DefL}, the second equality follows since $B^{\mathsf{inv}}$ is the inverse of $B$ over the range of $B$ in \Cref{def:Bmapping}, and the third follows from \Cref{lem:L'lambda=etaL'}.

Note that $\ket{s_x}$ and $\ket{s_y}$ denote the states of the control and target registers of the oracle gate $\tilde{O}_c$. We argued in the proof of \Cref{lem:transcriptstateisintherangeofB} that $\controlled{\tilde{O}_c}$ is only applied when $s_y = 0^{2p(n)}$ or $s_y = B\lambda_{c}B^{\mathsf{inv}}(s_x)$. Therefore, 
\begin{align}
    L\controlled{\tilde{O}_c} \ket{\phi^{(i-1)}_p}_{x,y} \ket{q^{(i-1)}}_a &= L \controlled{\tilde{O}_c} \ket{s_x} \ket{s_y} \ket{w_a} \\
    &= L 
    \begin{cases}
        \ket{s_x} \ket{s_y \oplus B\lambda_c B^{\mathsf{inv}}(s_x)} \ket{w_a} & s_y \in \left\{0^{2p(n)}, B\lambda_{c}B^{\mathsf{inv}}(s_x)\right\} \; \& \; w_a=1 \\
        \ket{s_x} \ket{s_y} \ket{w_a} & \text{otherwise}
    \end{cases} \\
    &= L
    \begin{cases}
        \ket{s_x} \ket{B\lambda_c B^{\mathsf{inv}}(s_x)} \ket{w_a} & s_y = 0^{2p(n)} \; \& \; w_a = 1 \\
        \ket{s_x} \ket{0^{2p(n)}} \ket{w_a} & s_y = B\lambda_{c}B^{\mathsf{inv}}(s_x) \; \& \; w_a = 1 \\
        \ket{s_x} \ket{s_y} \ket{w_a} & \text{otherwise}
	\end{cases} \\
    &=
	\begin{cases}
        \ket{v_x} \ket{\eta_c(v_x)} \ket{w_a} & v_y = 0^{2n} \; \& \; w_a = 1 \\
        \ket{v_x} \ket{0^{2n}} \ket{w_a} & v_y = \eta_c(v_x) \; \& \; w_a = 1 \\
        \ket{v_x} \ket{v_y} \ket{w_a} & \text{otherwise}
	\end{cases} \\
	&= \ket{v_x} \ket{v_y \oplus w_a \cdot \eta_c(v_x)} \ket{w_a} \\
	&= \controlled{O_c} \ket{v_x} \ket{v_y} \ket{w_a} \\
	&= \controlled{O_c} \ket{\psi^{(i-1)}_p}_{x,y} \ket{q^{(i-1)}}_a
\end{align}
where the second equality follows from \Cref{def:transcript} and an observation made above; the fourth from \cref{eq:L(s)=v}, \cref{prop:v=0iffs=0,prop:v=eta_c(v')iffs=lambda_c(s')} of \Cref{lem:addressvertexmapping}, and \cref{eq:Llambda(tx)=eta(vx)}; and the sixth from \Cref{def:genuinecircuit}.

\item $C_i = \controlled{e^{i \theta T}}$ for some $\theta \in [0,2\pi)$. Let $x, y \in [p(n)]$ denote the vertex register indices that $e^{i \theta T}$ acts on, and let $a$ denote the workspace register index of the control qubit. Then
\begin{align}
    L \controlled{e^{i \theta \tilde{T}}} \ket{\phi^{(i-1)}_p}_{x,y} \ket{q^{(i-1)}}_a  &= L \controlled{e^{i \theta \tilde{T}}} \ket{s_x} \ket{s_y} \ket{w_a} \\
    &= L 
    \begin{cases}
        \cos \theta \ket{s_x} \ket{s_y} \ket{w_a} + i \sin \theta \ket{s_y} \ket{s_x} \ket{w_a} & w_a = 1 \\
        \ket{s_x} \ket{s_y} \ket{w_a} & \text{otherwise} 
    \end{cases} \\
    &= 
    \begin{cases}
        \cos \theta \ket{v_x} \ket{v_y} \ket{w_a} + i \sin \theta \ket{v_y} \ket{v_x} \ket{w_a} & w_a = 1 \\
        \ket{v_x} \ket{v_y} \ket{w_a} & \text{otherwise}  
    \end{cases} \\
    &= \controlled{e^{i \theta T}} \ket{v_x} \ket{v_y} \ket{w_a} \\
    &= \controlled{e^{i \theta T}} \ket{\psi^{(i-1)}_p}_{x, y} \ket{q^{(i-1)}}_a
\end{align}
where the second equality follows from \Cref{def:transcript}, the third from \cref{eq:L(s)=v}, and the fourth from \Cref{def:genuinecircuit}.

\item $C_i = \mathcal{E}$. Let $x, y \in [p(n)]$ denote the vertex register indices and $a$ denote the workspace register index that $\mathcal{E}$ acts on. Then
\begin{align}
    L\tilde{\mathcal{E}} \ket{\phi^{(i-1)}_p}_{x,y} \ket{q^{(i-1)}}_a &= L\tilde{\mathcal{E}}\ket{s_x} \ket{s_y} \ket{w_a} \\
    &= L\left(\ket{s_x} \ket{s_y}\right) \ket{w_a \oplus \delta[s_x = s_y]} \\ 
    &= \ket{v_x} \ket{v_y} \ket{w_a \oplus \delta[v_x = v_y]} \\
    &= \mathcal{E} \ket{v_x} \ket{v_y} \ket{w_a} \\
    &= \mathcal{E}\ket{\psi^{(i-1)}_p}_{x,y} \ket{q^{(i-1)}}_a
\end{align}
where the second equality follows from \Cref{def:transcript}, the third from \cref{eq:L(s)=v} and \cref{prop:v=v'iffs=s'} of \Cref{lem:addressvertexmapping}, and the fourth from \Cref{def:genuinecircuit}.

\item $C_i = \mathcal{N}$. Let $x \in [p(n)]$ denote the vertex register index and $a$ denote the workspace register index that $\mathcal{N}$ acts on. Then
\begin{align}
    L\tilde{\mathcal{N}} \ket{\phi^{(i-1)}_p}_x \left.\ket{q^{(i-1)}}\right|_a &= L\tilde{\mathcal{N}}\ket{s_x} \ket{w_a} \\
    &= L\left(\ket{s_x}\right) \ket{w_a \oplus \delta[s_x = \noedgestring]} \\ 
    &= \ket{v_x} \ket{w_a \oplus \delta[v_x = \noedge]} \\
    &= \mathcal{N} \ket{v_x} \ket{w_a} \\
    &= \mathcal{N}\ket{\psi^{(i-1)}_p}_x \ket{q^{(i-1)}}_a
\end{align}
where the second equality follows from \Cref{def:transcript}, the third from \cref{eq:L(s)=v} and \cref{prop:v=noedgeiffs=noedge} of \Cref{lem:addressvertexmapping}, and the fourth from \Cref{def:genuinecircuit}.

\item $C_i = \mathcal{Z}$. Let $x \in [p(n)]$ denote the vertex register index and $a$ denote the workspace register index that $\mathcal{N}$ acts on. Then
\begin{align}
    L\tilde{\mathcal{Z}} \ket{\phi^{(i-1)}_p}_x \left.\ket{q^{(i-1)}}\right|_a &= L\tilde{\mathcal{N}}\ket{s_x} \ket{w_a} \\
    &= L\left(\ket{s_x}\right) \ket{w_a \oplus \delta[s_x = 0^{2p(n)}]} \\ 
    &= \ket{v_x} \ket{w_a \oplus \delta[v_x = 0^{2n}]} \\
    &= \mathcal{Z} \ket{v_x} \ket{w_a} \\
    &= \mathcal{Z}\ket{\psi^{(i-1)}_p}_x \ket{q^{(i-1)}}_a
\end{align}
where the second equality follows from \Cref{def:transcript}, the third from \cref{eq:L(s)=v} and \cref{prop:v=0iffs=0} of \Cref{lem:addressvertexmapping}, and the fourth from \Cref{def:genuinecircuit}.

\item $C_i$ is a gate on the workspace register. Since $L$ acts on the address space, $L$ and $C_i$ commute. Moreover, $C_i = \tilde{C_i}$ as we do not replace the gates acting on the workspace register in \Cref{def:transcript}. Thus,
\begin{equation}
    L\tilde{C_i} \ket{\phi^{(i-1)}_p} \ket{q^{(i-1)}} = \tilde{C_i} L\ket{\phi^{(i-1)}_p} \ket{q^{(i-1)}} = C_i \ket{\psi^{(i-1)}_p} \ket{q^{(i-1)}}.
\end{equation}

\end{enumerate}
\end{proof}

Notice that the non-oracle gates in \Cref{def:genuinecircuit} do not produce any `new information' about vertex labels. Based on this intuition, one might expect that the portion of $\ket{\psi^{(i-1)}_\Al}$ (respectively $\ket{\phi^{(i-1)}_\Al}$) that has never encountered the $\exit$ or a cycle will not encounter the $\exit$ or a cycle on the application of $C_i$ (respectively $\tilde{C}_i$) at the $i$th step. We formalize this as follows.

\begin{lemma} \label{lem:G(phigood(i))=phigood(i+1)}
Let $i \in [p(n)]$ and suppose that $C_i$ is a genuine non-oracle gate. Then
\begin{enumerate}
\item \label{enum:G(phigood(i))=phigood(i+1)} $\ket{\phi_{\good}^{(i)}} = \tilde{C}_i\ket{\phi_{\good}^{(i-1)}}$ and
\item \label{enum:G(psigood(i))=psigood(i+1)} $\ket{\psi_{\good}^{(i)}} = C_i\ket{\psi_{\good}^{(i-1)}}$.
\end{enumerate} 
\end{lemma}

\begin{proof}
We show \cref{enum:G(phigood(i))=phigood(i+1)} using \cref{prop:decompositionofCphigood} of \Cref{lem:combined}. \Cref{enum:G(psigood(i))=psigood(i+1)} follows by an analogous argument that instead uses \cref{prop:decompositionofCpsigood} of \Cref{lem:combined}.

Notice that in \Cref{def:genuinecircuit}, the only gates that alter the vertex space are $O_c$ gates or $e^{i \theta T}$ gates. But the $e^{i \theta T}$ gates only swap contents of the vertex register (without computing a new vertex label in a vertex register). In other words, genuine non-oracle gates do not introduce new vertex labels. This means that, as with $\ket{\phi^{(i-1)}_{\good}}$, the subgraph corresponding to any computational basis state in the support of $\tilde{C}_i\ket{\phi^{(i)}_{\good}}$ does not contain the $\exit$ and is more than one edge away from containing a cycle. That is, $\ket{\phi^{(i)}_{\bad}} = 0$. Therefore, by \cref{prop:decompositionofCphigood} of \Cref{lem:combined}, we have $\tilde{C}_i \ket{\phi^{(i-1)}_{\good}} = \ket{\phi^{(i)}_{\good}} + \ket{\phi^{(i)}_{\bad}} = \ket{\phi^{(i)}_{\good}}$. 
\end{proof}

For the analysis of oracle gates, we now define a subset of $\pgood^{(i-1)}$ that contains indices corresponding to computational basis states in the address space that do not contain the $\exit$ or a cycle even after the application of an oracle gate at the $i$th step. Inspired by the decomposition in \Cref{def:expandincompbasis}, we then define the components $\ket{\phi_\great^{(i-1)}}$ and $\ket{\psi_\great^{(i-1)}}$ of $\ket{\phi_\Al^{(i-1)}}$ and $\ket{\psi_\Al^{(i-1)}}$, respectively.

\begin{definition} \label{def:pgreat}
Let $i \in [p(n)]$. Suppose that $C_i = \controlled{O_c}$ for some $c \in \mathcal{C}$. Then, define 
\begin{equation}
    \pgreat^{(i-1)} \defeq \left\{p \in \pgood^{(i-1)}: \exists p' \in \pgood^{(i)} \text{ such that } \tilde{C}_i\ket{\phi^{(i-1)}_p} = \ket{\phi^{(i)}_{p'}}\right\}.
\end{equation} 
Also, let
\begin{equation}
    \ket{\phi_\great^{(i-1)}} \defeq \sum_{p \in \pgreat^{(i-1)}} \sum_q \alpha^{(i-1)}_{p, q} \ket{q^{(i-1)}} \ket{\phi^{(i-1)}_p}, \qquad \ket{\psi_\great^{(i-1)}} \defeq \sum_{p \in \pgreat^{(i-1)}} \sum_q \alpha^{(i-1)}_{p, q} \ket{q^{(i-1)}} \ket{\psi^{(i-1)}_p}.
\end{equation}
\end{definition}

In the following lemma, we show that the $L$ mapping preserves the relationship between computational basis states in the support of $\ket{\phi^{(i-1)}_\good}$ and $\ket{\phi^{(i)}_\good}$: applying the oracle gate to a computational basis state in the support of $\ket{\phi^{(i-1)}_\good}$ results in a computational basis state in the support of $\ket{\phi^{(i)}_\good}$ exactly when applying the oracle gate to a computational basis state in the support of $L\ket{\phi^{(i-1)}_\good}$ results in a computational basis state in the support of $L\ket{\phi^{(i)}_\good}$. 

\begin{lemma} \label{lem:alternativeconditionforpgreat}
Let $i \in [p(n)]$, $p \in \pgood^{(i-1)}$ and $p' \in \pgood^{(i)}$. Suppose that $C_i = \controlled{O_c}$ for some $c \in \mathcal{C}$. Then $\tilde{C}_i\ket{\phi^{(i-1)}_p} = \ket{\phi^{(i)}_{p'}}$ if and only if $C_i\ket{\psi^{(i-1)}_p} = \ket{\psi^{(i)}_{p'}}$.
\end{lemma}

\begin{proof}
First, suppose that $\tilde{C}_i\ket{\phi^{(i-1)}_p} = \ket{\phi^{(i)}_{p'}}$. Then
\begin{align}
    C_i \ket{\psi^{(i-1)}_p} = L \tilde{C}_i \ket{\phi^{(i-1)}_p} = L \ket{\phi^{(i)}_{p'}} = \ket{\psi^{(i)}_{p'}}
\end{align}
where we used \cref{prop:oraclecase} of \Cref{lem:LC(phi)=C(psi)} in the first step, the above supposition in the second, and \Cref{def:psi_p} in the last step.

Now, suppose that $C_i\ket{\psi^{(i-1)}_p} = \ket{\psi^{(i)}_{p'}}$. That is, $L\tilde{C}_i\ket{\phi^{(i-1)}_p} = L\ket{\phi^{(i)}_{p'}}$ by \cref{prop:oraclecase} of \Cref{lem:LC(phi)=C(psi)} and \Cref{def:psi_p}. From \cref{prop:decompositionofCphigood} of \Cref{lem:combined}, we know that $\tilde{C}_i\ket{\phi^{(i-1)}_{\good}} = \ket{\phi^{(i)}_{\good}} + \ket{\phi^{(i)}_\bad}$. Thus, since $\ket{\phi^{(i-1)}_p}$ is in the support of $\ket{\phi^{(i-1)}_\good}$ and $\tilde{C}_i = \controlled{\tilde{O}_c}$ is a permutation of the computational basis states, $\tilde{C}_i\ket{\phi^{(i-1)}_p}$ is in the support of $\ket{\phi^{(i)}_{\good}} + \ket{\phi^{(i)}_\bad}$. But by \cref{prop:phigoodandphibaddisjoint} of \Cref{lem:combined}, the support of $\ket{\phi^{(i)}_{\good}}$ and $\ket{\phi^{(i)}_{\bad}}$ is disjoint, so $\tilde{C}_i\ket{\phi^{(i-1)}_p}$ is either in the support of $\ket{\phi^{(i)}_{\good}}$ or $\ket{\phi^{(i)}_{\bad}}$. 

If $\tilde{C}_i\ket{\phi^{(i-1)}_p}$ is in the support of $\ket{\phi^{(i)}_{\bad}}$, then $\pibadphi\tilde{C}_i\ket{\phi^{(i-1)}_p} = \tilde{C}_i\ket{\phi^{(i-1)}_p}$ from \cref{prop:phibadisbadandphigoodisgood} of \Cref{lem:combined}. In other words, by \Cref{def:piphi} $\tilde{C}_i\ket{\phi^{(i-1)}_p}$ is a $\phi$-$\mathrm{bad}$ state. On the other hand, since $\ket{\phi^{(i)}_{p'}}$ is in the support of $\ket{\phi^{(i)}_\good}$, it is a $\phi$-$\mathrm{good}$ state by \cref{prop:phibadisbadandphigoodisgood} of \Cref{lem:combined} and \Cref{def:piphi}. But then it is not possible for $L \tilde{C}_i\ket{\phi^{(i-1)}_p}$ to be equal to $L\ket{\phi^{(i)}_{p'}}$ by \Cref{def:phi-badandphi-good}. It follows that $\tilde{C}_i\ket{\phi^{(i-1)}_p}$ is in the support of $\ket{\phi^{(i)}_{\good}}$. This means that there is a $p'' \in \pgood^{(i)}$ such that $\ket{\phi^{(i)}_{p''}} = \tilde{C}_i\ket{\phi^{(i-1)}_p}$ so $\ket{\psi^{(i)}_{p''}} = L\tilde{C}_i\ket{\phi^{(i-1)}_p} = L\ket{\phi^{(i)}_{p'}} = \ket{\psi^{(i)}_{p'}}$. Therefore, by \Cref{lem:pbadiscycle}, we have $\ket{\phi^{(i)}_{p'}} = \ket{\phi^{(i)}_{p''}} = \tilde{C}_i\ket{\phi^{(i-1)}_p}$, which proves the desired result.
\end{proof}

\Cref{def:pgreat} and \Cref{lem:alternativeconditionforpgreat} give rise to the following alternative definition of $\pgreat^{(i-1)}$.

\begin{corollary} \label{cor:alternatedefinitionofpgreat}
Let $i \in p[n]$. Suppose that $C_i = \controlled{O_c}$ for some $c \in \mathcal{C}$. Then 
\begin{equation}
    \pgreat^{(i-1)} = \left\{p \in \pgood^{(i-1)}: \exists p' \in \pgood^{(i)} \text{ such that } C_i\ket{\psi^{(i-1)}_p} = \ket{\psi^{(i)}_{p'}}\right\}.
\end{equation}
\end{corollary}

The next lemma provides a necessary and sufficient condition for membership in $\pgreat^{(i-1)}$: for any index $p \in \pgood^{(i-1)}$, we have $p \in \pgreat^{(i-1)}$ exactly when the oracle gate at the $i$th step applied to the computational basis state associated with $p$ in the address space results in a $\phi$-$\mathrm{good}$ state.

\begin{lemma} \label{lem:pgreatphiimplication}
Let $i \in [p(n)]$, $p \in \pgood^{(i-1)}$. Suppose that $C_i = \controlled{O_c}$ for some $c \in \mathcal{C}$. Then
\begin{enumerate}
\item \label{enum:pingreatphiimplication}
$p \in \pgreat^{(i-1)}$ implies $\pigoodphi \tilde{C}_i \ket{\phi^{(i-1)}_p} = \tilde{C}_i \ket{\phi^{(i-1)}_p}$.
\item \label{enum:pnotingreatphiimplication}
$p \notin \pgreat^{(i-1)}$ implies $\pigoodphi \tilde{C}_i \ket{\phi^{(i-1)}_p} = 0$. 
\end{enumerate}
\end{lemma}

\begin{proof}
We separately show \cref{enum:pingreatphiimplication,enum:pnotingreatphiimplication} using converse arguments.

\begin{enumerate}
\item 
Suppose that $p \in \pgreat^{(i-1)}$. Then, by \Cref{def:pgreat}, there is some $p' \in \pgood^{(i)}$ such that $\tilde{C}_i\ket{\phi^{(i-1)}_p} = \ket{\phi^{(i)}_{p'}}$. 
Thus, the computational basis state $\tilde{C}_i\ket{\phi^{(i-1)}_p}$ is in the support of $\ket{\phi^{(i)}_{\good}}$ by \Cref{def:expandincompbasis}. Therefore, by
\cref{prop:phibadisbadandphigoodisgood} of \Cref{lem:combined}, $\pigoodphi \tilde{C}_i\ket{\phi^{(i-1)}_p} = \tilde{C}_i\ket{\phi^{(i-1)}_p}$.

\item
Suppose that $p \not \in \pgreat^{(i-1)}$. Then, by \Cref{def:pgreat}, $\tilde{C}_i\ket{\phi^{(i-1)}_p} \neq \ket{\phi^{(i)}_{p'}}$ for all $p' \in \pgood^{(i)}$. Since $\ket{\phi^{(i)}_{p'}}$ and $\tilde{C}_i\ket{\phi^{(i-1)}_p}$ are computational basis states (in the vertex space), it means that $\bra{\phi^{(i-1)}_p}\tilde{C}^\dagger_i \ket{\phi^{(i)}_{p'}} = 0$ for all $p' \in \pgood^{(i)}$. By \Cref{def:expandincompbasis}, we know that $\ket{\phi^{(i)}_{\good}}$ is supported only on states in  $\left\{\ket{\phi^{(i)}_{p'}}: p' \in \pgood^{(i)}\right\}$. Thus, $\tilde{C}_i \ket{\phi^{(i-1)}_p}$ is not in the support of $\ket{\phi_\good^{(i)}}$. But since $\tilde{C}_i$ is a permutation of computational basis states and $p \in \pgood^{(i-1)}$, it follows that $\tilde{C}_i \ket{\phi^{(i-1)}_p}$ is in the support of $\tilde{C}_i \ket{\phi_\good^{(i-1)}}$, which is equal to $\ket{\phi_\good^{(i)}} + \ket{\phi_\bad^{(i)}}$ by \cref{prop:decompositionofCphigood} of \Cref{lem:combined}. This means that, by \cref{prop:phigoodandphibaddisjoint} of \Cref{lem:combined}, $\tilde{C}_i \ket{\phi^{(i-1)}_p}$ is in the support of $\ket{\phi_\bad^{(i)}}$. Therefore, by
\cref{prop:phibadisbadandphigoodisgood} of \Cref{lem:combined}, $\pibadphi \tilde{C}_i \ket{\phi^{(i-1)}_p} = \tilde{C}_i \ket{\phi^{(i-1)}_p}$. Since $\pigoodphi$ and $\pibadphi$ are orthogonal projectors, we conclude that $\pigoodphi \tilde{C}_i \ket{\phi^{(i-1)}_p} = 0$. 
\qedhere
\end{enumerate}
\end{proof}

Based on \Cref{lem:pgreatphiimplication}, we derive another condition for membership in $\pgreat$, which is analogous to the one in \Cref{lem:pgreatphiimplication}. These conditions help us establish \Cref{lem:O(phigreat(i))=phigood(i+1),lem:O(psigreat(i))=psigood(i+1)}, respectively.

\begin{lemma} \label{lem:pgreatpsiimplication}
Let $i \in [p(n)]$, $p \in \pgood^{(i-1)}$. Suppose that $C_i = \controlled{O_c}$ for some $c \in \mathcal{C}$. Then
\begin{enumerate}
\item \label{enum:pingreatpsiimplication}
$p \in \pgreat^{(i-1)}$ implies $\pigoodpsi C_i \ket{\psi^{(i-1)}_p} = C_i \ket{\psi^{(i-1)}_p}$.
\item \label{enum:pnotingreatpsiimplication}
$p \notin \pgreat^{(i-1)}$ implies $\pigoodpsi C_i \ket{\psi^{(i-1)}_p} = 0$.
\end{enumerate}
\end{lemma}

\begin{proof}
We prove \cref{enum:pingreatpsiimplication,enum:pnotingreatpsiimplication} using  \cref{enum:pingreatphiimplication,enum:pnotingreatphiimplication} of \Cref{lem:pgreatphiimplication}, respectively.

\begin{enumerate}
\item 
Suppose that $p \in \pgreat^{(i-1)}$. Note that $L\tilde{C}_i \ket{\phi^{(i-1)}_p} = C_i \ket{\psi^{(i-1)}_p}$ by \Cref{lem:LC(phi)=C(psi)}. We know, by \cref{enum:pingreatphiimplication} of \Cref{lem:pgreatphiimplication}, that $\pigoodphi \tilde{C}_i \ket{\phi^{(i-1)}_p} = \tilde{C}_i \ket{\phi^{(i-1)}_p}$. Thus, by \Cref{def:phi-badandphi-good}, $\tilde{C}_i \ket{\phi^{(i-1)}_p}$ is a $\phi$-$\mathrm{good}$ state, which means that $C_i \ket{\psi^{(i-1)}_p}$ is a $\psi$-$\mathrm{good}$ state. By \Cref{def:piphi}, the desired statement follows. 

\item Suppose that $p \notin \pgreat^{(i-1)}$. We know, by \cref{enum:pnotingreatphiimplication} of \Cref{lem:pgreatphiimplication}, that $\pigoodphi \tilde{C}_i \ket{\phi^{(i-1)}_p} = 0$. Since $\pigoodphi$ and $\pibadphi$ are orthogonal projectors by \Cref{def:piphi}, $\pibadphi \tilde{C}_i \ket{\phi^{(i-1)}_p} = \tilde{C}_i \ket{\phi^{(i-1)}_p}$. Similar to \cref{enum:pingreatpsiimplication}, we can deduce that $\tilde{C}_i \ket{\phi^{(i-1)}_p}$ is a $\phi$-$\mathrm{bad}$ state and $L\tilde{C}_i \ket{\phi^{(i-1)}_p} = C_i \ket{\psi^{(i-1)}_p}$ is a $\psi$-$\mathrm{bad}$ state. Hence, by \Cref{def:piphi}, $\pibadpsi C_i \ket{\psi^{(i-1)}_p} = C_i \ket{\psi^{(i-1)}_p}$, which implies the desired result. 
\qedhere
\end{enumerate}
\end{proof}

By the definition of $\pgreat^{(i-1)}$, there seems to be a bijective correspondence between elements of $\pgreat^{(i-1)}$ and $\pgood^{(i)}$. We make this precise by showing that the application of the oracle gate at the $i$th step to the state $\ket{\phi_\great^{(i-1)}}$, defined in \Cref{def:pgreat}, results in the state $\ket{\phi_\good^{(i)}}$ using \Cref{lem:pgreatphiimplication}.

\begin{lemma} \label{lem:O(phigreat(i))=phigood(i+1)}
Let $i \in [p(n)]$ and $c \in \mathcal{C}$. Suppose that $C_i = \controlled{O_c}$ for some $c \in \mathcal{C}$. Then $\ket{\phi_{\good}^{(i)}} = \tilde{C}_i \ket{\phi_\great^{(i-1)}}$.
\end{lemma}

\begin{proof}
Expanding $\tilde{C}_i\ket{\phi^{(i-1)}_{\good}}$ using \Cref{def:expandincompbasis,def:pgreat} gives
\begin{equation}
\label{eq:expansionofCphigood}
    \tilde{C}_i\ket{\phi^{(i-1)}_{\good}} = \sum_{p \in \pgreat^{(i-1)}} \sum_q \alpha^{(i-1)}_{p, q} \ket{q^{(i-1)}} \tilde{C}_i\ket{\phi^{(i-1)}_p} + \sum_{\substack{p \in \pgood^{(i-1)} \\ p \not \in \pgreat^{(i-1)}}} \sum_q \alpha^{(i-1)}_{p, q} \ket{q^{(i-1)}} \tilde{C}_i\ket{\phi^{(i-1)}_p}.
\end{equation}

Combining this with \cref{prop:decompositionofCphigood} of \Cref{lem:combined}, we find
\begin{equation} \label{eq:expansionofphibad+phigood}
    \ket{\phi^{(i)}_{\good}} + \ket{\phi^{(i)}_{\bad}} = \sum_{p \in \pgreat^{(i-1)}} \sum_q \alpha^{(i-1)}_{p, q} \ket{q^{(i-1)}} \tilde{C}_i\ket{\phi^{(i-1)}_p} + \sum_{\substack{p \in \pgood^{(i-1)} \\ p \not \in \pgreat^{(i-1)}}} \sum_q \alpha^{(i-1)}_{p, q} \ket{q^{(i-1)}} \tilde{C}_i\ket{\phi^{(i-1)}_p}.
\end{equation}

Applying $\pigoodphi$ on both sides, and using \cref{prop:phibadisbadandphigoodisgood} of \Cref{lem:combined} and \Cref{lem:pgreatphiimplication}, we obtain
\begin{equation}
\label{eq:applicationofpigoodphi}
    \ket{\phi^{(i)}_{\good}} = \sum_{p \in \pgreat^{(i-1)}} \sum_q \alpha^{(i-1)}_{p, q} \ket{q^{(i-1)}} \tilde{C}_i\ket{\phi^{(i-1)}_p} = \tilde{C}_i \ket{\phi^{(i-1)}_{\great}},
\end{equation}
where the last equality follows by \Cref{def:pgreat}.
\end{proof}

Next, we state and outline the proof of the vertex space analog of \Cref{lem:O(phigreat(i))=phigood(i+1)} in \Cref{lem:O(psigreat(i))=psigood(i+1)}. Note that this lemma uses the induction hypothesis of \Cref{lem:psigood=L(phigood)} as a premise since the definition of the state $\ket{\psi_\great^{(i)}}$ in \Cref{def:pgreat} is not derived from the definition of $\ket{\psi_\good^{(i)}}$, unlike the definition of $\ket{\phi_\great^{(i)}}$, which is derived from the expansion of $\ket{\phi_\good^{(i)}}$ in \Cref{def:expandincompbasis}.

\begin{lemma}
\label{lem:O(psigreat(i))=psigood(i+1)}
Let $i \in [p(n)]$ and $c \in \mathcal{C}$. Suppose that $C_i = \controlled{O_c}$ for some $c \in \mathcal{C}$ and $L\ket{\phi_{\good}^{(i-1)}} = \ket{\psi_{\good}^{(i-1)}}$. Then $\ket{\psi_{\good}^{(i)}} = C_i\ket{\psi_\great^{(i-1)}}$.
\end{lemma}

\begin{proof}
The proof is similar to that of \Cref{lem:O(phigreat(i))=phigood(i+1)}, with the main difference being the use of the supposition that $L\ket{\phi_{\good}^{(i-1)}} = \ket{\psi_{\good}^{(i-1)}}$ and \Cref{cor:alternatedefinitionofpgreat} along with \Cref{def:expandincompbasis} to establish an equation analogous to \cref{eq:expansionofCphigood}. Equations analogous to \cref{eq:expansionofphibad+phigood,eq:applicationofpigoodphi} are derived using \cref{prop:decompositionofCpsigood} of \Cref{lem:combined}, and \cref{prop:psibadisbadandpsigoodisgood} of \Cref{lem:combined} and \Cref{lem:pgreatpsiimplication}, respectively.
\end{proof}

The above analysis helps establish the following key lemma, which states that the states $\ket{\phi_{\good}^{(i)}}$ and $\ket{\psi_{\good}^{(i)}}$ are related by the mapping $L$. Intuitively, the oracle $\tilde{O}$ based on the address tree $\mathcal{T}$ can faithfully simulate (modulo mapping $L$) the portion of the state $\ket{\psi_\Al^{(i)}}$ of the algorithm $\mathcal{A}$ that does not encounter the $\exit$ or a cycle.

\begin{lemma} \label{lem:psigood=L(phigood)}
For all $i \in [p(n)] \cup \{0\}$, $L\ket{\phi_{\good}^{(i)}} = \ket{\psi_{\good}^{(i)}}$.
\end{lemma}

\begin{proof}
We prove this lemma by induction on $i$.

For $i=0$, note that
\begin{equation}
    \ket{\phi_{\good}^{(0)}} = \ket{\phi_{\Al}^{(0)}} - \ket{\phi_{\allbad}^{(0)}} = \ket{\phi_{\Al}^{(0)}} = \ket{\emptystring} \otimes \ket{0^{2p(n)}}^{\otimes(p(n)-1)} \otimes \ket{0}_{\workspace}
\end{equation}
where the first equality follows from \cref{prop:equivrepofphigood} of \Cref{lem:combined}, the second equality follows from \Cref{def:phigood}, and the last equality follows from \Cref{def:transcript}.
Similarly,
\begin{equation}
    \ket{\psi_{\good}^{(0)}} = \ket{\psi_{\Al}^{(0)}} - \ket{\psi_{\allbad}^{(0)}} = \ket{\psi_{\Al}^{(0)}} = \ket{\entrance} \otimes \ket{0^{2n}}^{\otimes(p(n)-1)} \otimes \ket{0}_{\workspace}
\end{equation}
where the first equality follows from \cref{prop:equivrepofpsigood} of \Cref{lem:combined}, the second equality follows from \Cref{def:phigood}, and the last equality follows from \Cref{def:genuinecircuit}.
The statement of the lemma for $i=0$ follows by noticing that $L(\emptystring) = \entrance$ and $L(0^{2n})=0^{2p(n)}$ by \cref{eq:L(s)=v}.

Now, assume that $L\ket{\phi_{\good}^{(i-1)}} = \ket{\psi_{\good}^{(i-1)}}$ for some $i \in [p(n)]$. Then, we have two cases depending on $C_i$:
\begin{enumerate}
\item $C_i$ is a genuine non-oracle gate. Then
\begin{align}
    L\ket{\phi_{\good}^{(i)}} &= L\tilde{C}_i\ket{\phi_{\good}^{(i-1)}} \\
    &= \sum_{p \in \pgood^{(i-1)}} \sum_q \alpha^{(i-1)}_{p, q} L\tilde{C}_i \ket{q^{(i-1)}} \ket{\phi^{(i-1)}_p} \\
    &= \sum_{p \in \pgood^{(i-1)}} \sum_q \alpha^{(i-1)}_{p, q} C_i \ket{q^{(i-1)}} \ket{\psi^{(i-1)}_p} \\
    &= C_i \ket{\psi_{\good}^{(i-1)}} \\
    &= \ket{\psi_{\good}^{(i)}}
\end{align}
where the first and last steps follow from \cref{enum:G(phigood(i))=phigood(i+1),enum:G(psigood(i))=psigood(i+1)} of \Cref{lem:G(phigood(i))=phigood(i+1)}, respectively; the second follows from \Cref{def:expandincompbasis}; the third follows from \Cref{lem:LC(phi)=C(psi)}; and the fourth follows from the induction hypothesis.

\item $C_i = \controlled{O_c}$ for some $c \in \mathcal{C}$. Then
\begin{align}
    L\ket{\phi_{\good}^{(i)}} &= L\tilde{C}_i\ket{\phi_{\great}^{(i-1)}} \\
    &= \sum_{p \in \pgreat^{(i-1)}} \sum_q \alpha^{(i-1)}_{p, q} L\tilde{C}_i \ket{q^{(i-1)}} \ket{\phi^{(i-1)}_p} \\
    &= \sum_{p \in \pgreat^{(i-1)}} \sum_q \alpha^{(i-1)}_{p, q} C_i \ket{q^{(i-1)}} \ket{\psi^{(i-1)}_p} \\
    &= C_i\ket{\psi_{\great}^{(i-1)}} \\
    &= \ket{\psi_{\good}^{(i)}}
\end{align}
where the first step follows from \Cref{lem:O(phigreat(i))=phigood(i+1)}, the second and fourth follow from \Cref{def:pgreat}, the third follows from \Cref{lem:LC(phi)=C(psi)}, and the last follows from the induction hypothesis and \Cref{lem:O(psigreat(i))=psigood(i+1)}. 
\qedhere
\end{enumerate}
\end{proof}

Finally, we show the following relationship between the norms of $\ket{\phi^{(i)}_{\good}}$ and $\ket{\psi^{(i)}_{\good}}$, which will be very useful in bounding the probability of success of Algorithm $\mathcal{A}$ in \Cref{subsec:analysisonall}.

\begin{lemma} \label{lem:||psigood||=||phigood||}
Let $i \in [p(n)] \cup \{0\}$. Then $\norm{\ket{\phi^{(i)}_{\good}}} = \norm{\ket{\psi^{(i)}_{\good}}}$.
\end{lemma}

\begin{proof}
It is clear, by \Cref{def:expandincompbasis}, that for $p, p' \in \pgood^{(i)}$ with $p \neq p'$, $\ket{\phi^{(i)}_p} \neq \ket{\phi^{(i)}_{p'}}$. Since $\ket{\phi^{(i)}_p}$ and $\ket{\phi^{(i)}_{p'}}$ are computational basis states, this means that $\bra{\phi^{(i)}_p}\ket{\phi^{(i)}_{p'}} = 0$ whenever $p \neq p'$. Notice that $\ket{\psi^{(i)}_p}$ and $\ket{\psi^{(i)}_{p'}}$ are in the support of $\ket{\psi_\good^{(i)}}$ by \Cref{lem:psigood=L(phigood)}. The contrapositive of \Cref{lem:pbadiscycle} implies that for $p \neq p'$, $\ket{\psi^{(i)}_p} \neq \ket{\psi^{(i)}_{p'}}$, which essentially means that $\bra{\psi^{(i)}_p}\ket{\psi^{(i)}_{p'}} = 0$ since $\ket{\psi^{(i)}_p}$ and $\ket{\psi^{(i)}_{p'}}$ are computational basis states. Combining these observations with \Cref{lem:psigood=L(phigood)}, we get
\begin{align}
    \norm{\ket{\psi_{\good}^{(i)}}} &= \norm{L\ket{\phi_{\good}^{(i)}}} \\
    &= \norm{\sum_{p \in \pgood^{(i)}} \sum_{q} \alpha^{(i)}_{p,q} \ket{q^{(i)}} \ket{\psi^{(i)}_p}} \\
    &= \sum_{p \in \pgood^{(i)}} \sum_{q} \left|\alpha^{(i)}_{p,q}\right|^2 \\
    &= \norm{\sum_{p \in \pgood^{(i)}} \sum_{q} \alpha^{(i)}_{p,q} \ket{q^{(i)}} \ket{\phi^{(i)}_p}} \\
    &= \norm{\ket{\phi_{\good}^{(i)}}}
\end{align}
as claimed.
\end{proof}

\subsection{The state is mostly good} \label{subsec:analysisonall}

In the remainder of this section, we conclude that it is hard for any rooted genuine quantum algorithm to find the $\exit$ (and hence, an $\entrance$--$\exit$ path). We  achieve this goal by bounding the mass of the quantum state $\ket{\psi_\Al}$ associated with any arbitrarily chosen rooted genuine quantum algorithm $\mathcal{A}$ that lies in the $\psi$-$\BAD$ subspace. We proceed by first using the result of \Cref{sec:3coloring} to bound the mass of the quantum state $\ket{\phi^{(i)}_\Al}$ that lies in the $\phi$-$\BAD$ subspace.

\begin{lemma} \label{lem:pibadphiissmall}
Let $i \in [p(n)] \cup \{0\}$. Then $\norm{\pibadphi \ket{\phi^{(i)}_\Al}}^2 \leq 4p(n)^4 \cdot 2^{-n/3}$.
\end{lemma}

\begin{proof}
By \Cref{def:piphi}, each computational basis state in $\pibadphi \ket{\phi^{(i)}_\Al}$ is a $\phi$-bad state. Therefore, the $i$th step of the classical \Cref{alg:C(T)1} outputs a computational basis state $\ket{\psi^{(i)}}$ corresponding to a subgraph of $\mathcal{G}$ that contains the $\exit$ or is at most one edge away from containing a cycle using at most $i \leq p(n)$ queries with probability $\norm{\pibadphi \ket{\phi^{(i)}_\Al}}^2$. In the case $\ket{\psi^{(i)}}$ does not contain a cycle, we can run a depth-first search of length $1$ on the subgraph corresponding to $\ket{\psi^{(i)}}$ using at most $i$ additional queries. Hence, we have found the $\exit$ or a cycle using at most $2i \leq 2p(n)$ classical queries with probability $\norm{\pibadphi \ket{\phi^{(i)}_\Al}}^2$. Noting that \Cref{alg:C(T)1} has the form of the classical query algorithms considered by \Cref{thm:classical3-colohardnessmain}, we see that $\norm{\pibadphi \ket{\phi^{(i)}_\Al}}^2 \leq 4p(n)^4 \cdot 2^{-n/3}$.
\end{proof}

From the result of \Cref{lem:pibadphiissmall}, one might intuitively conjecture that the size of the portion of the state $\ket{\phi^{(i)}_{\Al}}$ after $i$ steps that encountered the $\exit$ or a near-cycle at some point in its history is small. 
We formalize this as follows.

\begin{lemma} \label{lem:phiallbadissmall}
For all $i \in [p(n)] \cup \{0\}$, $\norm{\ket{\phi^{(i)}_{\allbad}}} \leq 2ip(n)^2 \cdot 2^{-n/6}$.
\end{lemma}

\begin{proof}
We prove the lemma by induction on $i$. The base case ($i=0$) is easy to observe as $\norm{\ket{\phi^{(0)}_{\allbad}}} = 0$. Now, pick any $i \in [p(n)]$ and suppose that the lemma is true for $i-1$, i.e., $\norm{\ket{\phi^{(i-1)}_{\allbad}}} \leq 2(i-1)p(n)^2 \cdot 2^{-n/6}$. Then, by \Cref{def:phigood,def:piphi} and \cref{prop:expansionofpibadphi} of \Cref{lem:combined},
\begin{align}
    \ket{\phi^{(i)}_{\allbad}} &= \ket{\phi^{(i)}_\bad} + \tilde{C}_i\ket{\phi^{(i-1)}_{\allbad}} \\
    &= \ket{\phi^{(i)}_\bad} + \pibadphi \tilde{C}_i\ket{\phi^{(i-1)}_{\allbad}} +  \pigoodphi \tilde{C}_i\ket{\phi^{(i-1)}_{\allbad}} \\
    &= \pibadphi \ket{\phi^{(i)}_{\Al}} + \pigoodphi \tilde{C}_i\ket{\phi^{(i-1)}_{\allbad}}, 
\end{align}
so we have
\begin{align}
    \norm{\ket{\phi^{(i)}_{\allbad}}} &= \norm{\pibadphi\ket{\phi^{(i)}_{\Al}} + \pigoodphi \tilde{C}_i\ket{\phi^{(i-1)}_{\allbad}}} \\
    &\leq \norm{\pibadphi \ket{\phi^{(i)}_{\Al}}} + \norm{\pigoodphi \tilde{C}_i\ket{\phi^{(i-1)}_{\allbad}}} \\
    &\leq \norm{\pibadphi \ket{\phi^{(i)}_{\Al}}} + \norm{\ket{\phi^{(i-1)}_{\allbad}}} \\
    &= \frac{2ip(n)^2}{2^{n/6}}
\end{align}
where the second step follows by the triangle inequality, the third by the fact that applying a unitary $\tilde{C}_i$ and the projector $\pigoodphi$ cannot increase the norm of any vector, and the fourth by \Cref{lem:pibadphiissmall} and the induction hypothesis. 
\end{proof}

The bound on the size of the portion of the state $\ket{\phi^{(i)}_{\Al}}$ after $i$ steps that never encountered the $\exit$ or a near-cycle directly follows from \Cref{lem:phiallbadissmall} as shown by the following corollary. 

\begin{corollary} \label{cor:lowerboundonphigood}
Let $i \in [p(n)]$. Then $\norm{\ket{\phi^{(i)}_{\good}}} \geq 1 - 2ip(n)^2 \cdot 2^{-n/6}$.
\end{corollary}

\begin{proof}
Observe that
\begin{equation}
    \norm{\ket{\phi^{(i)}_{\good}}} = \norm{\ket{\phi^{(i)}_{\Al}} - \ket{\phi^{(i)}_{\allbad}}} \geq \norm{\ket{\phi^{(i)}_{\Al}}} - \norm{\ket{\phi^{(i)}_{\allbad}}} \geq 1 - \frac{2ip(n)^2}{2^{n/6}}
\end{equation}
where the equality follows by \Cref{def:phigood}, the first inequality is an application of the triangle inequality, and the second inequality follows by \Cref{lem:phiallbadissmall} and the fact that $\ket{\phi^{(i)}_{\Al}}$ is a quantum state.
\end{proof}

In the next lemma, we bound the mass of the portion of the state $\ket{\psi^{(i)}_{\Al}}$ after $i$ steps that encountered the $\exit$ or a near-cycle at some point in its history. This is a crucial lemma for our result in this section where we invoke \Cref{lem:||psigood||=||phigood||} to deduce a statement about the quantum state of the genuine algorithm $\mathcal{A}$ using known properties of the state of our classical simulation of $\mathcal{A}$.

\begin{lemma} \label{lem:psiallbadissmall}
Let $i \in [p(n)]$. Then $\norm{\ket{\psi_\allbad^{(i)}}}^2 \leq 4i^2p(n)^2 \cdot 2^{-n/6}$.
\end{lemma}

\begin{proof}
Observe that
\begin{align}
    \norm{\ket{\psi_\allbad^{(i)}}} &= \norm{\sum_{j \in [i]} C_{j,i} \ket{\psi_\bad^{(j)}}} \\
    &\leq \sum_{j \in [i]} \norm{C_{j,i} \ket{\psi_\bad^{(j)}}} \\
    &= \sum_{j \in [i]} \norm{\ket{\psi_\bad^{(j)}}}
\end{align}
where the first step follows by \Cref{def:phigood}, the second by triangle inequality, and the third by the fact that the unitary $C_{j,i}$ preserves norms.

Thus, we have
\begin{align}
    \norm{\ket{\psi_\allbad^{(i)}}}^2 &\leq \left(\sum_{j \in [i]} \norm{\ket{\psi_\bad^{(j)}}}\right)^2 \\
    &\leq \sum_{j \in [i]} i\norm{\ket{\psi_\bad^{(j)}}}^2 \\
    &= i\left(1-\norm{\ket{\psi_\good^{(i)}}}^2\right) \\
    &= i\left(1-\norm{\ket{\phi_\good^{(i)}}}^2\right) \\
    &\leq i\left(1-\left(1-\frac{2ip(n)^2}{2^{n/6}}\right)^2\right) \\
    &\leq \frac{4i^2p(n)^2}{2^{n/6}}
\end{align}
where the two equalities follow by \Cref{lem:psigood^2+allpsibad^2=1,lem:||psigood||=||phigood||}, respectively, and the next-to-last inequality follows by \Cref{cor:lowerboundonphigood}. 
\end{proof}

We are now ready to establish our main theorem, which formally proves the hardness of finding an $\entrance$--$\exit$ path for \genuine, \rooted\ quantum query algorithms.

\begin{theorem} \label{thm:maintheorem}
No \genuine, \rooted\ quantum query algorithm for the path-finding problem can find a path from $\entrance$ to $\exit$ with more than exponentially small probability.
\end{theorem}

\begin{proof}
Since we let $\mathcal{A}$ be an arbitrary genuine, rooted quantum algorithm, it is sufficient to show that $\mathcal{A}$ cannot find a path from $\entrance$ to $\exit$ with more than exponentially small probability. Note that, by \Cref{def:phi-badandphi-good}, any computational basis state $\ket{\psi}$ that corresponds to a subgraph that stores an $\entrance$ to $\exit$ path must be $\psi$-$\mathrm{bad}$. That is, such a $\ket{\psi}$ must be in the support of $\pibadpsi\ket{\psi_\Al}$ from \Cref{def:piphi}. Recall from \Cref{def:genuine} that the genuine algorithm $\mathcal{A}$ measures the state $\ket{\psi_\Al}$ and outputs the resulting set of vertices. Thus, the probability that the genuine, rooted quantum query algorithm $\mathcal{A}$ finds an $\entrance$ to $\exit$ path is at most
\begin{align}
    \norm{\pibadpsi\ket{\psi_\Al}}^2 &= \norm{\pibadpsi\ket{\psi^{(p(n)}_\Al}}^2 \\
    &=\norm{\pibadpsi\ket{\psi^{(p(n)}_\good} + \pibadpsi\ket{\psi^{(p(n)}_\allbad}}^2 \\
    &= \norm{\pibadpsi\ket{\psi^{(p(n)}_\allbad}}^2 \\
    &\leq \norm{\ket{\psi^{(p(n)}_\allbad}}^2 \\
    &\leq \frac{4p(n)^4}{2^{n/6}}
\end{align}
where we used \Cref{def:phigood} in the first step, \cref{prop:equivrepofpsigood} of \Cref{lem:combined} in the second, \cref{prop:psibadisbadandpsigoodisgood} of \Cref{lem:combined} in the third, the fact that applying a projector $\pibadpsi$ cannot increase the norm of any vector in the fourth step, and \Cref{lem:psiallbadissmall} in the last.
\end{proof}

\section{Hardness of classical cycle finding with a 3-color oracle} \label{sec:3coloring}

In this section, we analyze the classical query complexity of finding the $\exit$ or a cycle in a randomly chosen 3-colored \wtg\ of size $n$. More precisely, we show in \Cref{thm:classical3-colohardnessmain} that the probability of finding the $\exit$ or a cycle for a natural class of classical algorithms is exponentially small even for a \wtg\ whose vertices are permuted according to the distribution $D_n$ specified in \Cref{def:colorpreservingpermutation} below. Informally, $D_n$ gives rise to the uniform distribution on \wtg s over the set that is constructed by fixing a 3-colored \wtg\ $\mathcal{G}$ and randomizing the edges of the $\weld$ (defined in \Cref{def:TLTRandweld}), making sure that the resulting graphs are valid 3-colored \wtg s. 

The key ingredient of our analysis is \Cref{lem:3-coloringmain} (see also \Cref{cor:3-coloringhardnessmain}), which informally says that for a \wtg\ sampled according to the aforementioned distribution, it is exponentially unlikely for a certain natural class of classical algorithms (i) to get `close' to the $\entrance$ or the $\exit$ starting on any vertex in the $\weld$ without backtracking, or (ii) to encounter two $\weld$ vertices that are connected by multiple `short' paths. Note that statement (i) implies that it is hard for any such classical algorithm to find the $\exit$ whereas statement (ii) has a similar implication for finding a cycle.

An astute reader might notice the resemblance of \Cref{lem:3-coloringmain} with Lemma 8 of \cite{ChildsCDFGS03}. Indeed, the latter lemma shows that it is hard for any classical algorithm with access to a colorless \wto\ to satisfy either statement (i) or (ii) mentioned above. However, the argument of \cite{ChildsCDFGS03} is different than ours in two major ways: our proof is by induction, and we use randomness of the $\weld$ and graph theoretic properties of any 3-coloring of a \wtg\ to argue the unlikeliness of statements (i) and (ii) while they use the hardness of guessing multiple coin tosses along with randomness of the $\weld$.

We begin by specifying basic notions about binary trees and \wtg s that facilitate proving \Cref{lem:equalcolorings} and \Cref{cor:equalsequences}.

\begin{definition}
Let $T$ be a binary tree of height $n$. We say that a vertex of $T$ is in \emph{column} $i$ if its distance from the root of $T$ is $i$, and an edge in $T$ is at \emph{level} $i$ if it connects a vertex in column $i-1$ to a vertex in column $i$. We extend this notion to the \wtg s. Precisely, a vertex of a \wtg\ $\mathcal{G}$ is in \emph{column} $i$ if its distance from the $\entrance$ is $i$, and an edge in $\mathcal{G}$ is at \emph{level} $i$ if it connects a vertex in column $i-1$ to a vertex in column $i$. 
\end{definition}

\begin{definition}
Let $T$ be a binary tree of height $n$, which is edge-colored using the 3 colors in $\mathcal{C}$. Let $v$ be any leaf of $T$, let $j \in [n]$, and let $u$ be the ancestor of $v$ in $T$ that is distance $j$ away from $v$. We define $t_v^j$ to be the length-$j$ sequence of colors from $v$ to $u$.
\end{definition}

The following lemma formalizes the observation that the colors of edges at any level of a binary tree are close to uniformly distributed among the possible 3 colors.

\begin{lemma} \label{lem:equalcolorings}
Let $T$ be a binary tree of height $n$ that is edge-colored using the 3 colors in $\mathcal{C}$. Let $\cbad \in \mathcal{C}$ denote the unique color such that there is no $\cbad$-colored edge incident to the root of $T$. For each $c \in \mathcal{C}$ and $i \in [n]$, let $\gamma(c, i)$ denote the number of $c$-colored edges at level $i$. Then
\begin{equation}
    \gamma(c, i) =
    \begin{cases}
        \floor{2^{i}/3} & i \text{ is odd and } c=\cbad \text{, or } i \text{ is even and } c \neq \cbad \\
        \ceil{2^{i}/3} & \text{otherwise}.
    \end{cases}
\end{equation}
\end{lemma}

\begin{proof}
We prove this lemma by induction on $i$. An edge is at level $1$ in $T$ if and only if it is incident to the root of $T$. The base case of the lemma directly follows. 

Assume that this lemma is true for some $i \in [n-1]$. Suppose that $i$ is odd. Note that, for any $c \in \mathcal{C}$, any vertex in column $i$ is connected to a vertex in column $i+1$ with a $c$-colored edge if and only if it is connected to a vertex in column $i-1$ with a $c'$-colored edge for $c \neq c'$. Thus the number of $c$-colored edges at level $i+1$ is
\begin{align}
    \gamma(c, i+1) &= \sum_{\substack{c' \in \mathcal{C} \\ c' \neq c}} \gamma(c', i) \\
    &= \begin{cases}
        2\ceil{2^{i}/3}  & c = \cbad \\
        \ceil{2^{i}/3} + \floor{2^{i}/3}  & \text{otherwise}
    \end{cases} \\
    &= \begin{cases}
        \ceil{2^{i+1}/3}  & c = \cbad \\
        \floor{2^{i+1}/3}  & \text{otherwise}.
    \end{cases} 
\end{align} 

The analysis for even $i$ is very similar.
\end{proof}

\Cref{lem:equalcolorings} directly implies the following corollary, which informally states that the number of paths from a particular level $n-j$ to the leaves of a binary tree are almost-uniformly distributed among all possible color sequences of length $j$ that do not contain an even-length palindrome.

\begin{corollary} \label{cor:equalsequences}
Let $T$ be a binary tree of height $n$, which is edge-colored using the 3 colors in $\mathcal{C}$. Let $j \in [n]$ and fix a length-$j$ sequence of colors $t \in \mathcal{C}^j$ that does not contain an even-length palindrome. Then the number of leaves $v$ of $T$ satisfying $t_v^j = t$ is at most $\ceil{2^{n-j+1}/3}$. 
\end{corollary}

\begin{proof}
Let $c \in \mathcal{C}$ be the last color appearing in the sequence $t$. Note that any ancestor of a leaf $v$ of $T$ that is distance $j$ away from $v$ must be in column $n-j$ of $T$. For each vertex $u$ in column $n-j$ with some edge at level $n-j+1$ incident to $u$ being $c$-colored, there is exactly one leaf $v$ such that $v$ is a descendant of $u$ and $t_v^j = t$. By \Cref{lem:equalcolorings}, there are at most $\ceil{2^{n-j+1}/3}$ $c$-colored edges at level $n-j+1$. Therefore, there are at most $\ceil{2^{n-j+1}/3}$ leaves of $T$ satisfying $t_v^j = t$.
\end{proof}

Consider the following induced subgraphs of $\mathcal{G}$.

\begin{definition} \label{def:TLTRandweld}
Define $T_L$, $T_R$, and $\weld$ to be the induced subgraphs of $\mathcal{G}$ on vertices in columns $\{0, \ldots, n\}$, columns $\{n+1, \ldots, 2n+1\}$, and columns $\{n, n+1\}$ of $\mathcal{G}$, respectively.
\end{definition}

Informally, $T_L$ and $T_R$ are induced subgraphs of $\mathcal{G}$ on vertices in the left and right binary trees of $\mathcal{G}$, respectively, while $\weld$ is the induced subgraph of $\mathcal{G}$ on the leaves of the left and right binary trees of $\mathcal{G}$. Note that $T_L$ and $T_R$ are height-$n$ subtrees of $\mathcal{G}$ rooted at $\entrance$ and $\exit$, respectively. Furthermore, the vertices of $T_L$ and $T_R$ provide a bipartition of the vertices of $\mathcal{G}$, and the edges of $T_L$, $T_R$, and $\weld$ provide a tripartition of the edges of $\mathcal{G}$.

For the rest of this section, we will assume that $n$ is a multiple of $3$ for simplicity of presentation. Indeed, this assumption is without loss of generality as one can replace each $n/3$ appearing in this section with $\floor{n/3}$ or $\ceil{n/3}$ as appropriate. Next, we consider $2^{2n/3}$ disjoint subtrees of $T_L$ (respectively $T_R$), each of which contains (as leaves) $2^{n/3}$ leaves of $T_L$ (respectively $T_R$). 

\begin{definition}
Fix an ordering of the vertices of $T_L$ (respectively $T_R$) in column $2n/3$. For any $i \in [2^{2n/3}]$, let $T_{L_i}$ (respectively $T_{R_i}$) denote the subtree of $T_L$ (respectively $T_R$) that is a binary tree of height $n/3$ rooted at the $i$th vertex in column $2n/3$ of $T_L$ (respectively $T_R$). Moreover, let $\mathbb{T} \defeq \{T_{L_i}, T_{R_i}: i \in [2^{2n/3}]\}$. 
\end{definition}

We categorize the vertices of $\weld$ depending on whether they belong to $T_L$ or $T_R$, and on the colors of the edges joining them to non-$\weld$ vertices.

\begin{definition} \label{def:c-leftandc-rightvertices}
For any leaf $v$ of the tree $T_L$ and any $c \in \mathcal{C}$, if the color of the edge connecting $v$ with $T_L$ in $\mathcal{G}$ is $c$, then we say that $v$ is a $c$-left vertex. Similarly, we define the notion of a $c$-right vertex for each $c \in \mathcal{C}$.
\end{definition}

As an example, the vertices colored {\color{Lavender} lavender} and {\color{Plum} plum} in \Cref{fig:weldedtree3colored} are red-left vertices.

Next, we define permutations that map valid 3-colored \wtg s to valid 3-colored \wtg s (as we show in \Cref{lem:colorpreservingpermutations}). Note that \Cref{def:c-leftandc-rightvertices} partitions the vertices of $\weld$ into 6 parts. The following definition is crucial for describing a distribution of \wtg s that is classically `hard'.

\begin{definition} [Color-preserving permutations] \label{def:colorpreservingpermutation}
A permutation $\sigma$ of the vertices of $\weld$ is a \emph{color-preserving permutation} if for any vertex $v$ of $\weld$, and any $c \in \mathcal{C}$, $v$ is $c$-left (respectively $c$-right) iff $\sigma(v)$ is $c$-left (respectively $c$-right). We will sometimes refer to a color-preserving permutation $\sigma$ as a permutation of $\valid$ that acts as $\sigma$ on vertices of $\weld$ and as identity on non-$\weld$ vertices.
For any color-preserving permutation $\sigma$, let $\weld^\sigma$ denote the graph obtained by applying $\sigma$ to the vertices of $\weld$.
For any color-preserving permutation $\sigma$, let $\mathcal{G}^\sigma$ denote the graph on vertex set $\valid$ obtained by permuting the $\weld$ edges of $\mathcal{G}$ according to $\sigma$ (and leaving the rest of the graph $\mathcal{G}$ as it is): if $u,v$ are $\weld$ vertices, then there is an edge of color $c \in \mathcal{C}$ in $\mathcal{G}$ joining vertices $u$ and $v$ if and only if there is an edge of color $c$ in $\mathcal{G}^\sigma$ joining vertices $\sigma(u)$ and $\sigma(v)$; otherwise, there is an edge of color $c \in \mathcal{C}$ in $\mathcal{G}$ joining vertices $u$ and $v$ if and only if there is an edge of color $c$ in $\mathcal{G}^\sigma$ joining vertices $u$ and $v$.
Define $D_n$ to be the uniform distribution over all color-preserving permutations $\sigma$.
\end{definition}

Note that $\mathcal{V}_{\mathcal{G}^\sigma} = \valid$, and the non-$\weld$ vertices and edges of $\mathcal{G}$ remain invariant under $\sigma$. Therefore, the induced subgraph of $\mathcal{G}^\sigma$ on vertices of $T_L$ (respectively $T_R$) is $T_L$ (respectively $T_R$). Moreover, $\weld^\sigma$ is the induced subgraph of $\mathcal{G}^\sigma$ on vertices of $\weld$.  \Cref{fig:colorpreservingpermutation} illustrates an example of a color-preserving permutation. 

For each $c \in \mathcal{C}$, by \Cref{lem:equalcolorings}, we have at least $\floor{2^n/3}$ $c$-left and at least $\floor{2^n/3}$ $c$-right vertices in $\mathcal{G}$. Hence, there are at least $\left(\floor{2^n/3}!\right)^6$ color-preserving permutations.

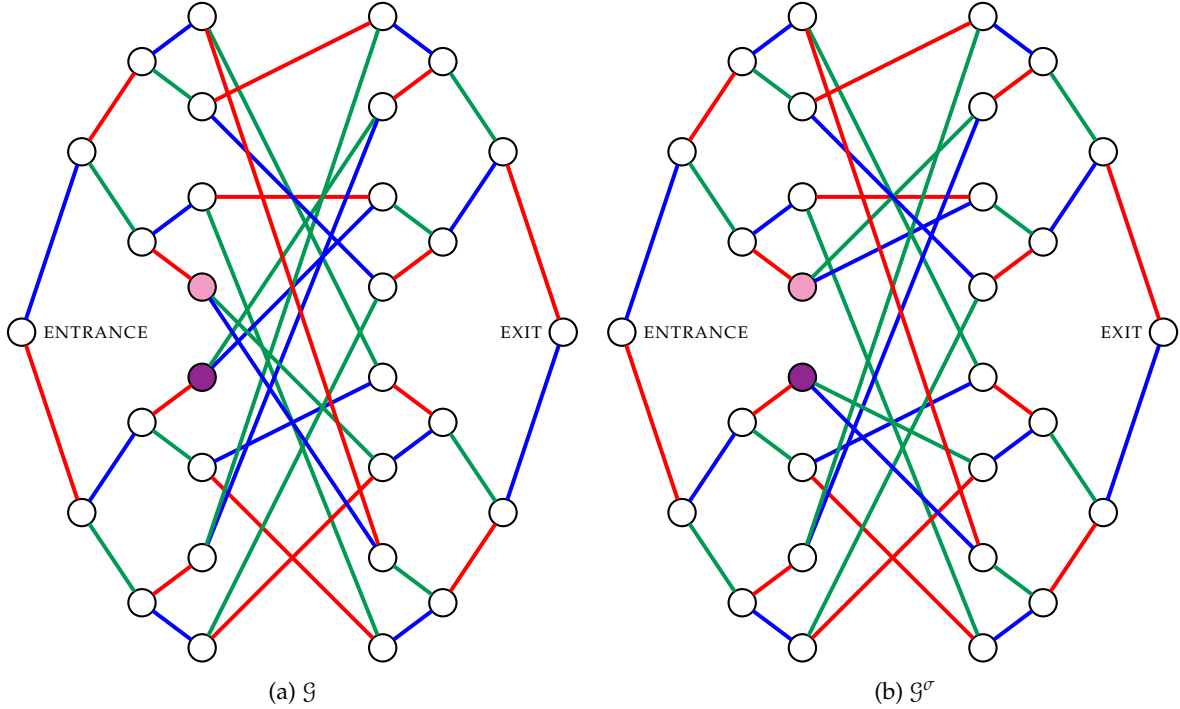
\begin{figure}
    \centering
    \begin{subfigure}[b]{0.46\linewidth}
    \begin{tikzpicture}[scale=0.2, auto, node distance=0.2cm, every loop/.style={}, thick, fill=black!20, every arrow/.append style={dash,thick}]
    \tikzset{VertexStyle/.style = {shape = ellipse, draw}}
    \Vertex[x=0,y=0,L=$ $]{ent} 
    \extralabel{0}{\scriptsize \textsc{entrance}}{ent}
    \Vertex[x=4,y=-12,L=$ $]{Lr}
    \Vertex[x=4,y=12,L=$ $]{Lb}
    \Vertex[x=8,y=-18,L=$ $]{Lrg}
    \Vertex[x=8,y=-6,L=$ $]{Lrb}
    \Vertex[x=8,y=6,L=$ $]{Lbg}
    \Vertex[x=8,y=18,L=$ $]{Lbr}
    \Vertex[x=12,y=-21,L=$ $]{Lrgb}
    \Vertex[x=12,y=-15,L=$ $]{Lrgr}
    \Vertex[x=12,y=-9,L=$ $]{Lrbg}
    \tikzset{VertexStyle/.style = {shape = ellipse, draw, fill=Plum}}
    \Vertex[x=12,y=-3,L=$ $]{Lrbr}
    \tikzset{VertexStyle/.style = {shape = ellipse, draw, fill=Lavender}}
    \Vertex[x=12,y=3,L=$ $]{Lbgr}
    \tikzset{VertexStyle/.style = {shape = ellipse, draw}}
    \Vertex[x=12,y=9,L=$ $]{Lbgb}
    \Vertex[x=12,y=15,L=$ $]{Lbrg}
    \Vertex[x=12,y=21,L=$ $]{Lbrb}
    
    \Vertex[x=36,y=0,L=$ $]{exit} 
    \extralabel{180}{\scriptsize \textsc{exit}}{exit}
    \Vertex[x=32,y=-12,L=$ $]{Rb}
    \Vertex[x=32,y=12,L=$ $]{Rr}
    \Vertex[x=28,y=-18,L=$ $]{Rbr}
    \Vertex[x=28,y=-6,L=$ $]{Rbg}
    \Vertex[x=28,y=6,L=$ $]{Rrb}
    \Vertex[x=28,y=18,L=$ $]{Rrg}
    \Vertex[x=24,y=-21,L=$ $]{Rbrb}
    \Vertex[x=24,y=-15,L=$ $]{Rbrg}
    \Vertex[x=24,y=-9,L=$ $]{Rbgb}
    \Vertex[x=24,y=-3,L=$ $]{Rbgr}
    \Vertex[x=24,y=3,L=$ $]{Rrbr}
    \Vertex[x=24,y=9,L=$ $]{Rrbg}
    \Vertex[x=24,y=15,L=$ $]{Rrgr}
    \Vertex[x=24,y=21,L=$ $]{Rrgb}
    
    \Edge[color=red](ent)(Lr)
    \Edge[color=blue](ent)(Lb)
    \Edge[color=ForestGreen](Lr)(Lrg)
    \Edge[color=blue](Lr)(Lrb)
    \Edge[color=ForestGreen](Lb)(Lbg)
    \Edge[color=red](Lb)(Lbr)
    \Edge[color=blue](Lrg)(Lrgb)
    \Edge[color=red](Lrg)(Lrgr)
    \Edge[color=ForestGreen](Lrb)(Lrbg)
    \Edge[color=red](Lrb)(Lrbr)
    \Edge[color=red](Lbg)(Lbgr)
    \Edge[color=blue](Lbg)(Lbgb)
    \Edge[color=ForestGreen](Lbr)(Lbrg)
    \Edge[color=blue](Lbr)(Lbrb)
    
    \Edge[color=blue](exit)(Rb)
    \Edge[color=red](exit)(Rr)
    \Edge[color=ForestGreen](Rb)(Rbg)
    \Edge[color=red](Rb)(Rbr)
    \Edge[color=ForestGreen](Rr)(Rrg)
    \Edge[color=blue](Rr)(Rrb)
    \Edge[color=red](Rbg)(Rbgr)
    \Edge[color=blue](Rbg)(Rbgb)
    \Edge[color=ForestGreen](Rbr)(Rbrg)
    \Edge[color=blue](Rbr)(Rbrb)
    \Edge[color=blue](Rrg)(Rrgb)
    \Edge[color=red](Rrg)(Rrgr)
    \Edge[color=ForestGreen](Rrb)(Rrbg)
    \Edge[color=red](Rrb)(Rrbr)
    
    \Edge[color=ForestGreen](Lbrb)(Rbgr)
    \Edge[color=blue](Rbgr)(Lrbg)
    \Edge[color=red](Lrbg)(Rbrb)
    \Edge[color=ForestGreen](Rbrb)(Lbgb)
    \Edge[color=red](Lbgb)(Rrbg)
    \Edge[color=blue](Rrbg)(Lrbr)
    \Edge[color=ForestGreen](Lrbr)(Rrgr)
    \Edge[color=blue](Rrgr)(Lrgr)
    \Edge[color=ForestGreen](Lrgr)(Rrgb)
    \Edge[color=red](Rrgb)(Lbrg)
    \Edge[color=blue](Lbrg)(Rrbr)
    \Edge[color=ForestGreen](Rrbr)(Lrgb)
    \Edge[color=red](Lrgb)(Rbgb)
    \Edge[color=ForestGreen](Rbgb)(Lbgr)
    \Edge[color=blue](Lbgr)(Rbrg)
    \Edge[color=red](Rbrg)(Lbrb)
    \end{tikzpicture}
    \caption{$\mathcal{G}$} \label{fig:weldedtree3colored} 
    \end{subfigure}
    \hspace{0.5em}
    \begin{subfigure}[b]{0.48\linewidth}
    \begin{tikzpicture}[scale=0.20, auto, node distance=0.2cm, every loop/.style={}, thick, fill=black!20, every arrow/.append style={dash,thick}]
    \tikzset{VertexStyle/.style = {shape = ellipse, draw}}
    \Vertex[x=0,y=0,L=$ $]{ent} 
    \extralabel{0}{\scriptsize \textsc{entrance}}{ent}
    \Vertex[x=4,y=-12,L=$ $]{Lr}
    \Vertex[x=4,y=12,L=$ $]{Lb}
    \Vertex[x=8,y=-18,L=$ $]{Lrg}
    \Vertex[x=8,y=-6,L=$ $]{Lrb}
    \Vertex[x=8,y=6,L=$ $]{Lbg}
    \Vertex[x=8,y=18,L=$ $]{Lbr}
    \Vertex[x=12,y=-21,L=$ $]{Lrgb}
    \Vertex[x=12,y=-15,L=$ $]{Lrgr}
    \Vertex[x=12,y=-9,L=$ $]{Lrbg}
    \tikzset{VertexStyle/.style = {shape = ellipse, draw, fill=Plum}}
    \Vertex[x=12,y=-3,L=$ $]{Lrbr}
    \tikzset{VertexStyle/.style = {shape = ellipse, draw, fill=Lavender}}
    \Vertex[x=12,y=3,L=$ $]{Lbgr}
    \tikzset{VertexStyle/.style = {shape = ellipse, draw}}
    \Vertex[x=12,y=9,L=$ $]{Lbgb}
    \Vertex[x=12,y=15,L=$ $]{Lbrg}
    \Vertex[x=12,y=21,L=$ $]{Lbrb}
    \Vertex[x=36,y=0,L=$ $]{exit} 
    \extralabel{180}{\scriptsize \textsc{exit}}{exit}
    \Vertex[x=32,y=-12,L=$ $]{Rb}
    \Vertex[x=32,y=12,L=$ $]{Rr}
    \Vertex[x=28,y=-18,L=$ $]{Rbr}
    \Vertex[x=28,y=-6,L=$ $]{Rbg}
    \Vertex[x=28,y=6,L=$ $]{Rrb}
    \Vertex[x=28,y=18,L=$ $]{Rrg}
    \Vertex[x=24,y=-21,L=$ $]{Rbrb}
    \Vertex[x=24,y=-15,L=$ $]{Rbrg}
    \Vertex[x=24,y=-9,L=$ $]{Rbgb}
    \Vertex[x=24,y=-3,L=$ $]{Rbgr}
    \Vertex[x=24,y=3,L=$ $]{Rrbr}
    \Vertex[x=24,y=9,L=$ $]{Rrbg}
    \Vertex[x=24,y=15,L=$ $]{Rrgr}
    \Vertex[x=24,y=21,L=$ $]{Rrgb}
    
    \Edge[color=red](ent)(Lr)
    \Edge[color=blue](ent)(Lb)
    \Edge[color=ForestGreen](Lr)(Lrg)
    \Edge[color=blue](Lr)(Lrb)
    \Edge[color=ForestGreen](Lb)(Lbg)
    \Edge[color=red](Lb)(Lbr)
    \Edge[color=blue](Lrg)(Lrgb)
    \Edge[color=red](Lrg)(Lrgr)
    \Edge[color=ForestGreen](Lrb)(Lrbg)
    \Edge[color=red](Lrb)(Lrbr)
    \Edge[color=red](Lbg)(Lbgr)
    \Edge[color=blue](Lbg)(Lbgb)
    \Edge[color=ForestGreen](Lbr)(Lbrg)
    \Edge[color=blue](Lbr)(Lbrb)
    
    \Edge[color=blue](exit)(Rb)
    \Edge[color=red](exit)(Rr)
    \Edge[color=ForestGreen](Rb)(Rbg)
    \Edge[color=red](Rb)(Rbr)
    \Edge[color=ForestGreen](Rr)(Rrg)
    \Edge[color=blue](Rr)(Rrb)
    \Edge[color=red](Rbg)(Rbgr)
    \Edge[color=blue](Rbg)(Rbgb)
    \Edge[color=ForestGreen](Rbr)(Rbrg)
    \Edge[color=blue](Rbr)(Rbrb)
    \Edge[color=blue](Rrg)(Rrgb)
    \Edge[color=red](Rrg)(Rrgr)
    \Edge[color=ForestGreen](Rrb)(Rrbg)
    \Edge[color=red](Rrb)(Rrbr)
    
    \Edge[color=ForestGreen](Lbrb)(Rbgr)
    \Edge[color=blue](Rbgr)(Lrbg)
    \Edge[color=red](Lrbg)(Rbrb)
    \Edge[color=ForestGreen](Rbrb)(Lbgb)
    \Edge[color=red](Lbgb)(Rrbg)
    \Edge[color=blue](Rrbg)(Lbgr)
    \Edge[color=ForestGreen](Lbgr)(Rrgr)
    \Edge[color=blue](Rrgr)(Lrgr)
    \Edge[color=ForestGreen](Lrgr)(Rrgb)
    \Edge[color=red](Rrgb)(Lbrg)
    \Edge[color=blue](Lbrg)(Rrbr)
    \Edge[color=ForestGreen](Rrbr)(Lrgb)
    \Edge[color=red](Lrgb)(Rbgb)
    \Edge[color=ForestGreen](Rbgb)(Lrbr)
    \Edge[color=blue](Lrbr)(Rbrg)
    \Edge[color=red](Rbrg)(Lbrb)
    \end{tikzpicture}
    \caption{$\mathcal{G}^\sigma$} \label{fig:weldedtreesigma} 
    \end{subfigure}
    \caption{Example of a color-preserving permutation $\sigma$ for the graph $\mathcal{G}$ in \Cref{fig:weldedtreecoloredlabeled}. The permutation $\sigma$ is the identity permutation except that it maps the vertex colored {\color{Lavender} lavender} to the vertex colored {\color{Plum} plum}. Note that the resulting graph $G^\sigma$ is a valid 3-colored \wtg .}
    \label{fig:colorpreservingpermutation}
\end{figure}

We now verify that the graphs obtained by applying color-preserving permutations on $\mathcal{G}$ are valid 3-colored \wtg s.

\begin{lemma} \label{lem:colorpreservingpermutations}
Let $\sigma$ be a color-preserving permutation. Then $\mathcal{G}^\sigma$ is a valid 3-colored \wtg.
\end{lemma}

\begin{proof}
We first argue that $\mathcal{G}^\sigma$ is a valid \wtg. Recall that the edges of $\mathcal{G}$ that are not in the $\weld$ remain invariant under $\sigma$. Thus, it remains to show that $\weld^\sigma$ is a cycle alternating between the vertices in columns $n$ and $n+1$. Since $\mathcal{G}$ is a \wtg, $\weld$ is a cycle on vertices denoted by $v_1, \ldots, v_{2^{n+1}}$ such that $v_{2i-1}$ is a vertex in column $n$ and $v_{2i}$ is a vertex in column $n+1$ for each $i \in [2^n]$, and $v_i$ is joined to $v_{i+1 \bmod 2^{n+1}}$ and $v_{i-1 \bmod 2^{n+1}}$ for each $i \in [2^{n+1}]$. Therefore, since $\sigma$ is color-preserving, for each $i \in [2^n]$, $\sigma(v_{2i-1})$ is a vertex in column $n$ and $\sigma(v_{2i})$ is a vertex in column $n+1$, and for each $i \in [2^{n+1}]$, $\sigma(v_i)$ is joined to $\sigma(v_{i+1 \bmod 2^{n+1}})$ and $\sigma(v_{i-1 \bmod 2^{n+1}})$. It follows that $\weld^\sigma$ is a cycle alternating between the vertices in columns $n$ and $n+1$. 

Now, we claim that $\mathcal{G}^\sigma$ admits a valid 3-coloring. Let $v$ be any vertex of $\weld$. Without loss of generality, assume that $v$ is red-right. Then, as $\mathcal{G}$ is a valid \wtg , $v$ is joined to two vertices $v_g$ and $v_b$ of $\weld$ with a green and a blue edge respectively (along with a vertex of $T_R$ with a red edge). Since $\sigma$ is color-preserving, we know that $\sigma(v)$ is also red-right in $\mathcal{G}$, and is joined to a vertex $v_r$ of $T_R$ with a red edge. Therefore, in $\mathcal{G}^\sigma$, $\sigma(v)$ is joined to $v_r$ with a red edge and to vertices $\sigma(v_g)$ and $\sigma(v_b)$ of $\weld$ with a green and a blue edge, respectively. On the other hand, non-$\weld$ vertices of $\mathcal{G}$ and their adjacency lists are unchanged by $\sigma$. Our desired claim follows.
\end{proof}

Recall that \Cref{def:eta_c} specifies the classical oracle function $\eta^\sigma_c$ for each $c \in \mathcal{C}$ associated with $\mathcal{G}^\sigma$ for the identity permutation $\sigma$. The following definition generalizes this by specifying the classical oracle function $\eta^\sigma_c$ for each $c \in \mathcal{C}$ associated with $\mathcal{G}^\sigma$ for any color-preserving permutation $\sigma$.

\begin{definition}
Let $\valid$, $I_c$, and $N_c$ for each $c \in \mathcal{C}$, $\noedge$, and $\invalid$ be defined as in \Cref{def:weldedtree,def:eta_c}. For any color-preserving permutation $\sigma$, let 
\begin{align}
    \eta^\sigma_c(v) \defeq
    \begin{cases}
        \sigma(\eta_c(\sigma^{-1}(v)) & v, \eta_c(v) \in \mathcal{V}_\weld \\
        \eta_c(v) & \text{otherwise}
    \end{cases}
\end{align}
where $\mathcal{V}_\weld$ refers to the set of vertices of $\weld$.

Let $\eta^\sigma \defeq \{\eta^\sigma_c:c \in \mathcal{C}\}$ be the oracle corresponding to the color-preserving permutation $\sigma$.
\end{definition}

We now define the notion of path-embedding in $\mathcal{G}^\sigma$ for a sequence of colors $t$, which informally refers to the path resulting from beginning at the $\entrance$ and following the edge colors given by $t$ in order.

\begin{definition}[Path-embedding]
Let $\sigma$ be any color-preserving permutation. Let $\ell \in [p(n)]$ and $t \in \mathcal{C}^{\ell}$. That is, $t = (c_1, \ldots, c_{\ell})$ for some $c_1, \ldots, c_{\ell} \in \mathcal{C}$. Then, define the \emph{path-embedding} of $t$ under the oracle $\eta^\sigma$, denoted by $\eta^\sigma(t)$, to be a length-$\ell$ tuple of vertex labels as follows. The $j$th element of $\eta^\sigma(t)$ is 
\begin{equation}
    \eta^\sigma(t)_j \defeq
    \begin{cases}
        \eta^\sigma_{c_1}(\entrance) & j=1 \\
        \eta^\sigma_{c_j}(\eta^\sigma(t)_{j-1}) & \text{otherwise}.
    \end{cases}
\end{equation}
We say that the path-embedding $\eta^\sigma(t)$ \emph{encounters a vertex} $v$ if $\eta^\sigma(t)_j = v$ for some $j \in [\ell]$, and that $\eta^\sigma(t)$ \emph{encounters an edge} joining vertices $v$ and $u$ if $\eta^\sigma(t)_j = v$ and $\eta^\sigma(t)_{j+1} = u$ (or the other way around) for some $j \in [\ell-1]$. Moreover, $\eta^\sigma(t)$ \emph{encounters a cycle} in $\mathcal{G}^\sigma$ if it encounters a sequence of vertices and edges that forms a cycle in $\mathcal{G}^\sigma$ and \emph{encounters a tree from $\mathbb{T}$} if it encounters a leaf of this tree.
\end{definition}

\Cref{fig:path-embedding} demonstrates an example of a path in $\mathcal{T}$ and the corresponding path-embedding in $\mathcal{G}$.
We restrict our attention to the color sequences that do not contain even-length palindromes. For such a sequence $t$, the path-embedding $\eta^\sigma(t)$ encounters a cycle exactly when it encounters a vertex twice.

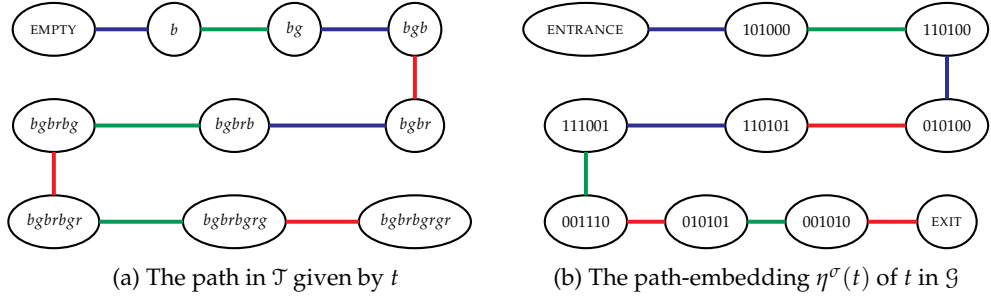
\begin{figure}
    \tiny
    \centering
    \begin{subfigure}[b]{0.4\linewidth}
    \begin{tikzpicture}[scale=0.32, auto, every loop/.style={}, thick, every arrow/.append style={dash,thick}]
    \tikzset{VertexStyle/.style = {shape = ellipse, minimum size = 20pt, draw}}
    \Vertex[x=0,y=30,L=$\emptyadd$]{emp}
    \Vertex[x=5,y=30,L=$b$]{b}
    \Vertex[x=10,y=30,L=$bg$]{bg}
    \Vertex[x=15,y=30,L=$bgb$]{bgb}
    \Vertex[x=15,y=26,L=$bgbr$]{bgbr}
    \Vertex[x=7.5,y=26,L=$bgbrb$]{bgbrb}
    \Vertex[x=0,y=26,L=$bgbrbg$]{bgbrbg}
    \Vertex[x=0,y=22,L=$bgbrbgr$]{bgbrbgr}
    \Vertex[x=7.5,y=22,L=$bgbrbgrg$]{bgbrbgrg}
    \Vertex[x=15,y=22,L=$bgbrbgrgr$]{bgbrbgrgr}
    \tikzset{every node/.style={opacity=0,text opacity=1,scale=0.75}}
    \tikzset{EdgeStyle/.style=auto, color=red}
    \Edge[color=Blue](emp)(b)
    \Edge[color=ForestGreen](b)(bg)
    \Edge[color=Blue](bg)(bgb)
    \Edge[color=Red](bgb)(bgbr)
    \Edge[color=Blue](bgbr)(bgbrb)
    \Edge[color=ForestGreen](bgbrb)(bgbrbg)
    \Edge[color=Red](bgbrbg)(bgbrbgr)
    \Edge[color=ForestGreen](bgbrbgr)(bgbrbgrg)
    \Edge[color=Red](bgbrbgrg)(bgbrbgrgr)
    \end{tikzpicture}
    \caption{The path in $\mathcal{T}$ given by $t$} 
    \label{fig:path-embeddingaddresses} 
    \end{subfigure}
    \begin{subfigure}[b]{0.4\linewidth}
    \hspace{0.65em}
    \begin{tikzpicture}[scale=0.32, auto, node distance=1cm, every loop/.style={}, thick, every arrow/.append style={dash,thick}]
    \tikzset{VertexStyle/.style = {shape = ellipse, minimum size = 20pt, draw}}
    \Vertex[x=0,y=30,L=$\entrance$]{emp}
    \Vertex[x=7.5,y=30,L=$101000$]{b}
    \Vertex[x=15,y=30,L=$110100$]{bg}
    \Vertex[x=15,y=26,L=$010100$]{bgb}
    \Vertex[x=7.5,y=26,L=$110101$]{bgbr}
    \Vertex[x=0,y=26,L=$111001$]{bgbrb}
    \Vertex[x=0,y=22,L=$001110$]{bgbrbg}
    \Vertex[x=5,y=22,L=$010101$]{bgbrbgr}
    \Vertex[x=10,y=22,L=$001010$]{bgbrbgrg}
    \Vertex[x=15,y=22,L=$\exit$]{bgbrbgrgr}
    \tikzset{every node/.style={opacity=0,text opacity=1,scale=0.75}}
    \tikzset{EdgeStyle/.style=auto, color=red}
    \Edge[color=Blue](emp)(b)
    \Edge[color=ForestGreen](b)(bg)
    \Edge[color=Blue](bg)(bgb)
    \Edge[color=Red](bgb)(bgbr)
    \Edge[color=Blue](bgbr)(bgbrb)
    \Edge[color=ForestGreen](bgbrb)(bgbrbg)
    \Edge[color=Red](bgbrbg)(bgbrbgr)
    \Edge[color=ForestGreen](bgbrbgr)(bgbrbgrg)
    \Edge[color=Red](bgbrbgrg)(bgbrbgrgr)
    \end{tikzpicture}
    \caption{The path-embedding $\eta^\sigma(t)$ of $t$ in $\mathcal{G}$} 
    \label{fig:path-embeddingvertices} 
    \end{subfigure}
    \caption{Example of a path-embedding for the graph $\mathcal{G}$ in \Cref{fig:weldedtreecoloredlabeled} and the identity permutation $\sigma$.}
    \label{fig:path-embedding}
\end{figure}

Now we describe notation for each time a certain path-embedding crosses the $\weld$ so that we can refer to the tree from $\mathbb{T}$ that it goes to and the $\weld$ edge that it goes through.

\begin{definition}
Let $\sigma$ be any color-preserving permutation and let $t$ be any sequence of colors that does not contain even-length palindromes. We use $T_{t,i}^\sigma \in \mathbb{T}$ to denote the $i$th subtree and $e^\sigma_{t,i}$ to denote the $i$th edge of the $\weld$ encountered by the path-embedding $\eta^\sigma(t)$.\footnote{It is possible for $\eta^\sigma(t)$ to encounter $\noedge$ or $\invalid$. However, once that happens, $\eta^\sigma(t)$ cannot further encounter any subtree from $\mathbb{T}$.} We refer to the event of the path-embedding $\eta^\sigma(t)$ encountering the $i$th edge of $\weld$ as the $i$th $\weld$-crossing. Furthermore, let $\ell^\sigma(t)$ denote the number of subtrees encountered by the path-embedding $\eta^\sigma(t)$. 
\end{definition}

Note that the number of $\weld$ edges encountered by the path-embedding $\eta^\sigma(t)$ is $\ell^\sigma(t)-1$. For each $i \in [\ell^\sigma(t)-1]$, the edge $e^\sigma_{t,i}$ joins a vertex in $T_{t,i}^\sigma$ to a vertex in $T_{t,i+1}^\sigma$. Next, we define $\ell^\sigma(t)$ prefixes of $t$ that are relevant for our analysis. Intuitively, for $i \in [\ell^\sigma(t)-1]$, the sequence of colors $\pre^\sigma_i(t)$ refers to the prefix of $t$ such that if we begin from the $\entrance$ and follow the edge colors given by $\pre^\sigma_i(t)$, we will arrive at the vertex reached by the $i$th edge of $\weld$ encountered by $\eta^\sigma(t)$.

\begin{definition}
Let $\sigma$ be any color-preserving permutation and let $t$ be any sequence of colors that does not contain even-length palindromes. Let $\pre^\sigma_{\ell^\sigma(t)}(t) \defeq t$. For each $i \in [\ell^\sigma(t)-1]$, let $\pre^\sigma_i(t)$ denote the longest prefix of $t$ such that the path-embedding $\eta^\sigma(\pre^\sigma_i(t))$ does not encounter the $i$th $\weld$-crossing.
\end{definition}

Notice that $\pre^\sigma_i(t)$ is a sequence that begins with the color of an edge incident to $\entrance$ and ends with the color of the edge encountered by the path-embedding $\eta^\sigma(t)$ just before the $i$th $\weld$-crossing. The following definition formalizes statements (i) and (ii) that we described intuitively at the beginning of this section.

\begin{definition} \label{def:desirable}
Let $\sigma$ be any color-preserving permutation and let $t$ be any sequence of colors that does not contain an even-length palindrome. We say that $t$ has \emph{small displacement} if after the first $\weld$-crossing, the path-embedding $\eta^\sigma(t)$ does not encounter any vertex that is distance at least $n/3$ away from the closest vertex of $\weld$. We say that $t$ is \emph{non-colliding} if for any edge $e$ joining some leaf of $T_{t,i}^\sigma$ with some leaf of $T_{t,j}^\sigma$ for some $i, j \in [\ell^\sigma(t)]$, $e$ must be $e^\sigma_{t,i}$ or $e^\sigma_{t,j}$. 
We say that $t$ is \emph{desirable} if $t$ has small displacement and is non-colliding. We also say that $t$ has \emph{large displacement} if it does not have small displacement, that $t$ is \emph{colliding} if it is not non-colliding, and that $t$ is \emph{undesirable} if it is not desirable.
\end{definition}

Note that, in the above definition, if $e = e^\sigma_{t,i}$, then $j=i+1$ and if $e = e^\sigma_{t,j}$, then $j=i-1$. Thus, $t$ being non-colliding essentially means that if there is an edge $e$ between trees $T_{t,i}^\sigma$ and $T_{t,j}^\sigma$ for some $i, j \in [\ell^\sigma(t)]$, then $j=i+1$ or $j=i-1$, and $T_{t,i}^\sigma$ and $T_{t,j}^\sigma$ are not joined by any edge other than $e$.

It is easy to see that for any sequence of colors $t$ that does not contain an even-length palindrome, beginning from the $\entrance$ and following the sequence of colors specified by $t$ will not result in reaching the $\exit$ if $t$ has small displacement, and will not result in going through a cycle if $t$ is non-colliding. The following lemma is crucial for our argument in this section, which essentially shows that any prefix of any fixed sequence of colors $t$ is unlikely to have large displacement or be colliding (as defined in \Cref{def:desirable}).

\begin{lemma} \label{lem:3-coloringmain}
Let $\ell \in [p(n)]$ and $t \in \mathcal{C}^{\ell}$ such that $t$ does not contain even-length palindromes. Choose the permutation $\sigma$ according to the distribution $D_n$. Then for all $i \in [\ell^\sigma(t)]$, $\pre_i(t)$ is desirable with probability at least $1-4i^2 2^{-n/3}$ over the choice of $\sigma$.
\end{lemma}

\begin{proof}
The proof is by induction on $i$. For the base case, note that $\pre^\sigma_1(t)$ does not encounter any edges in any tree in $\mathbb{T}$ other than the one it first reaches. Therefore, $\pre^\sigma_1(t)$ is desirable with certainty. 

Let $i \in [\ell^\sigma(t)-1]$. Choose $\sigma$ according to the distribution $D_n$. As the induction hypothesis, assume that $\pre^\sigma_i(t)$ is desirable with probability at least $1-4i^2 2^{-n/3}$. Our strategy is to bound the probability of $\pre^\sigma_{i+1}$ having large displacement conditioned on $\pre^\sigma_i(t)$ being desirable, and then bound the probability of $\pre^\sigma_{i+1}$ being colliding conditioned on $\pre^\sigma_i(t)$ being desirable and $\pre^\sigma_{i+1}$ having small displacement. We begin by supposing that $\pre^\sigma_i(t)$ is desirable. 

Since $t$ (and hence $\pre_{i+1}(t)$) does not contain an even-length palindrome, we know that $e^\sigma_{t,i-1} \neq e^\sigma_{t,i}$. By the definition of being non-colliding in \Cref{def:desirable}, we note that $T^\sigma_{t,i+1} \neq T^\sigma_{t,j}$ for any $j \in [i]$; otherwise, we would have $j = i-1$, which would mean that the trees $T^\sigma_{t,i-1}$ and $T^\sigma_{t,i}$ are joined by distinct edges $e^\sigma_{t,i-1}$ and $e^\sigma_{t,i}$, so $\pre^\sigma_i(t)$ would be colliding.
Consider the set $\Delta_i(t) \defeq \{\text{color-preserving~permutation }~\rho:\eta^\rho(\pre^\rho_i(t)) = \eta^\sigma(\pre^\sigma_i(t)) \text{ and } T^\rho_{t,i+1} \neq T^\rho_{t,j} \, \forall j \in [i]\}$. We know, from above, that $\sigma \in \Delta_i(t)$. Moreover, for any $\rho \in \Delta_i(t)$, $\pre^\rho_i(t)$ is undesirable. Since $\sigma$ is drawn from $D_n$, it follows that, conditioned on $\pre^\sigma_i(t)$ being desirable, $\sigma$ is drawn uniformly from $\Delta_i(t)$.

For any permutation $\rho$ and $j \in [\ell^\rho(t)]$, let $v^\rho_j$ be the vertex reached by the $j$th $\weld$-crossing with respect to the path-embedding $\eta^\rho(t)$. Without loss of generality, assume that $v^\sigma_i$ is a red-right vertex.\footnote{The same argument applies for any $c$-right or $c$-left vertex for any $c \in \mathcal{C}$.} Let $u \neq v^\sigma_i$ be any red-right leaf of a tree $T$ in $\mathbb{T}$ such that $T \neq T^\sigma_{t,j}$ for all $j \in [i]$. Let $\rho$ be the color-preserving permutation that is the composition with $\sigma$ of the permutation that maps $u$ to $v^\sigma_i$ (and vice versa) and is identity otherwise. Notice that $v^\rho_i= u$ and $T^\rho_{t,i+1} = T$. Since the path-embeddings $\eta^\sigma(t)$ and $\eta^\rho(t)$ do not encounter vertices $v^\sigma_i$ and $u$ before the $i$th $\weld$-crossing, we have $\eta^\rho(\pre^\rho_i(t)) = \eta^\sigma(\pre^\sigma_i(t))$. Furthermore, as $v^\sigma_i$ and $u$ are not leaves of any tree in $\{T^\sigma_{t,j}: j \in [i]\}$, we have $T^\rho_{t,j} = T^\sigma_{t,j}$ for all $j \in [i]$. Therefore, $T \neq T^\rho_{t,j}$ for all $j \in [i]$. Hence, $\rho \in \Delta_i(t)$. It follows that for all red-right vertices $u$ that are not leaves of any tree in $\{T^\sigma_{t,j}: j \in [i]\}$, the number of permutations $\rho \in \Delta_i(t)$ such that $v^\rho_i = u$ are the same. 
Thus, the probability, over the choice of $\sigma$, that $v^\sigma_i = u$ is the same for all red-right vertices $u$ that are not leaves of any tree in $\{T^\sigma_{t,j}: j \in [i]\}$.

Let $\suc^\sigma_i(t)$ be the sequence of $n/3$ colors beginning from $v^\sigma_i$ and reaching a vertex in column $2n/3$ of $T_L$ (if $v^\sigma_i$ is a leaf in $T_L$) or $T_R$ (if $v^\sigma_i$ is a leaf in $T_R$). Note that $\suc^\sigma_i(t)$ depends on the choice of $\sigma$ as two distinct red-right vertices (for instance) can have distinct color sequences that lead to their respective ancestors in column $2n/3$ of $T_R$.

Let $\sub^\sigma_i(t)$ denote the largest suffix of $\pre^\sigma_{i+1}(t)$ such that the path-embedding $\eta^\sigma(\sub_i(t))$ does not encounter the edge $e^\sigma_{t,i}$ (if it exists). This means that $\sub^\sigma_i(t)$ is a sequence that begins with the color of the edge encountered by the path-embedding $\eta^\sigma(t)$ just after the $i$th $\weld$-crossing and ends with the last color of $t$ if $i =  \ell^\sigma(t)-1$ and with the color of the edge encountered by the path-embedding $\eta^\sigma(t)$ just before the $(i+1)$st $\weld$-crossing otherwise. Let $\ell_i(t)$ denote the length of $\sub^\sigma_i(t)$. Note that, since $\sigma \in \Delta_i(t)$, $\sub^\sigma_i(t)$ does not encounter any $\weld$ edge after the $i$th $\weld$-crossing and the non-$\weld$ edges remain invariant under $\sigma$, so we write it as $\sub_i(t)$.

Let $\sub_i(t, n/3)$ denote the length-$n/3$ prefix of $\sub_i(t)$. The sequence $\sub_i(t, n/3)$ does not exist if $|\pre^\sigma_{i+1}(t)| < |\pre^\sigma_i(t)| + n/3$. But in that case, we know that $\sub_i(t)$, and hence $\pre^\sigma_{i+1}(t)$, only encounter vertices that are distance less than $n/3$ away from $v_i$, so $\pre^\sigma_{i+1}(t)$ has small displacement with certainty. 

Now, consider the case when $\sub_i(t, n/3)$ exists. Note that, since $\pre^\sigma_i(t)$ has small displacement by the induction hypothesis, $\pre^\sigma_{i+1}(t)$ has small displacement if and only if the sequence $\sub_i(t)$ beginning from $v^\sigma_i$ does not encounter any vertex that is distance at least $n/3$ away from $v^\sigma_i$. In other words, the probability that $\pre^\sigma_{i+1}(t)$ has large displacement is equal to the probability, over the choice of $\sigma$, of $\suc^\sigma_i(t) = \sub_i(t, n/3)$, which we compute as follows.

We argued above that, over the choice of $\sigma$, $v^\sigma_i$ is chosen uniformly at random from the red-right leaves of the trees not in $\{T^\sigma_{t,j}: j \in [i]\}$. This means that the required probability is upper-bounded by the ratio of the number of red-right vertices $v^\sigma_i$ that satisfy $\suc^\sigma_i(t) = \sub_i(t, n/3)$ and the number of those that are not leaves of any tree in $\{T^\sigma_{t,j}: j \in [i]\}$. By \Cref{cor:equalsequences}, the number of red-right leaves satisfying $\suc^\sigma_i(t) = \sub_i(t, n/3)$ is at most $\ceil{2^{2n/3+1}/3}$. On the other hand, by \Cref{lem:equalcolorings}, the total number of red-right vertices that are not leaves of any tree in $\{T^\sigma_{t,j}: j \in [i]\}$ is at least $\floor{2^{n}/3} -i \cdot 2^{n/3}$ as any tree in this set has $2^{n/3}$ leaves. Therefore, assuming that $\pre^\sigma_i(t)$ is desirable, the probability 
of $\pre^\sigma_{i+1}(t)$ having large displacement is
\begin{align} \label{eq:largedisplacement}
    \Pr_\sigma[E_{i+1}(\text{large displacement)} \mid E_{i+1}(\text{desirable)}] &= \Pr_\sigma[\suc^\sigma_i(t) = \sub_i(t, n/3)] \\
    & \leq \frac{\ceil{2^{2n/3+1}/3}} {\floor{2^{n}/3}-i \cdot 2^{n/3}}
    \leq \frac{4}{2^{n/3}}
\end{align}
where the last inequality follows since there can be at most $p(n)$ $\weld$-crossings, so $i < \ell^\sigma(t) \leq \ell \leq p(n)$.

Now, we bound the probability of $\pre^\sigma_{i+1}(t)$ being non-colliding conditioned on $\pre^\sigma_{i}(t)$ being desirable and $\pre^\sigma_{i+1}(t)$ having small displacement. Note that $\pre^\sigma_{i+1}(t)$ having small displacement implies that the path-embedding $\eta^\sigma(\pre^\sigma_{i+1}(t))$ will remain in the tree $T^\sigma_{t,i+1}$ until the $(i+1)$st $\weld$-crossing. Thus, by \Cref{def:desirable}, the above probability is equal to the probability of the tree $T^\sigma_{t,i+1}$ not having an edge with any of the trees in $\{T^\sigma_{t,j}: j \in [i]\}$ other than the edge $e^\sigma_{t,i}$. Our strategy is to bound the probability of any particular leaf of a tree in $\{T^\sigma_{t,j}: j \in [i]\}$ being a neighbor of some leaf of $T^\sigma_{t,i+1}$ not via the edge $e^\sigma_{t,i}$ and then apply the union bound.

Pick any leaf $v$ of any of the trees in $\{T^\sigma_{t,j}: j \in [i]\}$. Let $w$ be any neighbor of $u$ that is a vertex of $\weld^\sigma$. Without loss of generality, assume that $v$ is green-left and $w$ is red-right. First, consider the case when the edge joining $v$ with $w$ appears in the path-embedding $\eta^\sigma(\pre_{i+1}(t))$. We established above that $T^\sigma_{t,i+1} \neq T^\sigma_{t,j}$ for any $j \in [i]$.
Thus, if this edge is not $e^\sigma_{t,i}$, then $w$ cannot be a leaf of $T^\sigma_{t,i+1}$. Therefore, this case does not contribute positively to the required probability.

Now, suppose that this edge does not appear in $\eta^\sigma(\pre_{i+1}(t))$. By essentially the same argument that established that the probability that $v^\sigma_i=u$ is the same for all red-right vertices $u$ that are not leaves of any tree in $\{T^\sigma_{t,j}: j \in [i]\}$, the probability that $w = u$ is the same for all red-right vertices $u$ that are not leaves of any tree in $\{T^\sigma_{t,j}: j \in [i]\}$. This means that the probability of $w$ being a leaf of $T^\sigma_{t,i+1}$ is bounded by the ratio of the number of red-right vertices in $T^\sigma_{t,i+1}$ and the number of vertices that are not leaves of any tree in $\{T^\sigma_{t,j}: j \in [i]\}$. By \Cref{lem:equalcolorings}, the former quantity is at most $\ceil{2^{n/3}/3}$. The latter quantity is at least $\floor{2^{n}/3} - i \cdot 2^{n/3}$ as all trees in $\{T^\sigma_{t,j}: j \in [i]\}$ have $2^{n/3}$ leaves. Note that the total number of vertices that are leaves of any tree in $\{T^\sigma_{t,j}: j \in [i]\}$ is $i \cdot 2^{n/3}$, and each such vertex has $2$ neighbors in $\weld^\sigma$. Hence, by the union bound, the required probability is
\begin{align} \label{eq:colliding}
    \Pr_\sigma[E_{i+1}(\text{colliding)} \mid E_i(\text{desirable}), E_{i+1}(\text{small displacement)}] \leq \frac{2i \cdot 2^{n/3} \cdot \ceil{2^{n/3}/3}} {\floor{2^n/3}-i \cdot 2^{n/3}} \leq \frac{4i}{2^{n/3}}.
\end{align}

Now notice that $\pre^\sigma_{i+1}(t)$ can only be desirable if $\pre^\sigma_i(t)$ is desirable, $\pre^\sigma_{i+1}(t)$ has large displacement, or $\pre^\sigma_{i+1}(t)$ is colliding. Hence, the probability of $\pre^\sigma_j(t)$ being desirable satisfies
\begin{align}
    \Pr_\sigma[E_{i+1}(\text{undesirable)}] &\leq \Pr_\sigma[E_i(\text{undesirable)}] + \Pr_\sigma[E_{i+1}(\text{large displacement)} \mid E_i(\text{desirable})] \nonumber\\
    &\quad + \Pr_\sigma[E_{i+1}(\text{colliding)} \mid E_i(\text{desirable}), E_{i+1}(\text{small displacement)}] \\
    &\leq \frac{4i^2}{2^{n/3}} + \frac{4}{2^{n/3}} + \frac{4i}{2^{n/3}} \leq \frac{4(i+1)^2}{2^{n/3}}
\end{align}
where we used \cref{eq:largedisplacement,eq:colliding} for the second inequality. The lemma follows.
\end{proof}

The next corollary is sufficient to establish the classical hardness result of \Cref{thm:classical3-colohardnessmain}, even though it is a weaker statement about a special case of \Cref{lem:3-coloringmain}. Concretely, we show that for a fixed sequence of colors and a uniformly random color-preserving permutation $\sigma$, it is improbable for the corresponding path-embedding to contain the $\exit$ or a path that forms a cycle in $\mathcal{G}^\sigma$.

\begin{corollary} \label{cor:3-coloringhardnessmain}
Let $\ell \in [p(n)]$ and $t \in \mathcal{C}^{\ell}$ such that $t$ does not contain an even-length palindrome. Choose the permutation $\sigma$ according to the distribution $D_n$. Then the probability that the path-embedding $\eta^\sigma(t)$ encounters the $\exit$ or a cycle in $\mathcal{G}^\sigma$ is at most $4p(n)^2 \cdot 2^{-n/3}$.
\end{corollary}

\begin{proof}
By \Cref{def:desirable}, if $\eta^\sigma(t)$ encounters the $\exit$, then $t$ has large displacement, and if $\eta^\sigma(t)$ encounters a cycle in $\mathcal{G}^\sigma$, then $t$ is colliding. That is, $t$ is undesirable if it encounters the $\exit$ or a cycle in $\mathcal{G}^\sigma$. Since $\ell^\sigma_t \leq \ell \leq p(n)$, $t$ is undesirable with probability at most $4(\ell^\sigma_t)^2 2^{-n/3} \leq 4p(n)^2 2^{-n/3}$ over the choice of $\sigma$ by \Cref{lem:3-coloringmain}. 
\end{proof}

We can extend the result of \Cref{cor:3-coloringhardnessmain} about polynomial-length sequences of colors to polynomial-size subtrees of the address tree $\mathcal{T}$ (see \Cref{def:addresstree}). For this purpose, we define the notion of subtree-embedding of subtrees of $\mathcal{T}$. Intuitively, the subtree-embedding of a tree $T$ describes the subgraph of $\mathcal{G}^\sigma$ obtained by querying the oracle $\eta^\sigma$ according to the sequences of colors given by the vertex labels of $T$. 

\begin{definition}[Subtree-embedding] \label{def:treeembedding}
Let $\sigma$ be any color-preserving permutation. Let $\ell \in [p(n)]$ and $t \in \mathcal{C}^{\ell}$. Let $T$ be a subtree of the address tree $\mathcal{T}$ of size $p(n)$ that contains the vertex labeled $\emptyaddress$ but does not contain vertices having labels in $\specialaddresses \setminus \{\emptyaddress\}$. For any vertex of $T$ labeled by $t \neq \emptyaddress$, let $c_{|t|}$ denote the last color appearing in the sequence $t$ and let $\pre(t)$ denote the color sequence formed by removing the last color from $t$. Define the subtree-embedding of $T$ under the oracle $\eta^\sigma$, denoted $\eta^\sigma(T)$, to be a tree isomorphic to $T$ whose vertex labels are in $\valid$ and specified as follows. The vertex of $\eta^\sigma(T)$ corresponding to the vertex of $T$ labeled by $t$ is
\begin{equation}
    \eta^\sigma(T)_t \defeq 
    \begin{cases}
        \entrance & t = \emptyaddress \\
        \eta^\sigma_{c_{|t|}}(\eta^\sigma(T)_{\pre(t)}) & \text{otherwise}.
    \end{cases}
\end{equation}
We say that the subtree-embedding $\eta^\sigma(T)$ \emph{encounters the $\exit$} if it contains a vertex labeled $\exit$ and that $\eta^\sigma(T)$ \emph{encounters a cycle} if it contains two vertices having the same label. 
\end{definition}

\Cref{fig:tree-embedding} illustrates an example of a subtree of $\mathcal{T}$ and the corresponding subtree-embedding in $\mathcal{G}$. For any tree $T$ specified in \Cref{def:treeembedding}, the root of $T$ will always be $\emptyaddress$, so the root of $\eta^\sigma(T)$ will always be $\entrance$. The subtree-embedding $\eta^\sigma(T)$ of a tree $T$ will correspond to the subgraph of $\mathcal{G}^\sigma$ that contain vertices which can be reached by following the addresses given by vertex labels of $T$ in $\mathcal{G}^\sigma$ beginning at the $\entrance$.

\begin{figure}
    \tiny
    \centering
    \begin{subfigure}[b]{0.4\linewidth}
    \begin{tikzpicture}[scale=0.22, auto, node distance=0.2cm, every loop/.style={}, thick, every arrow/.append style={dash,thick}]
    \tikzset{VertexStyle/.style = {shape = ellipse,  minimum size = 20pt, draw}}
    \Vertex[x=0,y=30,L=$\emptyadd$]{emp}
    \Vertex[x=-8,y=26,L=$r$]{r}
    \Vertex[x=8,y=26,L=$b$]{b}
    \Vertex[x=-12,y=22,L=$r\comma g$]{rg}
    \Vertex[x=-4,y=22,L=$r\comma b$]{rb}
    \Vertex[x=4,y=22,L=$b\comma g$]{bg}
    \Vertex[x=12,y=22,L=$b\comma r$]{br}
    \Vertex[x=-8,y=18,L=$r\comma b\comma g$]{rbg}
    \Vertex[x=0,y=18,L=$r\comma b\comma r$]{rbr}
    \Vertex[x=8,y=18,L=$b\comma g\comma b$]{bgb}
    \Vertex[x=-12,y=14,L=$r\comma b\comma g\comma r$]{rbgr}
    \Vertex[x=-4,y=14,L=$r\comma b\comma r\comma b$]{rbrb}
    \Vertex[x=4,y=14,L=$b\comma g\comma b\comma r$]{bgbr}
    \Vertex[x=12,y=14,L=$b\comma g\comma b\comma g$]{bgbg}
    \tikzset{every node/.style={opacity=0,text opacity=1,scale=0.75}}
    \tikzset{EdgeStyle/.style=auto, color=red}
    \Edge[color=red](emp)(r)
    \Edge[color=blue](emp)(b)
    \Edge[color=ForestGreen](r)(rg)
    \Edge[color=blue](r)(rb)
    \Edge[color=ForestGreen](b)(bg)
    \Edge[color=red](b)(br)
    \Edge[color=ForestGreen](rb)(rbg)
    \Edge[color=red](rb)(rbr)
    \Edge[color=blue](bg)(bgb)
    \Edge[color=red](rbg)(rbgr)
    \Edge[color=blue](rbr)(rbrb)
    \Edge[color=red](bgb)(bgbr)
    \Edge[color=ForestGreen](bgb)(bgbg)
    \end{tikzpicture}
    \caption{The subtree $T$ of $\mathcal{T}$} 
    \label{fig:tree-embeddingaddresses} 
    \end{subfigure}
    \hspace{2.5em}
    \begin{subfigure}[b]{0.4\linewidth}
    \begin{tikzpicture}[scale=0.22, auto, node distance=0.2cm, every loop/.style={}, thick, every arrow/.append style={dash,thick}]
    \tikzset{VertexStyle/.style = {shape = ellipse, ,  minimum size = 20pt, draw}}
    \Vertex[x=0,y=30,L=$\entrance$]{emp}
    \Vertex[x=-8,y=26,L=$010110$]{r}
    \Vertex[x=8,y=26,L=$101000$]{b}
    \Vertex[x=-12,y=22,L=$011101$]{rg}
    \Vertex[x=-4,y=22,L=$101001$]{rb}
    \Vertex[x=4,y=22,L=$110100$]{bg}
    \Vertex[x=12,y=22,L=$101010$]{br}
    \Vertex[x=-8,y=18,L=$101111$]{rbg}
    \Vertex[x=0,y=18,L=$111001$]{rbr}
    \Vertex[x=8,y=18,L=$010100$]{bgb}
    \Vertex[x=-12,y=14,L=$100100$]{rbgr}
    \Vertex[x=-4,y=14,L=$110101$]{rbrb}
    \Vertex[x=4,y=14,L=$110101$]{bgbr}
    \Vertex[x=12,y=14,L=$100100$]{bgbg}
    \tikzset{every node/.style={opacity=0,text opacity=1,scale=0.75}}
    \tikzset{EdgeStyle/.style=auto, color=red}
    \Edge[color=red](emp)(r)
    \Edge[color=blue](emp)(b)
    \Edge[color=ForestGreen](r)(rg)
    \Edge[color=blue](r)(rb)
    \Edge[color=ForestGreen](b)(bg)
    \Edge[color=red](b)(br)
    \Edge[color=ForestGreen](rb)(rbg)
    \Edge[color=red](rb)(rbr)
    \Edge[color=blue](bg)(bgb)
    \Edge[color=red](rbg)(rbgr)
    \Edge[color=blue](rbr)(rbrb)
    \Edge[color=red](bgb)(bgbr)
    \Edge[color=ForestGreen](bgb)(bgbg)
    \end{tikzpicture}
    \caption{The subtree-embedding $\eta^\sigma(T)$ of $T$ in $\mathcal{G}$} 
    \label{fig:tree-embeddingvertices} 
    \end{subfigure}
    \caption{Example of a subtree-embedding for the graph $\mathcal{G}$ in \Cref{fig:weldedtreecoloredlabeled} and the identity permutation $\sigma$.}
    \label{fig:tree-embedding}
\end{figure}
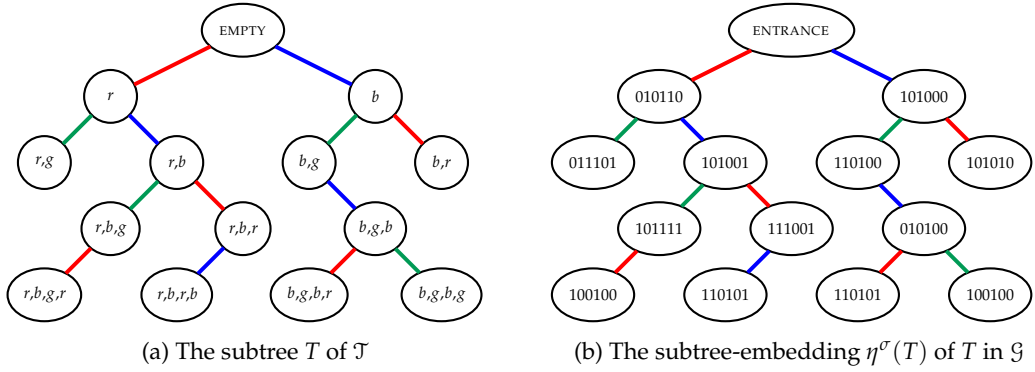

Next, we show that for a fixed sub-tree of $\mathcal{T}$ and a randomly chosen color-preserving permutation $\sigma$, it is not possible for the corresponding subtree-embedding to contain the $\exit$ or a path that forms a cycle in $\mathcal{G}^\sigma$, except with exponentially small probability.

\begin{lemma} \label{lem:subtreeembedding}
Let $T$ be a subtree of the address tree $\mathcal{T}$ of size $p(n)$ that contains the vertex labeled $\emptyaddress$ but does not contain vertices having labels in $\specialaddresses \setminus \{\emptyaddress\}$. Let the permutation $\sigma$ be chosen according to the distribution $D_n$. Then the probability that the subtree-embedding $\eta^\sigma(T)$ encounters the $\exit$ or a cycle is at most $4p(n)^4 2^{-n/3}$.
\end{lemma}

\begin{proof}
Suppose that the subtree-embedding $\eta^\sigma(T)$ contains a vertex $v$ labeled $\exit$. Let $t$ denote the label of the vertex of $T$ corresponding to $v$. Then the path-embedding $\eta^\sigma(t)$ encounters the $\exit$. Therefore, since $T$ has at most $p(n)$ vertices, the probability of $\eta^\sigma(T)$ encountering the $\exit$ is at most $p(n) \cdot 4p(n)^2 2^{-n/3} = 4p(n)^3 2^{-n/3}$ by \Cref{cor:3-coloringhardnessmain}.

Now, suppose that the subtree-embedding $\eta^\sigma(T)$ encounters a cycle. That is, it contains two vertices $v_1$ and $v_2$ having the same label. Without loss of generality, assume that the respective parents $u_1$ and $u_2$ of $v_1$ and $v_2$ in $T$ (if they exist) do not have the same labels; otherwise, re-label $v_1$ as $u_1$ and $v_2$ as $u_2$. Let $t_1$ and $t_2$ be the labels of the vertices of $T$ corresponding to $v_1$ and $v_2$, respectively. Let $t_{1,2}$ denote the sequence resulting from the concatenation of $t_1$ with the reverse of $t_2$. By our above assumption, no color can appear consecutively in $t_{1,2}$, so $t_{1,2}$ does not contain an even-length palindrome. Thus, starting from the $\entrance$ and following the sequence of colors given by $t_{1,2}$ in $\mathcal{G}^\sigma$ will result in returning to the $\entrance$ without backtracking. This means that the path-embedding $\eta^\sigma(t_{1, 2})$ encounters a cycle in $\mathcal{G}^\sigma$. Therefore, since there are at most $\binom{p(n)}{2}$ pairs of vertices of $T$, the probability of $\eta^\sigma(T)$ encountering a cycle is at most $\binom{p(n)}{2} \cdot 4p(n)^2 2^{-n/3} \leq 2p(n)^4 2^{-n/3}$ by \Cref{cor:3-coloringhardnessmain}.

We obtain the desired result by union bounding over the probabilities specified above.
\end{proof}

We now establish that for a uniformly random permutation $\sigma$ and any classical algorithm that samples from the subtrees of $\mathcal{T}$, we cannot hope to find the $\exit$ or a cycle in $\mathcal{G^\sigma}$ with more than exponentially small probability.

\begin{lemma} \label{lem:classical3-colohardnessmain}
Choose the permutation $\sigma$ according to the distribution $D_n$. Let $\mathcal{S}$ be the set of subtrees of the address tree $\mathcal{T}$ of size $p(n)$ that contain the vertex labeled $\emptyaddress$ but do not contain vertices having labels in $\specialaddresses \setminus \{\emptyaddress\}$.\footnote{Recall that $\mathcal{T}$ can be computed using 2 queries to the classical oracle $\eta^\sigma$.} Let $\mathcal{A}_{\mathsf{classical}}$ be a classical query algorithm that samples a tree $T$ from $\mathcal{S}$ and computes the subtree-embedding $\eta^\sigma(T)$. Then, the probability that $\mathcal{A}_{\mathsf{classical}}$ finds the $\exit$ or a cycle in $\mathcal{G^\sigma}$ is at most $4p(n)^4 2^{-n/3}$.
\end{lemma}

\begin{proof}
The algorithm $\mathcal{A}_{\mathsf{classical}}$ finds the $\exit$ if the subtree-embedding $\eta^\sigma(T)$ contains it. On the other hand, since $\eta^\sigma(T)$ corresponds to a connected subgraph of $\mathcal{G}^\sigma$, $\mathcal{A}_{\mathsf{classical}}$ finds a cycle in $\mathcal{G}^\sigma$ if there are two vertices in the tree $\eta^\sigma(T)$ that have the same label. Therefore, this lemma is a direct consequence of \Cref{lem:subtreeembedding} and convexity.
\end{proof}

We conclude this section with our main result about the existence of a distribution for which it is hard for a natural class of classical algorithms to find the $\exit$ or a cycle in the \wtg\ sampled according to this distribution.

\begin{theorem} \label{thm:classical3-colohardnessmain}
Let $\mathcal{S}$ be the set of subtrees of the address tree $\mathcal{T}$ of size $p(n)$ that contain the vertex labeled $\emptyaddress$ but do not contain vertices having labels in $\specialaddresses \setminus \{\emptyaddress\}$. Then there exists a distribution $\mathcal{D}_n$ over size-$n$ 3-colored \wtg s $\mathcal{G}'$ such that for any classical query algorithm $\mathcal{A}_{\mathsf{classical}}$ that samples a tree $T$ from $\mathcal{S}$ and computes the associated subtree-embedding in $\mathcal{G}'$, the probability that $\mathcal{A}_{\mathsf{classical}}$ finds the $\exit$ or a cycle in $\mathcal{G}'$ is at most $4p(n)^4 2^{-n/3}$.  
\end{theorem}

\begin{proof}
Consider the distribution $\mathcal{D}_n$ specified by the following sampling process: choose $\sigma$ according to the distribution $D_n$ and output $\mathcal{G}^\sigma$. From \Cref{lem:classical3-colohardnessmain}, we know that $\mathcal{D}_n$ satisfies the requirement of this theorem.
\end{proof}

\section*{Acknowledgments}

We thank Matt Kovacs-Deak for many helpful discussions.
AMC received support from the Army Research Office (grant W911NF-20-1-0015); the National Science Foundation (grant CCF-1813814); and the Department of Energy, Office of Science, Office of Advanced Scientific Computing Research, Accelerated Research in Quantum Computing program.
MC received support from the National Institute of Standards and Technology (NIST).
ASG received support from the United States Educational Foundation in Pakistan and the Institute of International Education via a Fulbright scholarship.

\newcommand{\etalchar}[1]{$^{#1}$}

\end{document}